\documentclass[sigconf, nonacm]{acmart}

\usepackage[ruled,vlined,linesnumbered]{algorithm2e}
\usepackage{xspace}
\usepackage{amsthm,amssymb,amsmath,epsfig,graphics,enumerate,float,subcaption}
\usepackage{bbm}
\usepackage{epstopdf}
\usepackage{svg}
\usepackage{balance}
\usepackage{enumitem}
\usepackage{color, xcolor}
\usepackage{setspace}
\usepackage{pifont}
\usepackage{tabularx}

\newcommand{\ignore}[1]{}
\newcommand{\nop}[1]{}
\newcommand{\eat}[1]{}
\newcommand{\kw}[1]{{\ensuremath{\mathsf{#1}}}\xspace}

\newcommand{\stitle}[1]{\vspace{1ex} \noindent{\bf #1}}
\long\def\comment#1{}

\newtheorem{example}{Example}

\newcommand{\sota}{\kw{SOTA}}
\newcommand{\kc}{\kw{KC}}
\newcommand{\dfs}{\kw{DFS}}
\newcommand{\bfs}{\kw{BFS}}
\newcommand{\bit}{\kw{BIT}}
\newcommand{\push}{\kw{push}}
\newcommand{\dfn}{\kw{DFN}}

\newcommand{\SpanTree}{\kw{SpanTree}}
\newcommand{\LEwalk}{\kw{LEwalk}}
\newcommand{\ForestMC}{\kw{ForestMC}}
\newcommand{\Approx}{\kw{ApproxKemeny}}
\newcommand{\DynamicMC}{\kw{DynamicMC}}
\newcommand{\RefinedMC}{\kw{RefinedMC}}
\newcommand{\StaticAlg}{\kw{TTF}}
\newcommand{\Basic}{\kw{BSM}}
\newcommand{\Improved}{\kw{ISM}}
\newcommand{\Bipush}{\kw{Bipush}}

\newcommand{\LazyForward}{\kw{LazyForward}}
\newcommand{\TrackingPPR}{\kw{TrackingPPR}}

\SetKwFunction{DFS}{DFS}\SetKwFunction{Wilson}{wilson}\SetKwFunction{Add}{Ins}\SetKwFunction{Del}{Del}
\SetKwFunction{LinkCut}{link-cut}
\SetKwFunction{CutLink}{cut-link}
\SetKwProg{Fn}{Function}{:}{}
\SetKwData{VolSum}{vol\_sum\!}\SetKwData{Next}{next\!}\SetKwData{BFSNext}{fix\_next\!}\SetKwData{Vol}{vol\!}
\SetKwData{Intree}{inTree\!}
\SetKwFunction{Subtree}{subtree}
\newcommand{\powergrid}{\kw{PowerGrid}}
\newcommand{\hepth}{\kw{Hep\textrm{-}\xspace th}}
\newcommand{\astroph}{\kw{Astro\textrm{-}\xspace ph}}
\newcommand{\emailenron}{\kw{Email\textrm{-}\xspace enron}}
\newcommand{\amazon}{\kw{Amazon}}
\newcommand{\dblp}{\kw{DBLP}}
\newcommand{\youtube}{\kw{Youtube}}
\newcommand{\roadPA}{\kw{roadNet\textrm{-}\xspace PA}}
\newcommand{\roadCA}{\kw{roadNet\textrm{-}\xspace CA}}
\newcommand{\orkut}{\kw{Orkut}}
\newcommand{\highschool}{\kw{HighSchool}}
\newcommand{\friendster}{\kw{Friendster}}
\newcommand{\twitter}{\kw{Twitter}}

\newif\iffullversion

\fullversiontrue  

\iffullversion
    \definecolor{RM}{rgb}{0,0,0} 
    \definecolor{R2}{rgb}{0,0,0} 
    \definecolor{R3}{rgb}{0,0,0} 
    \definecolor{R4}{rgb}{0,0,0} 
\else
    \definecolor{RM}{rgb}{0,0,1} 
    \definecolor{R2}{rgb}{0,0,1} 
    \definecolor{R3}{rgb}{0,0,1} 
    \definecolor{R4}{rgb}{0,0,1} 
\fi

\begin{document}

\title{Scalable and Provable Kemeny Constant Computation on Static and Dynamic Graphs: A 2-Forest Sampling Approach}

\author{
    {Cheng Li,
    Meihao Liao,
    Rong-Hua Li,
    Guoren Wang}
}
\affiliation{
	Beijing Institute of Technology
	\country{China}
}
\email{
	lichengbit@bit.edu.cn,
    mhliao@bit.edu.cn,
    lironghuabit@126.com,
    wanggrbit@gmail.com
}

\begin{abstract}
Kemeny constant, defined as the expected hitting time of random walks from a source node to a randomly chosen target node, is a fundamental metric in graph data management with numerous real-world applications. However, exactly computing the Kemeny constant on large graphs is highly challenging, as it requires inverting large graph matrices. Existing solutions primarily focus on approximate methods, such as random walk sampling, which still require large sample sizes and lack strong theoretical guarantees. To overcome these limitations, in this paper, we propose a novel approach for approximating the Kemeny constant via 2-forest sampling. We first present an unbiased estimator of the Kemeny constant in terms of spanning trees, by introducing a \textit{path mapping} technique that establishes a direct correspondence between spanning trees and specific sets of 2-forests. Compared to random walk-based estimators, 2-forest-based estimators yield leads to a better theoretical bound. Next, we design efficient algorithms to sample and traverse spanning trees, leveraging advanced data structures such as the Binary Indexed Tree (\bit) for efficiency optimization. Our theoretical analysis shows that our method computes the Kemeny constant with relative error $\epsilon$ in $O\left(\frac{\Delta^2\Bar{d}^2}{\epsilon^2}(\tau+n\min(\log n, \Delta))\right)$ time, where $\tau$ is the time to sample a spanning tree, $\Bar{d}$ is the average degree, and $\Delta$ is the diameter of the graph. This complexity is near-linear in practical scenarios. Furthermore, most existing methods are designed for static graphs and lack mechanisms for dynamic updates. To address this, we propose two sample maintenance strategies that efficiently update samples partially, while preserving estimation accuracy in dynamic graphs. Extensive experiments on 10 large real-world datasets show that our method outperforms state-of-the-art (\sota) approaches in both efficiency and accuracy, for both static and dynamic graphs.
\end{abstract}

\maketitle

\section{Introduction}

Kemeny constant (\kc), introduced by Kemeny and Snell~\cite{kemeny1969}, is a fundamental metric closely related to hitting time, which is a key quantity associated with random walks. Formally, \kc is defined as the expected steps of a random walk that starts from a fixed initial node and stops until reaching a randomly chosen target node according to the stationary distribution. A notable property of \kc is that it remains invariant regardless the choice of the starting node~\cite{kemeny1969}. Intuitively, a smaller \kc value suggests better overall connectivity of the graph, as nodes are easily reachable from one another. As a unique global network invariant, the Kemeny constant has been applied extensively in various domains of network analysis, including serving as an indicator of network robustness~\cite{patel2015robotic}, analyzing the disease transmission~\cite{app_covid19}, and quantifying global performance costs in road networks~\cite{crisostomi2011google}. In recent years, due to its intrinsic connection to hitting times, \kc has also been widely used in graph data management, such as network centrality representation~\cite{app_centrality, app_markovcentrality}, graph clustering~\cite{app_clustering, ht_clustering_chen2008clustering}, and recommendation systems~\cite{levene2002kemeny, vldb_app_ht_reccomandation_yin2012challenging}, among others.

It is well known that the computation of \kc can be reformulated as an eigenvalue problem involving the Laplacian or transition matrix of a graph~\cite{doyle2009kemeny, hunter2014role, levene2002kemeny, kemeny1969}. However, classical methods for computing eigenvalues require a time complexity of $O(n^3)$, making them impractical for large-scale graphs. To overcome this limitation, various sampling-based approaches have been developed to estimate \kc with improved scalability at the cost of some accuracy. Random walk-based methods, such as \DynamicMC~\cite{li2021efficient} and \RefinedMC~\cite{xia2024efficient}, approximate \kc by simulating truncated random walks. Although they avoid the unacceptable computational cost of eigenvalues, these methods still require a large number of samples to achieve acceptable accuracy. More recent methods, including \LEwalk~\cite{liao2023scalable} and \ForestMC~\cite{xia2024efficient}, leverage the connection between \kc and loop-erased random walks (LERWs) to improve sampling efficiency. While these methods may offer some improvements over standard random walks, they still lack strong theoretical guarantees, because the number of steps of a LERWs can be unbounded in worst case, making it difficult to bound sampling size for all kind of graphs. Consequently, LERW-based methods may fail to provide reliable estimates with great efficiency on real-life large graphs.

Recent studies increasingly focus on exploring the connection between spanning trees or spanning forests with the Laplacian matrix, to address classic random walk-related problems, such as the computation of PageRank \cite{ppr_tree1_liao2022efficient, ppr_tree2_liao2023efficient} and effective resistance~\cite{res4_liao2023efficient, er_liao2024efficient}. Since the scale of a tree or forest can be bounded by the graph diameter, which is more stable on various graphs compared to the length of random walks. This advantage offers the potential for tighter theoretical guarantees about the sample size. Motivated by this insight, we propose a novel formula of \kc expressed by the volume of 2-forests. Specifically, a 2-forest is a spanning forest consisting of exactly two disjoint trees, and the volume of a 2-forest is related to the sum of the degrees, in the underlying graph, of all nodes in one of its two trees (see Theorem~\ref{theo:FuvKC}).

Building on the proposed 2-forest formula of \kc, we develop a theoretical framework and an efficient algorithm to estimate \kc via sampling 2-forests. A key technique of our approach is \textit{path mapping} that maps a spanning tree to a set of 2-forests, enabling us to construct an unbiased estimator for \kc in terms of spanning trees. As a result, we can estimate \kc by sampling spanning trees and then deriving 2-forests from them, effectively bypassing the challenges associated with directly sampling 2-forests. To calculate the volume of 2-forests obtained from each spanning tree, we design a depth-first search (\dfs) algorithm and optimize it using Binary Indexed Trees (\bit)~\cite{Fenwick94}. Compared to the naive spanning tree traversal method used in similar tasks~\cite{res4_liao2023efficient, liao2023scalable}, our approach reduces the time complexity from $O(n\Delta)$ to $O(n \cdot \min(\Delta, \log n))$, where $\Delta$ is the diameter. This significantly improves performance on large-scale graphs. Finally, we provide theoretical guarantees for both the correctness and computational complexity of the proposed algorithm. The results show that our method computes the Kemeny constant with relative error $\epsilon$ in $O\left(\frac{\Delta^2\bar{d}^2}{\epsilon^2}(\phi+n\min(\log n, \Delta))\right)$ time, where $\phi$ is the time to sample a spanning tree, $\bar{d}$ is the average degree, and $\Delta$ is the diameter of the graph. This complexity is near-linear in practical scenarios.

Moreover, most existing methods for estimating \kc are designed for static graphs, while real-world networks are often dynamically changing. To address this gap, we develop two sample maintenance strategies that support efficient updates without requiring full recomputation on dynamic graphs. Benefit from the introduction of spanning trees, we can selectively adjust an amount of spanning tree samples according to the change of the entire sample space caused by updates, thereby indirectly preserving the correctness of the induced 2-forests. Specifically, when an edge is inserted, the basic maintenance method replaces a portion of samples by trees that include the new edge. When an edge is deleted, all spanning trees containing the deleted edge are substituted with newly sampled spanning trees from the updated graph. To further enhance efficiency, we introduce an improved sample maintenance method incorporating \textit{link-cut} and \textit{cut-link} operations. These operations enable the transformation of spanning trees that lack a specific edge into ones that contain it, significantly reducing update overhead. Although the resulting spanning trees may not follow a uniform distribution, we derive a correction mechanism that computes the deviation from uniformity and adjusts the weight of the samples accordingly. This ensures accurate \kc estimation while achieving substantial computational efficiency on dynamic graphs. Both two methods are faster than resampling from scratch. The correctness and time complexity of two methods are discussed.

We conduct extensive experiments to evaluate the performance of the proposed methods. The results show that, for static graphs, our new algorithm outperforms \sota methods, with significant improvements on some graphs. For evolving graphs, the two proposed maintenance strategies achieve significant speed-ups compared to static algorithm recalculations. Specifically, the basic maintenance algorithm provides up to an improvement of more than one order of magnitude, while the improved maintenance algorithm achieves a larger speedup, but with a slight sacrifice in accuracy. In summary, the main contributions of this paper are as follows.

\stitle{Novel Approximate Algorithm.} 
We propose a novel forest formula of Kemeny constant and design an unbiased estimator by introducing a technique called \textit{path mapping}, which establishes a direct connection between spanning trees and 2-forests. Based on this, we propose a sampling-based algorithm for approximating Kemeny constant and optimize the data structure with Binary Index Trees, which enables nearly linear time complexity and makes it easy to extend to dynamic graphs. A detailed theoretical analysis of the proposed algorithm is also provided.

\stitle{New Sample Maintenance Methods.} 
We propose two novel methods, \Basic and \Improved, to maintain the correctness of spanning tree samples as the graph updates, substantially reducing computational cost compared to rerunning static algorithms after each update. \Basic adjusts the spanning tree samples by pre-computing changes in the sample space, while \Improved further enhances efficiency by reusing prior computations.  We provide theoretical analysis showing that, even in the worst case, our algorithms are faster than re-sampling methods, with only a slight sacrifice in accuracy.

\stitle{Extensive Experiments.} 
We conduct comprehensive experiments on 10 real-world graphs to evaluate our algorithms. On static graphs, our method achieves up to an order of magnitude speedup over state-of-the-art algorithms while maintaining the same level of estimation accuracy. For example, on the road network \roadPA, our algorithm reaches a relative error of 0.03 in just 26 seconds, whereas two \sota methods \ForestMC and \SpanTree require 297 seconds and 676 seconds, respectively, to achieve similar accuracy. On dynamic graphs, we further demonstrate the efficiency of our sample maintenance strategies. For the large social networks \orkut, which contains 3 million nodes and 11 million edges, \Improved completes each deletion and insertion update in an average of 1.6 seconds and 4.3 seconds, respectively. In contrast, the fast re-sampling algorithm requires 109 seconds, demonstrating an improvement of approximately one order of magnitude speedup. 

\section{Preliminaries}\label{sec:preliminaries}
\subsection{Notations and Concepts} \label{subsec:notations}

\begin{table}[t!]\color{R2}
	\centering
	\caption{\color{R2} Frequently used notations.} \label{tab:notations}
	\begin{tabularx}{\linewidth}{|>{\raggedright\arraybackslash}p{1.5cm}|>{\raggedright\arraybackslash}X|}
		\hline
		\textbf{Notation} & \textbf{Description} \cr
        \hline
        $G=(V,E)$ & An undirected graph $G$ with node set $V$ and edge set $E$ \cr
        \hline
        $n, m$ & The number of nodes and edges in $G$, respectively \cr
        \hline
        $d(v)$ & The degree of node $v$ \cr
        \hline
        $\kappa(G)$ & The Kemeny constant of the graph $G$ \cr
        \hline
        $T_1\cup T_2$ & A 2-forest consisting of two trees, where $T_1$ and $T_2$ denote the node sets of the two trees\cr
        \hline
        $\tau$ & A spanning tree of the graph \cr
        \hline
        $\Gamma, \mathbb{F}$ & The set of all spanning trees and 2-forests of the graph, respectively \cr
        \hline
        $\mathrm{vol}(T_1)$ & The volume of node set $T_1$, defined as $\mathrm{vol}(T_1)=\sum_{v \in T_1} d(v)$\cr
        \hline
        $\mathcal{P}, \mathcal{P}'$ & A simple, directed path $\mathcal{P}$ and its reverse path $\mathcal{P}'$\cr
        \hline
        $Sub(\tau, u)$ & The node set of the subtree rooted at $u$ in tree $\tau$ \cr
        \hline
        $\epsilon$ & The relative error parameter \cr
        \hline
	\end{tabularx}
\end{table}
Let $G=(V,E)$ be a simple, connected, undirected graph with $n=|V|$ nodes and $m=|E|$ edges. The adjacency matrix $A$ of $G$ is a $n\times n$ matrix whose $(i,j)$ entry is $1$ if and only if edge $(i, j)\in E$ and 0, otherwise. The degree matrix $D$ of $G$ is a diagonal matrix where each entry $D_{ii}=d_i=\sum_{j=1}^{n} A_{ij}$ represents the degree of node $i$.

A random walk on graph $G$ is a stochastic process. At each step, the walker moves to a randomly chosen neighbor of the current node. The transition probability is described by the matrix $P=(p_{ij})=D^{-1}A$, where $p_{ij}$ is the probability of moving from node $i$ to node $j$. It is well known that the stationary distribution $\pi$ of a random walk on undirected graphs is proportional to node degrees, i.e., $\pi_i = d_i / 2m$. Table~\ref{tab:notations} lists the notations that are frequently used in this paper.

The hitting time $H(i, j)$ is a fundamental measure in the study of random walks \cite{lovasz1993random, condamin2007first}, defined as the expected number of steps required to visit the node $j$ for the first time, starting from a node $i$. For an arbitrarily fixed node $i$, the expected hitting time from $i$ to any other node $j$, where $j$ is chosen according to the stationary distribution, is a constant regardless of the choice of $i$. This constant is known as Kemeny constant \cite{kemeny1969}, denoted as $\kappa(G)$, and is given by $\kappa(G) = \sum_{j} \pi_{j} H(i, j)$. 

The Kemeny constant can also be expressed in terms of the pseudo-inverse of the normalized Laplacian matrix. The Laplacian matrix $L$ and the normalized Laplacian matrix $\mathcal{L}$ is defined as $L=D-A$ and $\mathcal{L}=D^{-1/2}LD^{-1/2}=I-D^{-1/2}AD^{-1/2}$, respectively. Let $0=\sigma_1 \leq \cdots \leq \sigma_n$ be the eigenvalues of $\mathcal{L}$, with the corresponding eigenvectors $u_1, \cdots, u_n$. The pseudo-inverse of $\mathcal{L}$ is defined as $\mathcal{L}^{\dagger}=\sum_{i=2}^{n}\frac{1}{\sigma_{i}}u_{i}u_{i}^{\top}$. 
According to \cite{lovasz1993random}, the Kemeny constant can be represented as the trace of the pseudo-inverse of the normalized Laplacian matrix, i.e., 
\begin{equation}\label{eq:KC_eigen}
    \kappa(G)=\mathrm{Tr}(\mathcal{L}^{\dagger})=\sum_{i=2}^{n}\frac{1}{\sigma_i}.
\end{equation}

From Eq. (\ref{eq:KC_eigen}), we see that the Kemeny constant can be computed exactly by finding the eigenvectors of the normalized Laplacian matrix, an $O(n^{3})$ operation. While theoretically sound, this approach becomes prohibitively expensive for large-scale graphs. Moreover, real-world networks often evolve dynamically, making the challenge of efficiently tracking the Kemeny constant over time even more acute. To address these limitations, we aim to develop scalable algorithms to approximate the Kemeny constant on both static and evolving graphs. Below, we provide a concise overview of the \sota methods for approximating \kc on large graphs.

\subsection{Existing \sota Methods and Their Defects} \label{subsec:existingmethods}

In this section, we provide an overview of algorithms for computing \kc, which can be broadly categorized into three classes: matrix-related methods, truncated random walk-based methods, and loop-erased random walk-based methods. For each category, we briefly discuss their underlying principles and limitations.

\stitle{Matrix Related Methods.} As shown in Eq.~(\ref{eq:KC_eigen}), \kc can be expressed by eigenvalues of the normalized Laplacian matrix $\mathcal{L}$. A naive solution is solving eigenvalues in $O(n^3)$ time, which is impractical for large-scale graphs. Xu et al.~\cite{xu2020power} proposed \Approx based on Hutchinson's Monte Carlo \cite{hutchinson1989stochastic}. It approximates the trace of $\mathcal{L}^{\dagger}$, and thus \kc as: $\kappa(G)=\mathrm{Tr}(\mathcal{L}^{\dagger})\approx\frac{1}{M}\sum_{i=1}^{M}x_i^{\top}\mathcal{L}^{\dagger}x_i,$ where $x_i$ are Rademacher random vectors, with each entry independently taking a value of $1$ or $-1$ with equal probability. To compute the quadratic forms of $\mathcal{L}^{\dagger}$, \Approx transforms the problem into solving Laplacian linear systems: $x_i^{\top}\mathcal{L}^{\dagger}x_i = \|BL^{\dagger}y_i\|^{2}$, where $y_i=D^{1/2}(I-\frac{1}{2m}D^{1/2}\mathbf{1}\mathbf{1}^{\top}D^{1/2})$. By using a Laplacian Solver to compute $L^{\dagger}y_i$, \kc can be approximated. However, the efficiency and accuracy of \Approx are limited by the performance of the specific Laplacian solver used.

\stitle{Truncated Random Walk Based Methods.} \DynamicMC\cite{li2021efficient} and \RefinedMC \cite{xia2024efficient} utilize truncated random walks to approximate \kc. This kind of approach takes advantage of the relationship between \kc and the transition matrix $P$, that is, 
\begin{equation}
    \kappa(G)=n-1+\sum_{k=1}^{\infty}[\mathrm{Tr}(P^{k})-1].
\end{equation} 
The $i$-th diagonal entry of $P^{k}$ represents the transition probability from $i$ to itself after $k$ steps, which becomes negligible when $k$ is large. Therefore, by choosing a sufficiently large truncation length $k$ and estimating the diagonal entry of $P^{k}$ using truncated random walks initiated from each node, these methods can achieve a small approximation error. To improve performance, \DynamicMC employs GPU acceleration to parallelize walk simulations, while \RefinedMC improves sampling efficiency by optimizing sample size, truncation length, and the number of starting nodes. However, because of the lack of strong theoretical guarantees, the improvement effect of these optimizations is limited. As a result, these methods has been proven less efficient compared to LERW-based methods as shown in~\cite{xia2024efficient}.

\stitle{Loop-Erased Random Walk (LERW) Based Methods.} Recently, the relationship between $\mathcal{L}^{-1}$ with $(I-P_v)^{-1}$ has been independently explored by \cite{xia2024efficient} and \cite{liao2023scalable} from the perspective of resistance distance and the inverse of the Laplacian submatrix $L_{v}^{-1}$, respectively. Specifically, \kc can be expressed as:
\begin{equation} \label{eq:KC_TrI-P}
\kappa(G)=\mathrm{Tr}(\mathcal{L}^{\dagger})=\mathrm{Tr}(I-P_v)^{-1}+\frac{(\mathcal{L}^{\dagger})_{vv}}{\pi_v}.
\end{equation}

Both methods, \ForestMC proposed in \cite{xia2024efficient} and \LEwalk proposed in \cite{liao2023scalable}, employ loop-erased random walks (LERWs) to approximate $\mathrm{Tr}(I-P_v)^{-1}$, but differ in their computation of $(\mathcal{L}^{\dagger})_{vv}/\pi_v$: \ForestMC~\cite{xia2024efficient} uses truncated random walks, while \LEwalk~\cite{liao2023scalable} utilizes $v$-absorbed random walks. The LERW technique itself is well-established, most notably forming the basis of Wilson's algorithm for uniform random spanning trees sampling \cite{wilson1996generating}.

\iffullversion
The Wilson algorithm constructs a spanning tree through an iterative process: it begins with a single-node tree and sequentially adds LERWs trajectory starting from remaining nodes in arbitrary order. Each LERW terminates upon hitting the current tree, and its acyclic path is incorporated into the tree. Crucially, the expected total length of these LERWs is equal to $\mathrm{Tr}(I-P_v)^{-1}$, allowing an efficient approximation of the first term in Eq.~(\ref{eq:KC_TrI-P}). By leveraging LERWs, we can obtain diagonal-related information of the matrix inverse in nearly linear time, significantly improving efficiency compared to performing $n$ independent random walk from each node. As a result, this technique has been widely adopted in recent studies on single-source problems or tasks involving matrix diagonals~\cite{ppr_tree1_liao2022efficient, angriman2020approximation, res4_liao2023efficient}, offering a practical alternative to traditional random walk based methods.
\else
\fi

However, analyzing the variance of LERW-based methods remains challenging, primarily due to the unbounded length of loop-erased random walks. In the worst case, the number of steps can grow arbitrarily large, causing such methods to fail on certain graphs. Although \ForestMC provides a theoretical guarantee by bounding the absolute error with high probability as the time complexity of algorithm is $O(\epsilon^{-2}\Delta^{2}d_{max}^{2\Delta}\log^{3}n\cdot \mathrm{Tr}(I-P_v)^{-1})$, the term $d_{max}^{2\Delta}$ makes it impractical for most real-world graphs. There is still no good theoretical guarantee for this type of method. Moreover, all aforementioned methods are designed for static graphs and cannot handle dynamic updates efficiently. Even minor modifications to the graph require complete recomputation, significantly limiting their applicability in evolving graph scenarios. 

\begin{figure*}[ht!]
    \includegraphics[width=0.8\linewidth]{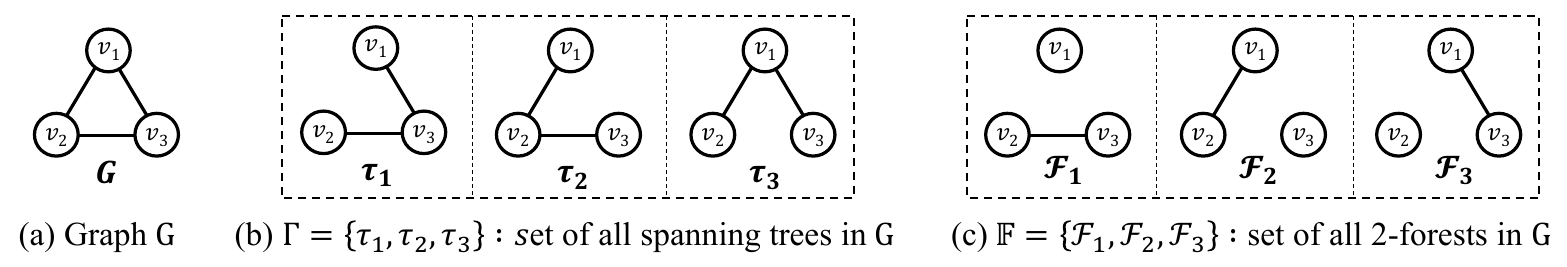}
    \vspace{-0.4cm}
    \caption{An example graph, its spanning trees and 2-forests}
    \Description{An example graph, its spanning trees and 2-forests}
    \label{fig:graph tree forest}
    \vspace{-0.3cm}
\end{figure*}

\section{\kc Estimation Algorithm} \label{sec:static algorithm}

In this section, we propose a novel sampling-based algorithm for approximating the Kemeny constant. We begin by a new formulation of \kc based on 2-forests. Next, we introduce a technique that can maps spanning trees to 2-forests, and we leverage this to design an unbiased estimator for \kc. Based on this framework, we develop a sampling algorithm that estimates \kc via sampling uniform random spanning trees (UST), and further optimize it using a Binary Index Tree and depth-first search. The correctness and complexity analysis for proposed algorithm are also discussed.

\subsection{New Forest Formula of \kc} \label{subsec:NewFormula}

First, we derive a new expression for \kc using 2-forests. Given two fixed nodes $u$ and $v$, we consider the set of 2-forests in which $u$ and $v$ belong to different trees, denoted as $\mathbb{F}_{u|v}$. Utilizing $n-1$ such sets $\mathbb{F}_{u|v}$, we present a new formula for computing \kc.

\begin{theorem}[Forest Formula of \kc]\label{theo:FuvKC}
  \begin{equation} \label{eq:FuvKC}
  \kappa(G) = \frac{1}{2m|\Gamma|}\sum_{u\neq r}d(u)\sum_{T_1\cup T_2 \in \mathbb{F}_{r\mid u} \atop r\in T_1} \mathrm{vol}(T_1),
  \end{equation}
  {\color{RM} where $\Gamma$ denotes the set of all spanning trees of $G$, $r$ is an arbitrarily fixed node, $\mathrm{vol}(T_1) = \sum_{u\in T_1}d(u)$ represents the volume of the tree containing $r$, and $\mathbb{F}_{r\mid u}$ is the set of 2-forests in which $r$ and $u$ belong to different trees.}
\end{theorem}

\iffullversion
    \begin{proof}
      \begin{align*}
        \kappa(G) &= \frac{1}{2m|\Gamma|}\sum_{T_1 \cup T_2 \in \mathbb{F}} \mathrm{vol}(T_1)\mathrm{vol}(T_2) \\
        &= \frac{1}{2m|\Gamma|}\sum_{T_1 \cup T_2 \in \mathbb{F} \atop r\in T_1} \mathrm{vol}(T_1) \sum_{u\in T_2} d(u) \\
        &= \frac{1}{2m|\Gamma|}\sum_{u\neq r}\sum_{T_1\cup T_2 \in \mathbb{F} \atop r\in T_1, u\in T_2} \mathrm{vol}(T_1)d(u) \\
        &= \frac{1}{2m|\Gamma|}\sum_{u\neq r}d(u)\sum_{T_1\cup T_2 \in \mathbb{F}_{r\mid u} \atop r\in T_1} \mathrm{vol}(T_1).
      \end{align*}
    The first equality follows from Corollary 1.7 in~\cite{chung2023forest}, which expresses \kc in terms of all 2-forests, where $\mathbb{F}$ denotes the set of all 2-forests in $G$. Note that for any given 2-forest $T_1\cup T_2$, the designation of which tree is $T_1$ or $T_2$ does not affect anything. Therefore, we deduce the second equation by simply treating the tree containing $r$ as $T_1$, and utilizing the definition of $\mathrm{vol}(T)$.
    \end{proof}
\fi

\begin{example}\label{ex:forest_formula}
    \color{R3}
    As shown in Fig.~\ref{fig:graph tree forest}, (a) illustrates a graph $G$, (b) displays all its spanning trees, and (c) shows all its 2-forests. Assume node $v_1$ is selected as $r$. Then, $\mathbb{F}_{v_2|v_1} = \{\mathcal{F}_1, \mathcal{F}_3\}$ and $\mathbb{F}_{v_3|v_1} = \{F_1, F_2\}$. 
    
    For any given 2-forest, we denote the node set of the tree containing the root $r$ as $T_1$, and the other as $T_2$. For example, in $\mathcal{F}_1$, we have $T_1 = \{v_1\}$ and $T_2 = \{v_2, v_3\}$. Hence, $\mathrm{vol}(T_1) = d(v_1) = 2$. Note that this degree refers to $v_1$'s degree in the graph $G$, not in the forest.
\end{example}

Theorem~\ref{theo:FuvKC} expresses \kc in terms of $n-1$ specific subsets of $\mathbb{F}$, denoted as $\mathbb{F}_{r \mid u}$. To design an unbiased estimator based on this formulation, it is crucial to establish a relationship between spanning trees and $\mathbb{F}_{r \mid u}$, as the normalization factor in Eq.~(\ref{eq:FuvKC}) involves $\Gamma$. To bridge this gap, we introduce a technique called \textit{path mapping}, which constructs a special correspondence between spanning trees and 2-forests by fixing a path. 

An interesting observation about the 2-forests in $\mathbb{F}_{r|u}$ is that they can be obtained from a spanning tree by deleting any edge along the path between $u$ to $v$, thereby ensuring that $u$ and $v$ lie in different components. Consequently, a single spanning tree yields $|P_{u \to v}|$ valid 2-forests, where $|P_{u \to v}|$ is the number of edges on the $u$–$v$ path.

However, the sets of 2-forests generated from different spanning trees overlap. For example, in Fig.~\ref{fig:graph tree forest}, both $\tau_1$ and $\tau_2$ can produce $\mathcal{F}_1$ by deleting one edge. This redundancy prevents a direct correspondence between $\mathbb{F}_{r|u}$ and $\Gamma$. To resolve this issue, we propose a novel key technique called Path Mapping which can systematically avoids such overlaps. The formal definition of Path Mapping is provided below.

\begin{definition}[Path Mapping]\label{def:path mapping}
  Let $\mathcal{P}$ be a simple path from node $u$ to node $r$ in the graph. We define \emph{path mapping} that associates each spanning tree $\tau \in \Gamma$ with a subset of $\mathbb{F}_{r \mid u}$, i.e., $\mathcal{P} : \Gamma \mapsto 2^{\mathbb{F}_{r \mid u}}$, where $2^{S}$ denotes the power set of $S$. For any spanning tree $\tau$, there exists a unique path between $u$ and $r$, denoted by $\mathcal{P}^{(\tau)}_{u \to r}$. By removing all edges shared by both $\mathcal{P}$ and $\mathcal{P}^{(\tau)}_{u \to r}$, we obtain a set of 2-forests. Depending on the relative direction of the overlapping edges, we define two types of path mappings:
  \begin{itemize}[leftmargin=1em]
      \item \textit{Forward Path Mapping}: if the removed edge has same direction in two path, the obtained 2-forest set is denoted as $\mathbb{F}^{P}(\tau)$, i.e. ,
  $$\qquad \mathbb{F}^{\mathcal{P}}(\tau)=\left\{ \tau\setminus (i,j) \ \big|\ (i,j) \in \mathcal{P} \land (i,j) \in \mathcal{P}^{(\tau)}_{u\to r} \right\},$$
      \item \textit{Reverse Path Mapping}: if the removed edge has opposite direction in two path, the obtained 2-forest set is denoted as  $\mathbb{F}^{P'}(\tau)$, i.e. , 
  $$\qquad \mathbb{F}^{\mathcal{P}'}(\tau)=\left\{ \tau\setminus (i,j) \ \big|\ (i,j)) \in \mathcal{P} \land (j,i) \in \mathcal{P}^{(\tau)}_{u\to r} \right\}.$$
  \end{itemize}
\end{definition}

\begin{example}\label{ex:path_mapping}\color{R3}
    Fig.~\ref{fig:path mapping example} illustrates an example of path mapping. To convert a spanning tree into forests in $\mathbb{F}_{v_2|v_1}$, we first identify a simple path from $v_2$ to $v_1$, such as $\mathcal{P}_{v_2}:v_2 \to v_3 \to v_1$. Then, for a given spanning tree, e.g., $\tau_1$, we locate the path from $v_2$ to $v_1$ within $\tau_1$, that is $\mathcal{P}^{(\tau_1)}_{v_2\to v_1}:v_2 \rightarrow v_3 \rightarrow v_1$. We observe that two edges in $\mathcal{P}^{(\tau_1)}_{v_2\to v_1}$ also appear in $\mathcal{P}_{v_2}$ and have the same direction. We can obtain two 2-forests, $\mathcal{F}_3$ and $\mathcal{F}_1$, by removing $(v_2, v_3)$ and $(v_3, v_1)$ respectively. 
    
    As for $\tau_3$, the path from $v_2$ to $v_1$ is $\mathcal{P}^{(\tau_3)}_{v_2\to v_1}:v_2 \rightarrow v_1$. Although the edge $(v_1, v_3)$ also exists in $\mathcal{P}_{v_2}$ and $\tau_3$, it is not part of $\mathcal{P}^{(\tau_3)}_{v_2\to v_1}$, and thus cannot produce a 2-forest. Therefore, applying path mapping with $\mathcal{P}_{v_2}$ to $\tau_3$ results in an empty set.
\end{example}

With Path Mapping, we restrict the deletable edges of a spanning tree to a specific path to reduce redundancy, yet this remains insufficient. To further address this, we distinguish two cases depending on whether the deleted edge aligns with the path orientation, introducing forward and reverse Path Mapping. This refinement further eliminates redundancy and ultimately enables our goal: establishing a precise correspondence between $\Gamma$ and $\mathbb{F}_{r|u}$ through Path Mapping. 

The following theorem formalizes this connection by showing that, given any simple path from $u$ to $r$, the set of 2-forests obtained from all spanning trees in $\Gamma$ via forward Path Mapping, minus those obtained via reverse Path Mapping, is exactly $\mathbb{F}_{r|u}$.

\begin{theorem} \label{theo:PathMappingEqual}
  Let $\mathcal{P}$ be a simple path from $u$ to $r$, then
  $$\mathbb{F}_{r\mid u} = \sum_{\tau \in \Gamma} \left(\mathbb{F}^{\mathcal{P}}(\tau) - \mathbb{F}^{\mathcal{P}'}(\tau)\right).$$
\end{theorem}

\begin{proof}[Proof Sketch]\color{RM}
  The forward path mapping covers all of $\mathbb{F}_{r|u}$, since there must be an edge in $\mathcal{P}$ connecting two trees in any 2-forests. Moreover, for each 2-forest, forward mapping produces exactly one more instance than reverse mapping, as $\mathcal{P}$ contains precisely one additional edge from the tree of $u$ to that of $r$.
\end{proof}

\iffullversion
  \begin{lemma}\label{lemma:PathMappingCupEqual}
    Given an arbitrary simple path $\mathcal{P}$ from $u$ to $r$, the following equation holds:
    $$\bigcup_{\tau \in \Gamma} \mathbb{F}^{\mathcal{P}}(\tau)=\mathbb{F}_{r\mid u}.$$
  \end{lemma}

  \begin{proof}
    The proof idea is to prove that for any 2-forest, there must be a spanning tree that can be mapped to it by a feasible path.
    For each 2-forest $T_1 \cup T_2 \in \mathbb{F}_{r\mid u}$, the node set $V$ is divided into two distinct sets, $T_1$ and $T_2$, and we assume that the tree contains $r$ is $T_1$. Since $P$ starts from $u$ and ends at $r$, there must be at least one edge $e$ in $P$ which leaves from the node in $T_2$ and enters to the node in $T_1$. 
    Otherwise, $u$ and $r$ are in the same one connected component, or path $P$ fails to connect $u$ and $r$. Adding the edge accrossing two components to the 2-forest can get a spanning tree $T_1 \cup T_2 \cup e$, which means this 2-forest $T_1 \cup T_2$ is an element of $\mathbb{F}^{\mathcal{P}}(T_1 \cup T_2 \cup e)$.
  \end{proof}
  
  \begin{lemma}\label{lemma:NumberOfForest}
    For any 2-forest $T_1 \cup T_2 \in \mathbb{F}_{r\mid u}$, we have 
    $$\left|\{ \tau \mid T_1 \cup T_2 \in \mathbb{F}^{\mathcal{P}}(\tau)\}\right| - \left|\{ \tau \mid T_1 \cup T_2 \in \mathbb{F}^{\mathcal{P}'}(\tau)\}\right| = 1. $$
  \end{lemma}

  \begin{proof}
    \begin{figure}
      \centering
      \includegraphics[width=0.4\linewidth]{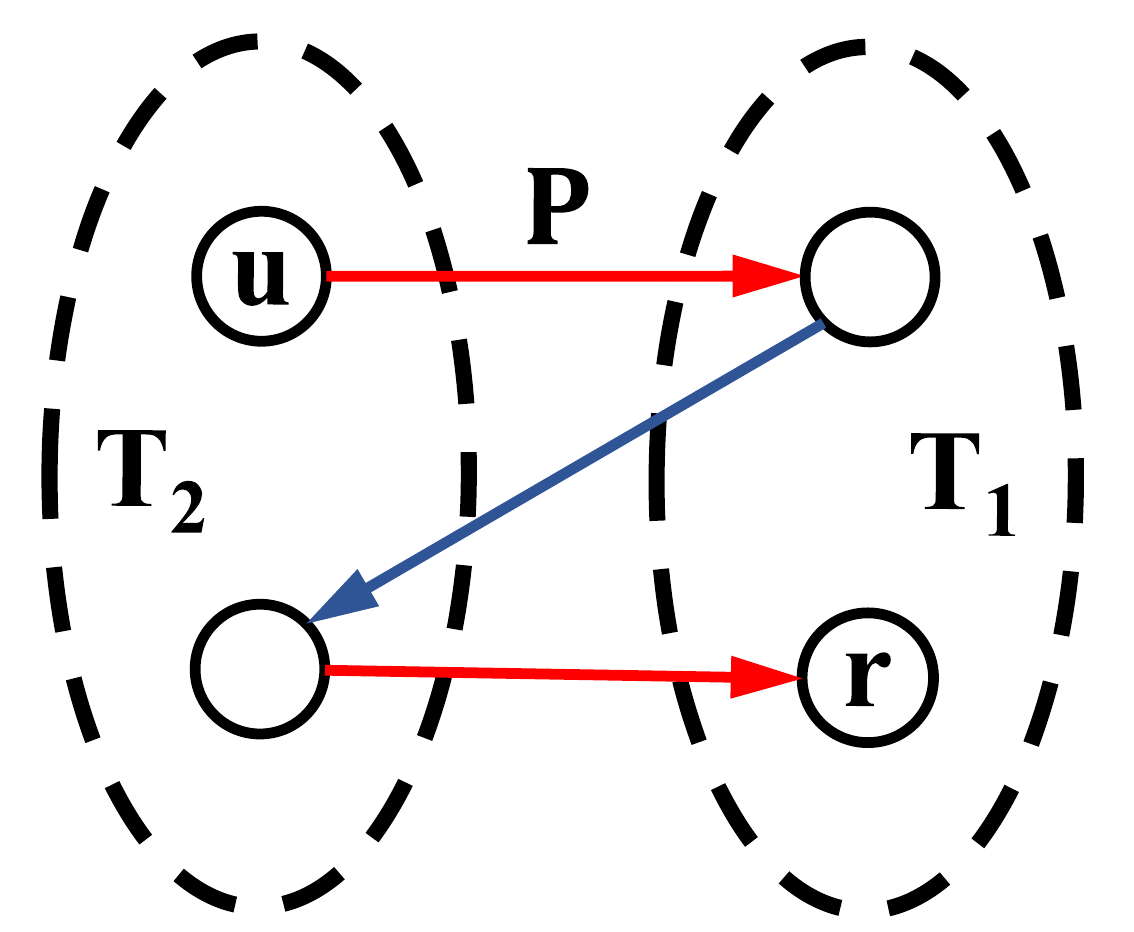}
      \Description{edges in $P$ accrossing through two trees of 2-forest}
      \caption{Proof of Lemma~\ref{lemma:NumberOfForest}}
      \label{fig:path}
    \end{figure}
    for each 2-forest $T_1 \cup T_2 \in \mathbb{F}_{r\mid u}$, we assume $r$ is in $T_1$. As proved in Lemma~\ref{lemma:PathMappingCupEqual}, 
    if a spanning tree $\tau$ is composed of $T_1\cup T_2$ and an edge $e$ in $\mathcal{P}$, which leaves from $T_2$ and enters to $T_1$ (as the red edges in Fig.~\ref{fig:path}). 
    Then $T_1 \cup T_2$ appears in the set obtained by path mapping $\tau$, i.e. , $T_1\cup T_2 \in \mathbb{F}^{\mathcal{P}}(T_1\cup T_2 \cup (e_1, e_2))$, for each $(e_1, e_2)\in \mathcal{P}$, satisfying $e_1 \in T_2, e_2 \in T_1$.
    Similarly $\mathbb{F}^{\mathcal{P}'}(\tau')$ has $T_1\cup T_2$ if $\tau'$ is composed of $T_1\cup T_2$ and the edge that is pointed to $T_2$ from $T_1$ (see the blue edge in Fig.~\ref{fig:path}). Therefore 
    \begin{align*}
      & \left|\{ \tau \mid T_1 \cup T_2 \in \mathbb{F}^{\mathcal{P}}(\tau)\}\right| - \left|\{ \tau \mid T_1 \cup T_2 \in \mathbb{F}^{\mathcal{P}'}(\tau)\}\right| \\
    =& \left|(e_1, e_2) \in \mathcal{P} \mid  e_1 \in T_2 \land e_2 \in T_1 \right| \\ 
      &- \left|(e_1, e_2) \in \mathcal{P} \mid e_1 \in T_1 \land e_2 \in T_2 \right| \\
    =& 1.
    \end{align*}
    The last equality holds because $\mathcal{P}$ is a simple path from $u \in T_2$ to $r \in T_1$, which makes the edges from $T_2$ to $T_1$ exactly one more than the edge from $T_1$ to $T_2$ in $\mathcal{P}$.
  \end{proof}

  \begin{proof}[Proof of Theorem~\ref{theo:PathMappingEqual}]
    According to Lemma~\ref{lemma:PathMappingCupEqual} and Lemma~\ref{lemma:NumberOfForest}, we can know that $\sum_{\tau \in \Gamma}\mathbb{F}^{\mathcal{P}}(\tau)$ includes all elements in $\mathbb{F}_{r|u}$, 
    while $-\sum_{\tau \in \Gamma}\mathbb{F}^{\mathcal{P}'}(\tau)$ eliminates the excess. Therefore, the right term of equation precisely matches $\mathbb{F}_{r\mid u}$.
  \end{proof}
\else 
\fi

\begin{figure}[t!]
    \centering
    \includegraphics[width=\linewidth]{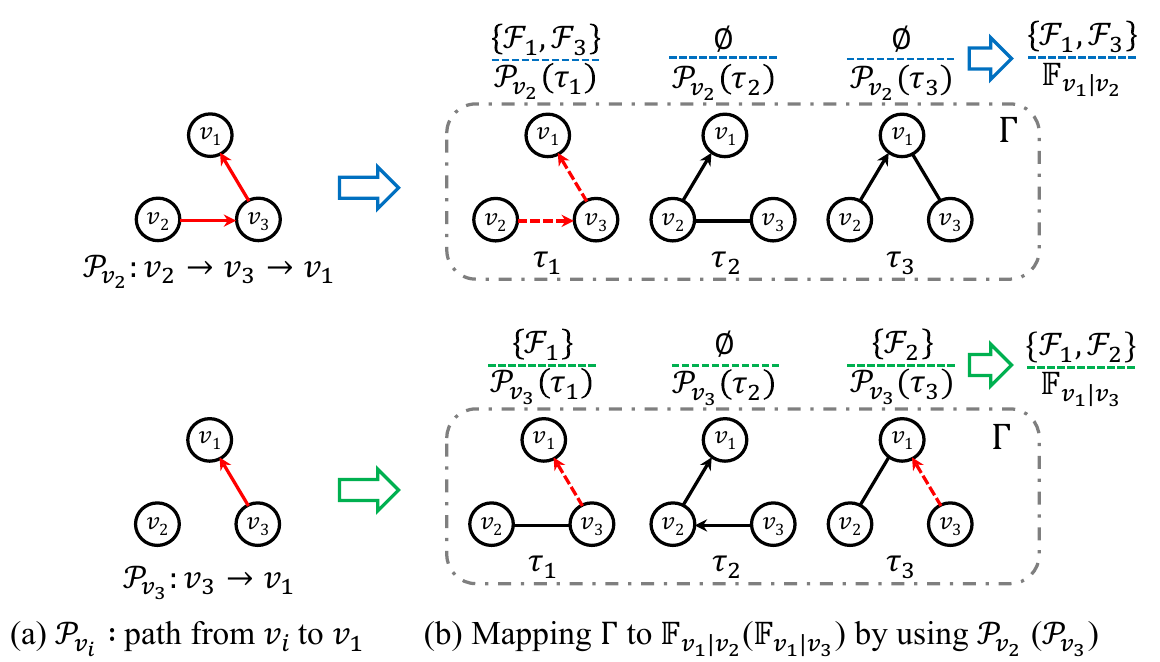}
    \caption{Illustration of Transforming Spanning Trees to 2-Forests Using Path Mapping}
    \label{fig:path mapping example}
    \vspace{-0.1cm}
\end{figure}

\begin{example}\label{ex:path_mapping_equal}
Fig.~\ref{fig:path mapping example} illustrates the mapping from the spanning tree set $\Gamma$ to the 2-forest set $\mathbb{F}_{v|r}$. By selecting a root node and fixing a path from another node to the root, each spanning tree in $\Gamma$ can be mapped to one or more 2-forests in $\mathbb{F}_{v|r}$ by removing edges along the path. For instance, in $\tau_1$, removing edges along the fixed path generates $\mathcal{F}_1$ and $\mathcal{F}_3$, while $\tau_2$ and $\tau_3$ yield no mappings. Combining the mappings from all trees exactly recovers $\mathbb{F}_{v_1\mid v_2}$, ensuring that sampling a tree uniformly from $\Gamma$ induces a uniform distribution over the corresponding 2-forests.
\end{example}

\stitle{Discussion.}  
In this way, we successfully establish a connection between $\Gamma$ and $\mathbb{F}_{r \mid u}$ through path mapping, which lays the foundation for designing a feasible sampling algorithm. Previously, it was difficult to directly sample 2-forests. Now, we can first sample uniform random spanning trees (USTs) and then apply path mapping to transform each spanning tree sample into multiple (or possibly zero) 2-forests. As a result, each sampled 2-forest is obtained with probability $1/|\Gamma|$. In the next section, we will present the sampling algorithm in detail and provide the corresponding theoretical analysis.

\subsection{The Tree-To-Forest Algorithm} \label{subsec:static}

In this section, we introduce \StaticAlg (Tree-To-Forest), a new sampling-based algorithm for estimating the Kemeny constant. The key idea is to first sample USTs and then transform UST samples into 2-forests by path mapping, which allow us to estimate \kc according to Eq.~(\ref{eq:FuvKC}). We begin by outlining the algorithmic intuition and establish its correctness, and then present the pseudo-code, implementation details, and complexity analysis.

\stitle{High-Level Idea of \StaticAlg.} 
\textbf{(1) Identify the root and paths.} We select an arbitrary node $r$ as the root and construct a breadth-first search (BFS) tree rooted at $r$. This BFS tree provides a set of simple paths $\mathcal{P}_u$ from each node $u$ to $r$. 
\textbf{(2) Sample USTs.} We efficiently sample USTs using Wilson's algorithm~\cite{wilson1996generating}, which form the basis of our estimator.
\textbf{(3) From tree to forest.} Each sampled UST is then converted into 2-forests using forward and reverse path mapping with $\{\mathcal{P}_u \mid u\neq r\}$. 
\textbf{(4) Estimate \kc.} Finally, we estimate \kc using Eq.~(\ref{eq:FuvKC}) based on the obtained 2-forests.

The core of our algorithm lies in steps (3) and (4). Importantly, we do not need to explicitly transform sampled USTs into 2-forests. Instead, it suffices to compute the contributions that these 2-forests would make to \kc. Therefore, we merge the two steps and directly estimate \kc from the sampled USTs. We propose two approaches for this task: a naive method and an optimized method.

\stitle{Traverse Each Path (Naive Method). }
For each sampled UST $\tau$, we traverse the path from each node $u$ to the root $r$ within the sampled tree, and compare its edges with the pre-determined path $\mathcal{P}_u$ from step (1). Each matching edge with same orientation indicates a valid 2-forest in $\mathbb{F}^{\mathcal{P}_u}(\tau)$ (or $\mathbb{F}^{\mathcal{P}_u'}(\tau)$ when edge has opposite orientation). For each valid 2-forest, we accumulate the volume of the component containing $r$, i.e. $\mathrm{vol}(T_1)$, and multiply it by $d(u)$ to contribute to the KC estimate. Repeating this process for all nodes $u \neq r$ yields the contribution of $\tau$, denoted by $f(\tau)$

\begin{theorem}[Correctness of Naive Method]\label{theo:naive_correctness}
    Given $\omega$ UST samples and a fixed root $r$ with paths $\{\mathcal{P}_u \mid u\neq r\}$. Let $\hat{\tau}_i$ denote the $i$-th sampled UST. Then, $\tilde{\kappa}=\frac{1}{2m\omega}\sum_{i=1}^{\omega}f(\hat{\tau}_i)$ is an unbiased estimator of $\kappa(G)$, where $$f(\tau) = \sum_{u \neq r} d(u) \left( \sum_{T_1 \cup T_2 \in \mathbb{F}^{\mathcal{P}_u}(\tau)} \mathrm{vol}(T_1) 
    - \sum_{T_1 \cup T_2 \in \mathbb{F}^{\mathcal{P}_u'}(\tau)} \mathrm{vol}(T_1) \right).$$
\end{theorem}

\iffullversion
    \begin{proof}
    Based on Theorem~\ref{theo:FuvKC} and Theorem~\ref{theo:PathMappingEqual}, we can represent the set $\mathbb{F}_{r \mid u}$ by applying the path mapping technique to the spanning tree set $\Gamma$. With this representation, \kc can be expressed as:
        \begin{align*}
        \kappa(G) 
        &= \frac{1}{2m|\Gamma|}\sum_{u\neq r} d(u) \sum_{T_1 \cup T_2 \in \mathbb{F}_{r \mid u},\, r\in T_1} \mathrm{vol}(T_1) \\
        &= \frac{1}{2m|\Gamma|}\sum_{u\neq r} d(u) \sum_{\tau \in \Gamma} 
           \left( \sum_{T_1 \cup T_2 \in \mathbb{F}^{\mathcal{P}}(\tau) \atop r\in T_1} \mathrm{vol}(T_1) 
                - \sum_{T_1 \cup T_2 \in \mathbb{F}^{\mathcal{P}'}(\tau) \atop r\in T_1} \mathrm{vol}(T_1) \right) \\
        &= \frac{1}{2m|\Gamma|}\sum_{\tau \in \Gamma} \sum_{u\neq r} d(u) 
           \left( \sum_{T_1 \cup T_2 \in \mathbb{F}^{\mathcal{P}}(\tau) \atop r\in T_1} \mathrm{vol}(T_1) 
                - \sum_{T_1 \cup T_2 \in \mathbb{F}^{\mathcal{P}'}(\tau) \atop r\in T_1} \mathrm{vol}(T_1) \right) \\
        &= \frac{1}{2m|\Gamma|}\sum_{\tau \in \Gamma} f(\tau).
        \end{align*}
    \end{proof}
\fi

\stitle{Traverse by Depth-First Search (Optimized Method).}  
We further enhance traversal efficiency by maintaining an auxiliary variable \texttt{vol\_sum} for each node and performing a DFS on the sampled UST $\tau$. When visiting a node $v$, we update \texttt{vol\_sum} by adding or subtracting $2m-\mathrm{vol}(\textit{Sub}(\tau, v))$ if the edge connecting $v$ to its parent $p(v)$ in $\tau$ also appears in the BFS tree obtained in step(1), where $\textit{Sub}(\tau, v)$ denotes the subtree rooted at $v$ in $\tau$. After the update, we query the \texttt{vol\_sum} at $v$, multiply it by $d(v)$ to contribute to the KC estimate, recursively continue the DFS on its children, and finally restore the affected \texttt{vol\_sum} values by undoing the update.

\begin{theorem}[Correctness of Optimized Method]\label{theo:optimized_correctness}
    Given an arbitrary spanning tree $\tau$ and a fixed spanning tree $\tau_0$, let $r$ be the root of both trees. Let $\mathcal{P}_u$ denote the path from node $u$ to $r$ in $\tau_0$. Then, the following equality holds:
  \begin{align}
    \sum_{u\neq r}\sum_{T_1\cup T_2 \in \mathbb{F}^{\mathcal{P}_{u}}(\tau)} \mathrm{vol}(T_1) = \sum_{(v, p(v)) \in \tau_0 \atop u \in S_1} (2m-\mathrm{vol}(\textit{Sub}(\tau, v))), 
    \\
    \sum_{u\neq r}\sum_{T_1\cup T_2 \in \mathbb{F}^{\mathcal{P}_{u}'}(\tau)} \mathrm{vol}(T_1) = \sum_{(p(v), v) \in \tau_0 \atop u \in S_2} (2m-\mathrm{vol}(\textit{Sub}(\tau, v))),
  \end{align}
  where $\textit{Sub}(\tau, v)$ is the node set in the subtree rooted at $v$ within $\tau$, and $S_1=\textit{Sub}(\tau, v) \cap \textit{Sub}(\tau_0, v), S_2=\textit{Sub}(\tau, v) \cap \textit{Sub}(\tau_0, p(v))$. $p(v)$ is the parent node of $v$ in $\tau$ and therefore each edge $(v, p(v))$ is directed towards the root $r$ in $\tau$.
\end{theorem}

\iffullversion
    \begin{proof}
        For each node $u$, we apply the forward path mapping $\mathcal{P}_u$, which gives
        \begin{align*}
            \sum_{T_1\cup T_2 \in \mathbb{F}^{\mathcal{P}_u}(\tau)} \mathrm{vol}(T_1)= \sum_{(v, p(v))\in \tau_0 \cap \tau \atop u\in Sub(\tau, v)\cap Sub(\tau_0, v)}\big(2m-\mathrm{vol}(Sub(\tau, v))\big).
        \end{align*}
        By the definition of path mapping, to obtain a valid 2-forest in $\mathbb{F}^{\mathcal{P}_u}(\tau)$, the deleted edge $(v, p(v))$ in $\tau$ must appear on both paths from $u$ to $r$ in the two spanning trees $\tau$ and $\tau_0$. This implies that $u$ must belong to the intersection of the subtrees of $v$ in $\tau$ and $\tau_0$, denoted as $S_1=Sub(\tau_,v)\cap Sub(\tau_0,v)$. Therefore, summing over all $u \neq r$, we obtain
        \begin{align*}
             \sum_{u\neq r}\sum_{T_1\cup T_2 \in \mathbb{F}^{\mathcal{P}_{u}}(\tau)} \mathrm{vol}(T_1) = \sum_{(v, p(v)) \in \tau\cap \tau_0 \atop u \in S_1} (2m-\mathrm{vol}(\textit{Sub}(\tau, v))).
        \end{align*}
        Similar equation holds for the reverse path mapping, $$\sum_{u\neq r}\sum_{T_1\cup T_2 \in \mathbb{F}^{\mathcal{P}_{u}'}(\tau)} \mathrm{vol}(T_1) = \sum_{(p(v), v) \in \tau_0 \atop u \in S_2} (2m-\mathrm{vol}(\textit{Sub}(\tau, v))).$$
    \end{proof}
\fi

\begin{algorithm}[t!]
  \caption{\StaticAlg}
  \small
  \label{alg:static}
  \LinesNumbered
  \KwIn{A graph $G=(V,E)$, the sample size $\omega$, the root node $r$}
  \KwOut{$\tilde{\kappa}$ as the estimation of Kemeny Constant}
  $\tilde{\kappa} \gets 0$\;
  Generate $\tau_0$ using BFS with root $r$\;
  \texttt{DFS\textsubscript{in}}, \texttt{DFS\textsubscript{out}} $\gets \tau_0$\; 
  \texttt{vol\_sum} $\gets \texttt{BIT}(|V|)$\;
  \For{$i\gets 1$ \KwTo $\omega$}{
    $\tau_i \gets$ \Wilson$(G,r)$\;
    \For{each node $u \in V$}{
      \texttt{vol[$u$]} $\gets$ the volume of $Sub(\tau_i,u)$\;
    }
    $\tilde{\kappa} \gets \tilde{\kappa}$ + \DFS($r$, $\tau_i$, $\tau_0$)\;
  }
  \Return $\tilde{\kappa} / (2m\cdot \omega)$;
  
  \Fn{\DFS($v$, $\tau$, $\tau_0$)}{
    Let $p(v)$ be the parent node of $v$ in $\tau$\;
    \If{$(v, p(v)) \in \tau_0$}{
        \texttt{vol\_sum}.$Add$(\texttt{DFS\textsubscript{in}}[$v$], \texttt{DFS\textsubscript{out}}[$v$], $2m-$\texttt{vol[$v$]}) \;
    }
    \ElseIf{$(p(v), v) \in \tau_0$}{
        \texttt{vol\_sum}.$Add$(\texttt{DFS\textsubscript{in}}[$p(v)$], \texttt{DFS\textsubscript{out}}[$p(v)$], \texttt{vol[$v$]}$-2m$) \;
    }
    $\tilde{\kappa} \gets d(v)\cdot$ \texttt{vol\_sum}.$Query(\texttt{DFS\textsubscript{in}}[v])$\;
    \For{each node $i \in Child_{\tau}(v)$}{
        $\tilde{\kappa} \gets \tilde{\kappa}+$ \DFS$(i, \tau, \tau_0)$\;
    }
    \If{$(v, p(v)) \in \tau_0$}{
        \texttt{vol\_sum}.$Add$(\texttt{DFS\textsubscript{in}}[$v$], \texttt{DFS\textsubscript{out}}[$v$], \texttt{vol[$v$]}$-2m$) \;
    }
    \ElseIf{$(p(v), v) \in \tau_0$}{
        \texttt{vol\_sum}.$Add$(\texttt{DFS\textsubscript{in}}[$p(v)$], \texttt{DFS\textsubscript{out}}[$p(v)$], \texttt{vol[$v$]}$-2m$) \;
    }
    \Return{$\tilde{\kappa}$}\;
  }
\end{algorithm}

The pseudo-code of \StaticAlg is presented in Algorithm~\ref{alg:static}. Given a graph $G$, a root node $r$, and a sample size $\omega$, the algorithm begins by constructing a fixed spanning tree $\tau_0$ using breadth-first search (\bfs), followed by computing its \dfn (Lines 2--3). An auxiliary data structure \texttt{vol\_sum} is initialized as a \bit of length $|V|$ (Line 4). It then samples $\omega$ USTs using Wilson algorithm~\cite{wilson1996generating} (Line 6). For each sampled UST $\tau_i$, the algorithm first computes the volume for all subtrees (Lines 7--8), and then invokes the \DFS function to accumulate contributions to the \kc estimator (Line 9). Finally, the algorithm returns $\tilde{\kappa} / (2m\cdot \omega)$ as the approximation of \kc (Line 10).

The core procedure \DFS operates as follows. Let $p(v)$ denote the parent node of the visited node $v$ in $\tau$. The algorithm checks whether the edge $(v, p(v))$ exists in $\tau_0$. If so, it increments \texttt{vol\_sum} for all nodes in the of $\tau_0$ rooted at $v$ by $2m - \texttt{vol}[v]$ (Lines 13--14). Conversely, if the reverse edge $(p(v), v)$ exists in $\tau_0$, it decrements \texttt{vol\_sum} for all nodes in $Sub(\tau_0, p(v))$ by $2m - \texttt{vol}[v]$ (Lines 15--16). After these updates, the algorithm queries \texttt{vol\_sum[$v$]} and incorporates its value into the \kc estimation (Line 17). Then, it recursively invokes \DFS on each child node of $v$ (Lines 18--19). Before returning from \DFS, related changes made to \texttt{vol\_sum} must be revert to maintain correctness for subsequent computations (lines 20-23).  

\begin{figure*}[t!]
  \centering
  \includegraphics[width=0.9\textwidth]{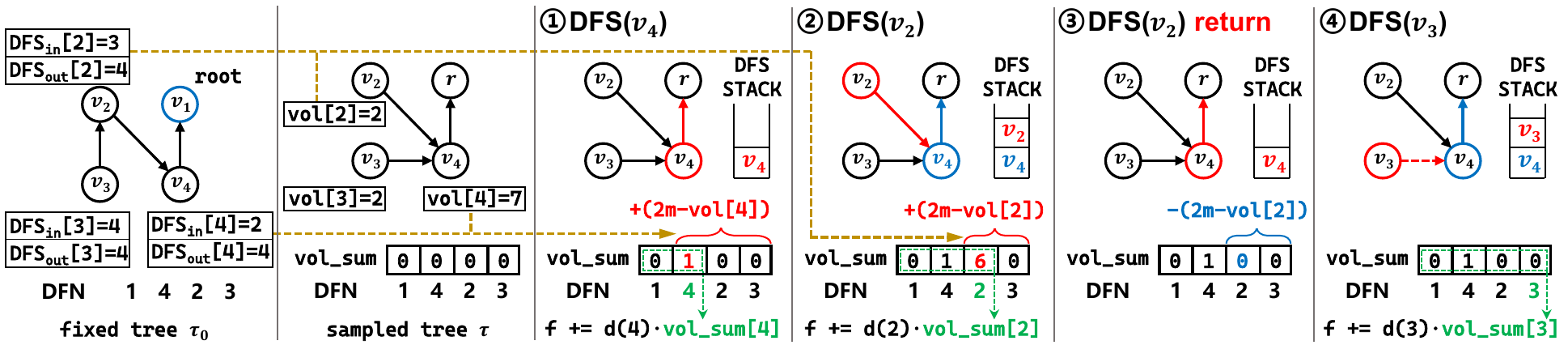}
  \vspace{-0.4cm}
  \caption{Illustration of the DFS process}
  \Description{Illustration of the DFS process}
  \vspace{-0.2cm}
  \label{fig:dfs}
\end{figure*}

\stitle{Implementation Details.} 
We describe the key implementation details of our approach, including the traversal strategy and efficient management of the auxiliary structure \texttt{vol\_sum}.

\textbf{1. Why Use DFS.} Based on Theorem~\ref{theo:optimized_correctness}, in principle, any traversal method can be used to visit all nodes. For each node $v$, if the edge $(v, p(v))$ connecting $v$ to its parent also exists in the fixed spanning tree $\tau_0$, we need to identify the intersection of the subtrees rooted at $v$ in $\tau$ and $\tau_0$, and update the \texttt{vol\_sum} for these nodes by $2m - \mathrm{vol}(\textit{Sub}(\tau, v))$, or decrease it if the edge $(v, p(v))$ has opposite directions in the two trees. Explicitly comparing subtrees to find their intersection is costly, with potential complexity $O(n)$, which motivates the choice of DFS traversal. By using DFS, we can update the \texttt{vol\_sum} values for all nodes in the subtree of $v$ in $\tau_0$ without explicitly identifying intersection nodes. While this may temporarily update some nodes incorrectly, these nodes are not visited during the DFS call for $v$ in $\tau$, and restoring \texttt{vol\_sum} after DFS ensures correctness when each node is visited.

\textbf{2. Implementation of Auxiliary Structure \texttt{vol\_sum}.} To efficiently implement these updates, we combine Depth-First Numbering (DFN) with Binary Index Tree (BIT, or Fenwick Tree~\cite{Fenwick94}).

\textit{2.1 Depth-First Numbering.} DFN records the entry and exit times of each node during DFS traversal, denoted as \texttt{DFS\textsubscript{in}} and \texttt{DFS\textsubscript{out}}. With this notation, the set of nodes in a subtree rooted at $u$ can be precisely characterized by the interval $[\texttt{DFS\textsubscript{in}}[u], \texttt{DFS\textsubscript{out}}[u]]$. Consequently, updating all nodes within a subtree is equivalent to updating the nodes corresponding to a contiguous interval of DFN. For instance, as illustrated in Fig.~\ref{fig:dfs}, the first column shows the entry and exit times of nodes $v_2, v_3,$ and $v_4$. The DFN of the tree $\tau_0$ is $v_1 \rightarrow v_4 \rightarrow v_2 \rightarrow v_3$. The subtree rooted at $v_2$ corresponds to the interval $[3,4]$, which exactly covers the nodes $v_2$ and $v_3$.

\textit{2.2 Binary Index Tree.} For each visited node, two updates (lines 13--16 and lines 20--23) and one query (line 17) are required. Standard arrays allow $O(1)$ queries but $O(n)$ updates, while difference arrays allow $O(1)$ interval updates but $O(n)$ queries. The BIT makes a balance so that when storing a difference array using BIT, it supports interval updates and point queries in $O(\log n)$ time:
\begin{itemize}[leftmargin=2em]
\item \textit{Add$(l, r, v)$}: add value $v$ to all elements in the range $[l, r]$.  
\item \textit{Query$(i)$}: query the $i$-th element value.  
\end{itemize}
This combination of DFS traversal, DFN ordering, and BIT enables efficient and correct maintenance of the \texttt{vol\_sum} auxiliary structure throughout the traversal process.

\begin{example}
    Fig.~\ref{fig:dfs} illustrates how \DFS and \bit efficiently maintain \texttt{vol\_sum} and compute $f(\tau)$ for sampled USTs. The first column shows the fixed tree $\tau_0$ with its \dfn ($1-4-2-3$), the second column displays the sampled tree $\tau$ along with its subtree volumes.

    Beginning the DFS from $v_4$, we examine its outgoing edges $(v_4, v_1)$ and find it also appears in $\tau_0$. Using the \bit, we efficiently update \texttt{vol\_sum} by adding $2m-$\texttt{vol[4]} to all nodes in subtree of $v_4$ in $\tau_0$, using $O(\log n)$ time. After updating \texttt{vol\_sum}, we query \texttt{vol\_sum[$v_4$]}, which corresponds to the second element in \dfn. Since \texttt{vol\_sum} is maintained as a differential array, the actual value is computed as a prefix sum using \bit, also in $O(\log n)$ time.

    The same procedure applies to $v_2$ and $v_3$. For $v_2$, since both edges $(v_2, v_4)$ and $(v_4, v_1)$ appear in both trees and lie on the path from $v_2$ to the root. Thus, \texttt{vol\_sum[2]} is influenced by both the update for $(v_2, v_4)$ and the earlier update for $(v_4, v_1)$. In contrast, the path from $v_3$ to the root in $\tau$ does not include $(v_2, v_4)$. Therefore, before DFS backtracks from $v_2$, we must remove its contribution (third step), to ensure the correctness of query \texttt{vol\_sum[3]} in the fourth step.
\end{example}

\begin{theorem}[Time Complexity of Algorithm~\ref{alg:static}]\label{theo:time_complexity_static}
    The time complexity of Algorithm~\ref{alg:static} is $$O\left(\omega\cdot \left(\mathrm{Tr}(I-P_r)^{-1}+n \min(\Delta, \log n)\right) \right).$$ 
\end{theorem}

\iffullversion
    \begin{proof}
    The time complexity of Wilson Algorithm for generating a UST is $O\left(\mathrm{Tr}\left((I - P_{r})^{-1}\right)\right)$, where $P_r$ is the transition probability matrix with its $r$-th column and row removed. Calculating the volume of each subtree requires $O(n)$ time. 
    For the \texttt{DFS} function, traversing all nodes takes $O(n)$ time. For each node, updating and querying the \texttt{vol\_sum} incurs an additional $O(\log n)$ cost, due to the properties of the Binary Indexed Tree. Consequently, the time complexity of each iteration is $O(n\log n)$. 
    
    However, if the graph has a small diameter $\Delta$ such that $\Delta < \log n$, the computation of \texttt{vol\_sum} can be limited to within $\Delta$ steps from each node to the root, effectively reducing the complexity to $O(n \Delta)$. Therefore, we use $O(n \cdot \min(\Delta, \log n))$ to capture both scenarios.
    Multiplying by the number of samples $T$, the overall time complexity of the algorithm is as stated in the theorem.
    \end{proof}
\fi

\begin{theorem}[Error Bound of Algorithm~\ref{alg:static}] \label{theo:error_bound}
  If the sample size satisfies $\omega \geq \frac{8m^{2}\Delta^{2}_{G}\log (2/p_f)}{n^2\varepsilon^{2}}$, then Algorithm~\ref{alg:static} outputs an estimate $\tilde{\kappa}$ such that $|\tilde{\kappa} - \kappa| \leq \varepsilon \kappa$ with probability at least $1 - p_f$.
\end{theorem}

\iffullversion
    \begin{proof}
      For each sampled tree $\tau$, the number of 2-forests mapped by $\tau_0$ for a node $u$ is bounded by the distance from $u$ to the root $r$. If we construct $\tau_0$ such that its depth is minimized, this distance is at most $\Delta$, the diameter of the graph. Additionally, for each forest, the volume $\mathrm{vol}(T_1)$ is bounded by $2m$. Consequently, $f(\tau)$ can be bounded as $f(\tau) \leq \sum_{u} d(u) \cdot 2m \Delta \leq 4m^{2} \Delta.$
      
      Applying Hoeffding's inequality \cite{hoeffding1994probability}, we derive the following bound on the probability of deviation:
      \begin{align*}
        \mathrm{Pr}\left(\left| \frac{1}{2mT} \sum_{i=1}^{T} f(\tau_i) - \kappa \right| \geq \varepsilon \kappa\right) & \leq \mathrm{Pr}\left(\left| \frac{1}{2mT} \sum_{i=1}^{T} f(\tau_i) - \kappa \right| \geq \varepsilon\cdot n\right) \\
        &\leq 2 \exp\left(-\frac{2\varepsilon^{2}n^2T}{(4m\Delta)^{2}}\right) \\
        &\leq 2 \exp\left(-\frac{2\varepsilon^{2}n^2 \cdot 8m^{2}\Delta^{2}\log(\frac{2}{p_f}) }{16m^{2}\Delta^{2} \cdot \varepsilon^{2}n^{2}}\right) \\ 
        &\leq p_f.
      \end{align*}
      This completes the proof.
    \end{proof}
\fi

By combining Theorem~\ref{theo:time_complexity_static} and Theorem~\ref{theo:error_bound}, we conclude that \StaticAlg can compute \kc in $O\left(\frac{\Delta^2}{\epsilon^2}(\phi+n\min(\log n, \Delta))\right)$ time to achieve a relative error $\epsilon$, where $\phi = \mathrm{Tr}(I-P_r)^{-1}$ is the time required for UST sampling. Our experiments indicates that time complexity is near-linear in practice. Algorithm~\ref{alg:static} details the optimized estimation in $O(n \log n)$ time. For low-diameter graphs, a naive $O(n\Delta)$ implementation can also be used, which matches the computational cost of \SpanTree~\cite{liao2023scalable}. As shown in Section~\ref{sec:experiment}, our method consistently outperforms \SpanTree across all tested graphs, validating the effectiveness of the $O(n \log n)$ design over the $O(n\Delta)$ baseline in real-world datasets.

\stitle{Discussion. } Compared to prior \sota methods, our approach offers stronger theoretical guarantees, particularly over LERW-based methods. Although ForestMC~\cite{xia2024efficient} provides an error bound, it requires $O(\epsilon^{-2}\Delta^{2}d_{max}^{2\Delta}\phi\log^{3}n )$ samples to achieve a relative error $\epsilon$, where the factor $d_{max}^{2\Delta}$ is impractical for most real-world datasets. For instance, the diameter of \dblp used in our experiments is $21$, and its maximum degree is $343$, which results in an astronomical sample size of $O(d_{max}^{2\Delta})=O(343^{42})\approx 10^{106}$, rendering the bound meaningless. The fundamental reason LERW-based methods lack strong theoretical guarantees lies in the high variance of LERWs. Consider a line graph: if a random walk starts at one endpoint and terminates upon reaching the other, the number of steps ranges from at least $n-1$ to potentially infinity. This makes it challenging to establish tight bounds for arbitrary graphs, and performance becomes highly dependent on the graph structure. In contrast, our bound depends on the graph diameter, which is typically small in real-world networks, as demonstrated in our experiments.

We identify two closely related works~\cite{chung2023forest, liao2023scalable}. 
Chung and Zeng~\cite{chung2023forest} represent entries of the Laplacian pseudoinverse and normalized Laplacian pseudoinverse using rooted 2-forests, which further leads to a series of forest-based formulas, including \kc. Our Theorem~\ref{theo:FuvKC} is derived as a direct extension of their Corollary~1.7. Liao et al.~\cite{liao2023scalable} propose the \SpanTree method, which to the best of our knowledge is the only existing approach that also estimates the Kemeny constant via UST sampling, and its computational framework is similar to our \StaticAlg. We next highlight the distinctions and innovations of our approach compared with these two works.
 
\stitle{Comparison to \cite{chung2023forest}.} 
Corollary~1.7 in \cite{chung2023forest} provides a clean and elegant forest formula for \kc, offering a combinatorial interpretation in terms of enumerating 2-forests. However, its practical utility is limited, since uniformly sampling 2-forests is computationally intractable on large-scale graphs. Unlike spanning trees, a 2-forest requires an exact partition of the node set into two disjoint connected components, and a naive strategy that first partitions the nodes and then generates a spanning tree in each part incurs exponential overhead with $O(2^{|V|})$ possible partitions. To the best of our knowledge, no efficient method for direct sampling of 2-forests currently exists, which prevents this formula from being used in practice. This limitation actually motivates our work: by introducing a new formula for \kc (Eq.~(\ref{eq:FuvKC})) together with the path mapping technique, we make the theoretical result of~\cite{chung2023forest} practically computable and enable efficient estimation through sampling.

\stitle{Comparison to \SpanTree \cite{liao2023scalable}.} 
While our method also leverages USTs, similar to \SpanTree proposed in \cite{liao2023scalable}, the foundational principles differ significantly. \SpanTree estimates \kc based on Eq.~(\ref{eq:KC_TrI-P}) via an electrical interpretation, aggregating current flow across sampled USTs. In contrast, our approach is built upon the forest formulas of \kc (Eq.~(\ref{eq:FuvKC})) and introduces a novel path mapping technique to convert trees into 2-forests, which further allows us to incorporate the optimization of \bit to enhance computational efficiency. Experimental results confirm that our method dominates \SpanTree across all tested scenarios.

\section{\kc Computation on Dynamic Graphs} \label{sec:dynamic algorithm}

In this section, we focus on the problem of calculating \kc on dynamic graphs. Our algorithms are improved based on the static algorithm presented in Section~\ref{sec:static algorithm}, where we store needed information for those sampled spanning trees and maintain the correctness of samples to efficiently update \kc. In Section 4.1, we first present a basic maintenance method, and in Section 4.2, we introduce an improved method that more efficiently maintains samples. Theoretical analysis of both algorithms is proved respectively.

\subsection{Basic Samples Maintenance (BSM)}
After a new edge is added or removed from the graph, the spanning trees of the updated graph $G'$ largely overlap with those of the original graph, with the only difference being the additional (or missing) spanning trees that involve the updated edge $e$, assuming both graphs remain connected. Hence, a natural idea is to adapt the sampled USTs by supplementing or removing the spanning trees that differ between $G$ and $G'$. The resulting set of “valid” spanning trees can then be directly used to estimate the \kc.

\stitle{Insertion Case.}  
We first consider the scenario where a new edge $e = (u, v)$ is inserted into the graph. Let $\Gamma$ denote the spanning tree set of the original graph $G$, and let $\Gamma'$ represent the corresponding set of the updated graph $G'$ after inserting the edge $e$. The set $\Gamma'$ comprises all spanning trees in $\Gamma$, along with additional spanning trees that include the newly inserted edge $e$. We define the subset of all these new spanning trees in $\Gamma'$ that contain $e$ as $\Gamma_e'$.  

\begin{figure}[t!]
    \centering
    \includegraphics[width=0.95\linewidth]{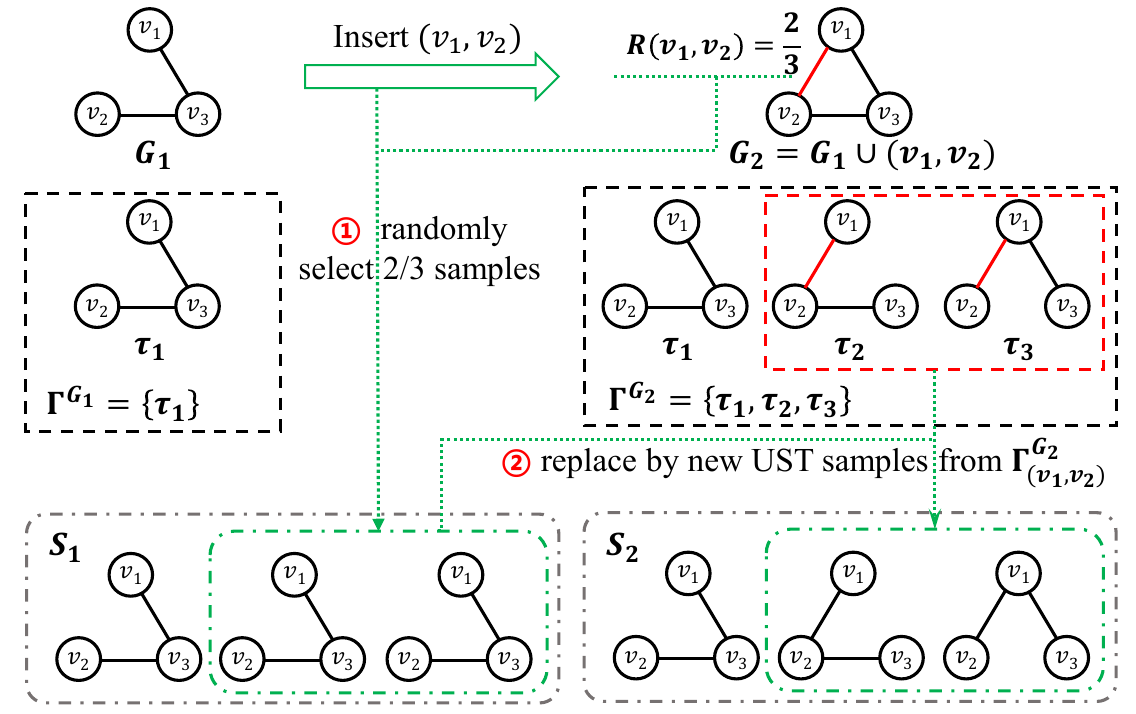}
    \vspace{-0.3cm}
    \caption{Illustration of Basic Samples Maintenance for Edge Insertion}
    \label{fig:edge-ins}
    \vspace{-0.3cm}
\end{figure}

A well-established fact states that the proportion of spanning trees containing a given edge is equal to the effective resistance of that edge, i.e., $|\Gamma_e'|/|\Gamma'| = R(e)$~\cite{sp_tree_hayashi2016efficient, lovasz1993random}. Here, $R(e)$ denotes the effective resistance between nodes $u$ and $v$ for edge $e = (u, v)$ in $G'$. Before the update, we already have some sampled spanning trees that are uniformly distributed in $\Gamma$. Our goal is to incorporate the additional spanning trees in $\Gamma_e'$ that appear after the update. Using the relationship described above, we can determine the expected proportion of these new spanning trees in the samples for $G'$. By replacing part of the old samples with these added USTs, we obtain an unbiased estimation on the updated graph.

\begin{example}
    Fig.~\ref{fig:edge-ins} illustrates the process of maintaining the UST-samples when an edge is inserted. When the edge $(v_1, v_2)$ is added to $G_1$, we first compute its effective resistance in the updated graph $G_2$, which is $R(v_1, v_2) = \frac{2}{3}$. According to the properties of effective resistance, in the new graph, the edge $(v_1, v_2)$ should appear in approximately $\frac{2}{3}$ of the USTs. To update the sample set accordingly, we randomly select $\frac{2}{3}$ of the USTs from the original samples $S_1$ and replace them with the same number of new USTs that include $(v_1, v_2)$. This process is performed using a modified version of Wilson's algorithm. Finally, we utilize the updated samples $S_2$ to estimate the \kc of $G_2$.
\end{example}

\begin{theorem}[Correctness of \Basic for Insertion Case] \label{theo:unbias estimation insert update}
  \begin{align} \label{eq:unbias_add}
    \mathbb{E}_{\tau \sim U(\Gamma')}[f(\tau)] = &(1 - R(e)) \mathbb{E}_{\tau \sim U(\Gamma)} [f(\tau)] \notag \\
    &+ R(e) \mathbb{E}_{\tau \sim U(\Gamma_e')} [f(\tau)],  
  \end{align}
  where $U(\Gamma)$ represents the uniform distribution over the set $\Gamma$, i.e., $\tau$ is sampled with probability $\frac{1}{|\Gamma|}$ for all $\tau \in \Gamma$.   The function $f(\tau)$ corresponds to the calculation performed on the given tree $\tau$. In this paper, it refers to the definition in Theorem~\ref{theo:naive_correctness}.
\end{theorem}

\iffullversion
  \begin{proof}
    \begin{align*}
      \mathop{\mathbb{E}}\limits_{\tau \sim U(\Gamma')}[f(\tau)] =& \frac{1}{|\Gamma'|}\sum_{\tau \in \Gamma'} f(\tau) \\
      =& \frac{1}{|\Gamma'|} \left(\sum_{\tau \in \Gamma}f(\tau) + \sum_{\tau \in \Gamma_{e}'} f(\tau) \right) \\
      =& \frac{1}{|\Gamma'|} \sum_{\tau \in \Gamma}f(\tau) + \frac{1}{|\Gamma'|} \sum_{\tau \in \Gamma_{e}'} f(\tau) \\
      =& \frac{(1-R(e))}{|\Gamma|} \sum_{\tau \in \Gamma}f(\tau) + \frac{R(e)}{|\Gamma_{e}'|}  \sum_{\tau \in \Gamma_{e}'} f(\tau) \\
      =& (1-R(e)) \mathop{\mathbb{E}}\limits_{\tau \sim U(\Gamma)} [f(\tau)] 
      \\ &+ R(e) \mathop{\mathbb{E}}\limits_{\tau \sim U(\Gamma_e)} [f(\tau)].
    \end{align*}
    The second line holds because of $\Gamma \cup \Gamma_{e}' = \Gamma'$ and $\Gamma \cap \Gamma_{e}' = \emptyset$. And the forth line can be deduced by $|\Gamma_e'| / |\Gamma'| = R(e)$.
  \end{proof}
\fi

\stitle{Remarks.}
Note that the updated samples set are not yet followed the uniform distribution $U(\Gamma')$. Nevertheless, Theorem~\ref{theo:unbias estimation insert update} guarantees that the estimator remains unbiased when we uniformly sample from both $\Gamma'$ and $\Gamma_{e}'$ and apply appropriate reweighting. Although the error bound may change, we assume that the sample size does not need adjustment under small updates in real-world datasets. Our experiments in Section~\ref{sec:experiment} confirm that this method effectively preserves estimation accuracy in practice.

\stitle{Deletion Case.} As for deleting an edge $e$, it is more easier to find the way to maintain samples. Only those sampled trees which contains $e$ should not exist in new samples. So we just need to remove all these tree and replace them by regular USTs for new graph. To avoid confusion, we still use $\Gamma$ to represent the case without edge $e$, and $\Gamma'$ to represent the spanning trees of the graph that contains the edge $e$. However, this time we aim to use samples from $U(\Gamma')$ to approximate the results under $U(\Gamma)$.

\begin{figure}[t!]
    \centering
    \includegraphics[width=\linewidth]{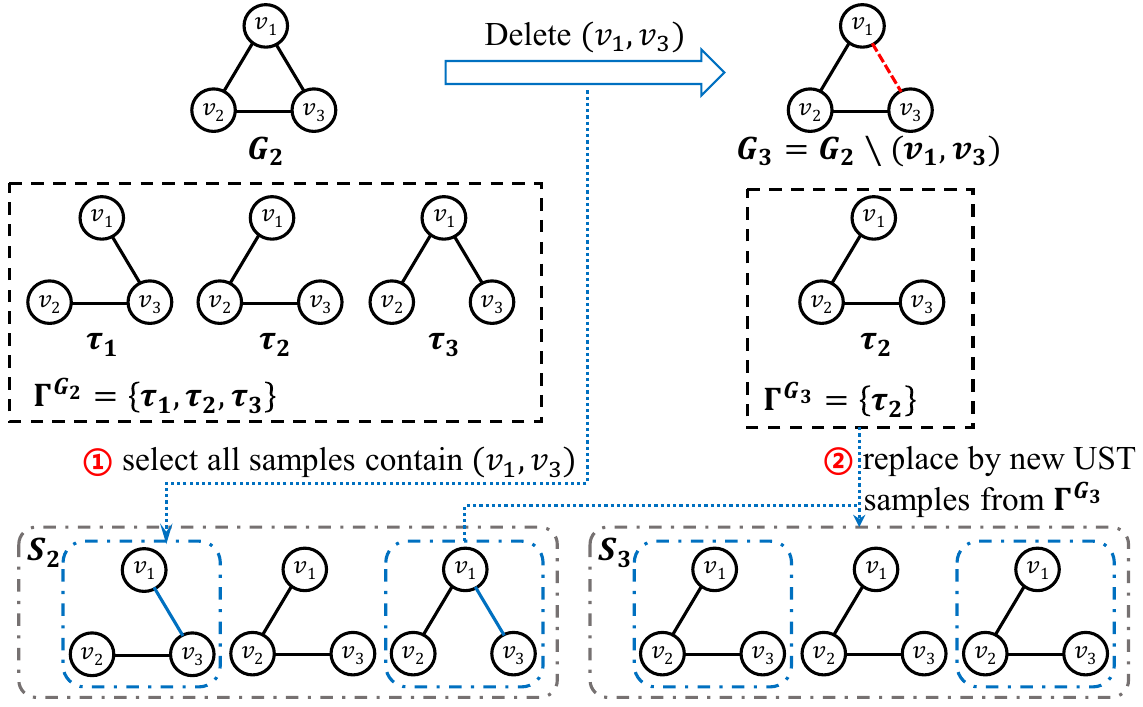}
    \caption{Illustration of BSM for Edge Deletion}
    \label{fig:edge-del}
    \vspace{-0.4cm}
\end{figure}

\begin{example}
    Fig.~\ref{fig:edge-del} illustrates the process of updating the UST-samples when an edge is deleted. When the edge $(v_1, v_3)$ is removed from $G_2$, all spanning trees that contain $(v_1, v_3)$ must be excluded from the updated sample set.  To achieve this, we first identify all USTs in the original sample set $S_2$ that include $(v_1, v_3)$. These trees are then replaced with the same number of newly generated USTs from the updated graph $G_3$, resulting in the new sample set $S_3$. Finally, updated samples in $S_3$ are used to approximate the \kc of $G_3$.
\end{example}

\begin{theorem}[Correctness of \Basic for Deletion Case]\label{theo:unbias estimation delete update}
\begin{equation*}\label{eq:unbias_del}
    \mathbb{E}_{\tau \sim U(\Gamma)}[f(\tau)] = \mathbb{E}_{\tau \sim U(\Gamma')}[f(\tau) \mid e\notin \tau].
\end{equation*}
\end{theorem}

\iffullversion
\begin{proof}
  \begin{align*}
    \mathbb{E}_{\tau \sim U(\Gamma')}[f(\tau) \mid e\notin \tau] &= \sum_{\tau \in \Gamma'}f(\tau) \mathrm{Pr}[\tau \sim U(\Gamma') \mid e \notin \tau] \\
    &= \sum_{\tau \in \Gamma}f(\tau) \frac{\mathrm{Pr}[\tau \sim U(\Gamma') \land \tau \in \Gamma]}{\mathrm{Pr}[e \notin \tau]} \\
    &= \sum_{\tau \in \Gamma}f(\tau) \frac{1}{|\Gamma'|} / \frac{|\Gamma|}{|\Gamma'|} \\
    &= \frac{1}{|\Gamma|} \sum_{\tau \in \Gamma}f(\tau) \\
    &= \mathbb{E}_{\tau \sim U(\Gamma)}[f(\tau)]
  \end{align*}
\end{proof}
\fi

\iffullversion
    \begin{algorithm}[t!] 
      \small
      \caption{\Basic}
      \label{alg:update_simple_resample}
      \LinesNumbered
      \KwIn{A graph $G=(V,E)$, root $r$, an updated edge (u, v), original sample set $S$}
      \KwOut{$\tilde{\kappa}'$ as an updated estimation, updated sample set $S'$}
      $\tilde{\kappa}' \gets \tilde{\kappa}(G)$, $S' \gets S$\;
      \If{Update is \Add$(u, v)$}{
        Calculate $R(u, v)$ in $G \cup (u, v)$\;
        Randomly select $\lceil R(u, v) \cdot |S| \rceil$ trees from $S$\;
        \For{each selected tree $\tau$} {
          $\tilde{\kappa}' \gets \tilde{\kappa}' - f(\tau) / |S|,\quad S' \gets S' \setminus \tau$\;
          $\tau' \gets$ \Wilson$(G', (u, v))$\;
          Redirect edges in $\tau'$ to point towards $r$\;
          $f(\tau’)\gets$\DFS $(r, \tau', \tau_0)$\;
          $\tilde{\kappa}' \gets \tilde{\kappa}' + f(\tau') / |S|,\quad S' \gets S' \cup \tau'$\;
        }
      }
      \If{Update is \Del$(u,v)$}{
        \For{each tree $\tau \in S$ that contains edge $(u, v)$}{
          $\tilde{\kappa}' \gets \tilde{\kappa}' - f(\tau) / |S|,\quad S' \gets S' \setminus \tau$\;
          $\tau' \gets$ \Wilson$(G', r)$\;
          $f(\tau’)\gets$\DFS $(r, \tau', \tau_0)$\;
          $\tilde{\kappa}' \gets \tilde{\kappa}' + f(\tau') / |S|,\quad S' \gets S' \cup \tau'$\;
        }
      }
      \Return $\tilde{\kappa}', S'$;
    \end{algorithm}
    
    The pseudo-code for the basic sample maintenance (\Basic) algorithm, covering both edge insertion and deletion cases, is presented in Algorithm~\ref{alg:update_simple_resample}. For each sampled tree, both $f(\tau)$ and its corresponding edges need to be stored. In the case of edge insertion, we first compute the effective resistance $R(u, v)$ in the updated graph to determine the number of indices that need to be replaced (Lines 3--4). For the computation of single-pair effective resistance, we employ the state-of-the-art method \Bipush~\cite{er_liao2024efficient}. For the selected trees that are to be discarded, their contribution to the estimator $\tilde{K}$ should be subtracted (Line 6). To sample a uniform spanning tree that includes the newly inserted edge $(u, v)$, A modified Wilson algorithm can be applied by setting $u$ and $v$ as the root \cite{wilson2_avena2018random, wilson3_predari2023greedy} and then redirecting all tree edges toward the root node $r$ (Lines 7--8). Finally, we compute $f(\tau')$ for these newly generated trees following the approach in Algorithm~\ref{alg:static} and store them as part of the updated samples (lines 9--10). In deletion case, we just check all samples whether contains $(u, v)$. For those trees with $(u, v)$, we replace them with uniformly spanning trees without $(u, v)$ by using the regular Wilson algorithm on the graph $G'$, and recalculate their contribution for $\tilde{K}$ (Lines 12-16). 
\fi

\begin{theorem}[Correctness and Time Complexity of \Basic] \label{theo:correct complexity for simple update}
    \Basic returns an estimate $\tilde{\kappa}'$ that satisfies the relative-error guarantee. The overall time complexity of \Basic is $$O\left(R(e)\omega\cdot \left(\mathrm{Tr}(I - P_e)^{-1} + n \min(\Delta, \log n)\right) \right),$$ 
where $e=(u,v)$ is the inserted (or deleted) edge, $R(e)$ denotes its effective resistance in the graph that includes $e$, and $P_e$ is the transition matrix with the $u$-th and $v$-th rows and columns removed.
\end{theorem}

\iffullversion
  \begin{proof}
      To establish the correctness of \Basic, we leverage Theorem~\ref{theo:unbias estimation insert update} and Theorem~\ref{theo:unbias estimation delete update} to construct an unbiased estimator for the KC of the updated graph. However, to accelerate the computation, we avoid recalculating $f(\tau)$ for unchanged USTs, even though its value undergoes slight variations due to changes in the degrees of the updated edge’s endpoints. Fortunately, the resulting error in the estimator $\tilde{\kappa}'$ can be bounded by $\Delta$, which is negligible when considering the relative error, as $\kappa$ is larger than $n$ and thus significantly greater than $\Delta$. Experimental results further confirm that the estimation maintains a high level of accuracy.

      Regarding time complexity, only $R(e) \cdot T$ USTs are updated. The time complexity for constructing and querying these USTs remains similar to the static case. The only exception arises in the insertion case, where the modified Wilson algorithm incurs a complexity of $O(\mathrm{Tr}(I - P_{e})^{-1})$ instead of $O(\mathrm{Tr}(I - P_{r})^{-1})$. However, the difference between these two complexity is minimal and remains within the same level. Consequently, the overall time complexity for updating aligns with the statement of the theorem.
  \end{proof}
\fi

\stitle{Discussion.} Compared to Algorithm~\ref{alg:static}, \Basic selectively updates only an $R(e)$ fraction of the samples, and by appropriately reweighting the old and new contributions, the estimator remains unbiased. This strategy substantially reduces the number of samples to be generated and the associated computations. As shown in Table~\ref{tab:datasets}, the effective resistance is typically small (can be below 0.1 in social networks), further lowering the computational cost.

\subsection{Improved Sample Maintenance (ISM)} \label{subsec:improved_sample_maintenance}

Although the basic maintenance method significantly reduces computation compared to static algorithms, replacing a portion of the old samples still results in some loss of information. Can we maintain an unbiased estimation while maximizing the utility of the computations contributed by these replaced trees? To address this question, we proposed Improved Sample Maintenance (\Improved) methods, discussing the insertion and then deletion scenarios in order.

\stitle{Insertion Case.} When considering to transform a tree that does not contain edge $e$ into one that does, what happens if we simply add edge $e$ to the tree? This operation creates a cycle, and by removing any edge other than $e$ from this cycle, a new tree that includes $e$ can be obtained. We refer to this process as a \textit{link-cut} operation. Specifically, \textit{link-cut}($\tau, e$) is defined as linking the edge $e$ and cutting one of the other edges in the cycle created by $e$ and path in tree.

\iffullversion
  \begin{figure}[t!]
      \centering
      \begin{subfigure}{0.44\linewidth}
          \centering
          \includegraphics[width=\textwidth]{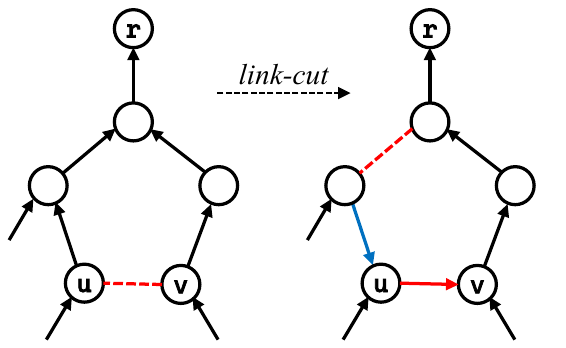}
          \caption{link-cut operation}
          \label{fig:link-cut}
      \end{subfigure}
      \hfill
      \begin{subfigure}{0.44\linewidth}
          \centering
          \includegraphics[width=\textwidth]{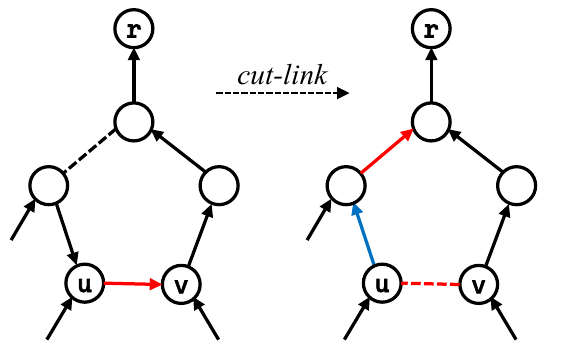}
          \caption{cut-link operation}
          \label{fig:cut-link}
      \end{subfigure}
      \caption{Illustration of link-cut and cut-link operation}
      \vspace{-0.2cm}
      \label{fig:improved sample maintenance}
    \Description{xx}
  \end{figure}
  Using this method to generate spanning trees containing edge $e$ is not only significantly faster than the Wilson algorithm but also preserves parts of the tree's structure, allowing us to recompute only the affected parts. As illustrated in Fig.~\ref{fig:link-cut}, after linking the edge $(u, v)$ and cutting a randomly selected edge in the cycle, all nodes whose paths to the root are altered lie within the subtree of $u$. This enables us to leverage the previous results for unaffected nodes, further reducing the computational cost.
\else
\fi

\begin{figure}[t!]
  \centering
  \includegraphics[width=0.88\linewidth]{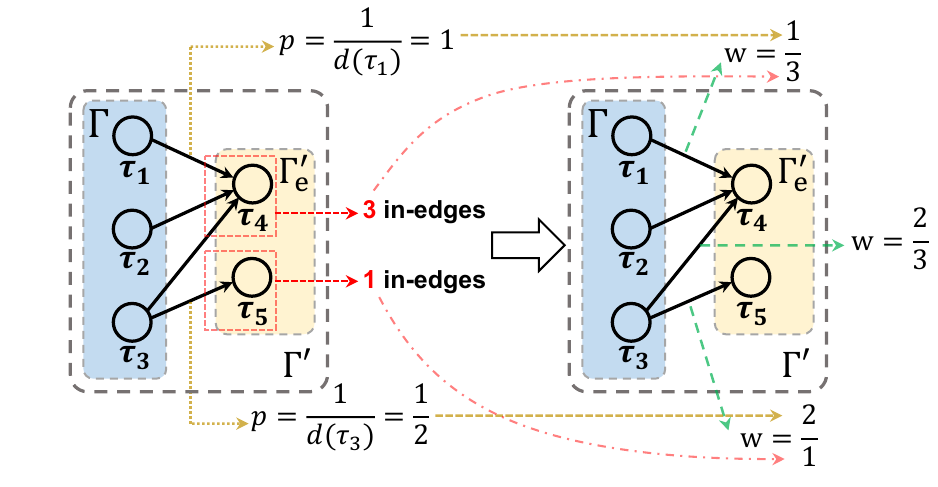}
  \vspace{-0.4cm}
  \caption{Link-Cut Bipartite Graph $\mathcal{B}$}
  \vspace{-0.4cm}
  \label{fig:link_cut_bipartite_graph}
  \Description{xx}
\end{figure}

However, the spanning trees sampled by \textit{link-cut} are not uniformly distributed. To better understand this, consider an abstract graph where each spanning tree is treated as a node, and an edge exists between two nodes if one tree can be transformed into the other through a link-cut operation. This forms a bipartite graph between $\Gamma$ and $\Gamma_{e}'$, denoted as $\mathcal{B}$. Let $\mathcal{N}_{\mathcal{B}}(\tau)$ represents the set of neighbors of a tree $\tau$, corresponding to the spanning trees that can be obtained by performing a link-cut operation on $\tau$. The degree of each tree $\tau \in \Gamma$, denoted as $d(\tau)$, is determined by the length of the cycle formed when edge $e$ is added to $\tau$.

Clearly, if we randomly transform a tree in $\Gamma$ into one of its neighbors with equal probability $1/d(\tau)$ using the link-cut operation, some trees in $\Gamma_{e}'$ will have higher selection probabilities than others. Fortunately, this bias can be corrected by appropriately weighting the tree. As illustrated in Figure~\ref{fig:link_cut_bipartite_graph}, the non-uniformity of link-cut sampling arises primarily due to two factors:  
\begin{itemize}[leftmargin=2em]
\item The out-degree of nodes $\tau \in \Gamma$ are different.
\item The in-degree of nodes $\tau_e \in \Gamma_{e}'$ are different.
\end{itemize}
By appropriately reweighting based on these two factors, we can ensure that the estimator remains unbiased, even though the transformation probabilities are not uniform.

\begin{example}\label{ex:link-cut}
Figure~\ref{fig:link_cut_bipartite_graph} illustrates a link-cut bipartite graph for an evolving graph. Here $\tau_1$ to $\tau_5$ represent all spanning trees in the updated graph. Among them, $\tau_1$ to $\tau_3$ belong to the original tree set $\Gamma$, while $\tau_4$ and $\tau_5$ are new trees introduced by the update. Each edge in the bipartite graph corresponds to a possible transformation between trees via the \textit{link-cut} operation.  

Due to differences in outdegree, selecting an outgoing edge from $\tau_1$ is more likely than from $\tau_3$. Moreover, $\tau_4$ can be reached from $\tau_1, \tau_2, \tau_3$, whereas $\tau_5$ can only be reached from $\tau_3$. If we uniformly select a tree from $\Gamma$ and apply the link-cut operation, the probabilities of obtaining the new trees are $p(\tau_4)=\frac{5}{6}$ and $p(\tau_5)=\frac{1}{6}$, which are biased.

By applying appropriate reweighting, we can correct for this bias to preserve the expected contributions of each tree: $
\mathbb{E}[\tau_4]=\frac{1}{3}\times\frac{1}{3}+\frac{1}{3}\times\frac{1}{3}+\frac{1}{3}\times\frac{1}{2}\times\frac{2}{3}=\frac{1}{3},
\mathbb{E}[\tau_5]=\frac{1}{3}\times\frac{1}{2}\times2=\frac{1}{3}.$
This ensures that the estimator remains unbiased, even though the transformation probabilities are not uniform.
\end{example}

\begin{theorem} \label{theo:improved insertion estimation}
  \begin{equation} \label{eq:link-cut_est}
    \small \mathop{\mathbb{E}}_{\tau \sim U(\Gamma_{e}')}[f(\tau)] = \frac{1-R(e)}{R(e)} \mathop{\mathbb{E}}_{\tau \sim U(\Gamma)}\left[\mathop{\mathbb{E}}_{\tau_e \sim U(\mathcal{N}_{B}(\tau))}\left[f(\tau_e) \frac{d(\tau)}{d(\tau_e)}\right]\right].
  \end{equation}
\end{theorem}

\iffullversion
  \begin{proof}
    \begin{align*}
      & \mathop{\mathbb{E}}_{\tau \sim U(\Gamma)}\left[\mathop{\mathbb{E}}_{\tau_e \sim U(\mathcal{N}_{\mathcal{B}}(\tau))}\left[f(\tau_e) \frac{d(\tau)}{d(\tau_e)}\right]\right] \\
    = & \frac{1}{|\Gamma|} \sum_{\tau\in \Gamma} \mathop{\mathbb{E}}_{\tau_e \sim U(\mathcal{N}_{\mathcal{B}}(\tau))}\left[f(\tau_e) \frac{d(\tau)}{d(\tau_e)}\right] \\
    = & \frac{1}{|\Gamma|} \sum_{\tau\in\Gamma} \left(\frac{1}{d(\tau_e)} \sum_{\tau_e \in \mathcal{N}_{\mathcal{B}}(\tau)} f(\tau_e) \frac{d(\tau)}{d(\tau_e)} \right) \\
    = & \frac{1}{|\Gamma|} \sum_{\tau\in\Gamma} \sum_{\tau_e \in \mathcal{N}_{\mathcal{B}}(\tau)} f(\tau_e) \frac{1}{d(\tau_e)} \\
    = & \frac{1}{|\Gamma|} \sum_{\tau_e \in \Gamma_{e}'} d(\tau_e) (f(\tau_e) \frac{1}{d(\tau_e)}) \\
    = & \frac{R(e)}{1-R(e)} \frac{1}{|\Gamma_{e}'|} \sum_{\tau_e \in \Gamma_{e}'} f(\tau_e) \\
    = & \frac{R(e)}{1-R(e)} \mathop{\mathbb{E}}_{\tau \sim U(\Gamma_{e}')}[f(\tau)]
    \end{align*}
  \end{proof}
\fi

Assigning the weight of trees obtained via the \textit{link-cut} operation as $d(\tau)/d(\tau_e)$ provides an estimate for $f(\tau)$ under the uniform distribution in $\Gamma_{e}'$. By substituting Eq.~(\ref{eq:link-cut_est}) into the second term on the right-hand side of Eq.~(\ref{eq:unbias_add}), we obtain an unbiased estimator of \kc for the updated graph $G'$.

\stitle{Deletion Case.} Similarly, we can transform a tree containing edge $e$ into one without it through the \textit{cut-link} operation. \textit{cut-link}($\tau_e, e$) involves cutting $e$ from a spanning tree, resulting in a 2-forest, and then linking another edge from graph $G$ to reconnect the 2-forest. This operation changes the paths of all nodes within the subtree of the newly linked edge, while other paths remain unchanged. As in the insertion case, we can correct the distribution by assigning weights, with only slight differences in implementation.

\begin{theorem} \label{theo:improved deletion estimation}
  \begin{equation}
    \small \mathop{\mathbb{E}}_{\tau \sim U(\Gamma)}[f(\tau)] = \frac{R(e)}{1-R(e)} \mathop{\mathbb{E}}_{\tau_e \sim U(\Gamma_{e}')}\left[\mathop{\mathbb{E}}_{\tau \sim U(\mathcal{N}_{B}(\tau_e))}\left[f(\tau) \frac{d(\tau_e)}{d(\tau)}\right]\right].
  \end{equation}
\end{theorem}

\begin{proof} [Proof Sketch]\color{RM}
  The proof is similar to that of Theorem~\ref{theo:improved insertion estimation}. By using the definition of expectation and the properties of the bipartite graph $\mathcal{B}$, we can derive the desired result.
\end{proof}

\iffullversion
    The implementation of improved sample maintenance (\Improved) algorithm is shown in Algorithm~\ref{alg:update_link_cut_resample}. The main distinction lies in the fact that the new spanning tree is obtained directly through the link-cut process, which ensures that only a portion of the nodes in the original tree are affected. Updates of $f(\tau')$ are performed exclusively for these affected nodes, and the corresponding weights are computed (Lines 6--8, 19--21). Additionally, after each update, the weights of all trees need to be normalized (Lines 10--13, 23--24). Initially, the weight of each tree is set to $1 / |\omega| $, so there is no need to divide by $|\omega|$ again when calculating \kc in the final step.
    
    \begin{algorithm}[t!] 
      \small
      \caption{\Improved}
      \label{alg:update_link_cut_resample}
      \LinesNumbered
      \KwIn{A graph $G=(V,E)$, root $r$, an updated edge $(u, v)$, original sample set $S$}
      \KwOut{updated estimation $\tilde{\kappa}'$, updated sample set $S'$}
      $\tilde{\kappa}' \gets 0$\;
      \If{Update is \Add$(u, v)$}{
        Compute $R(u, v)$ in $G \cup (u, v)$\;
        Randomly select $\lceil R(u, v) \cdot |S| \rceil$ trees in $S$\;
    
        \For{each selected tree $\tau$} {
          $\tau' \gets$ \LinkCut$(\tau, (u, v))$\;
          Update $f(\tau')$\;
          Calculate weight as $w(\tau')=\frac{d(\tau)}{d(\tau')} w(\tau)$\;    
        }
        \For{each tree $\tau \in S'$} {
          \If{$\tau$ is updated}{$w(\tau)\gets w(\tau)/w_{sum}(\tau_{selected}) \cdot R(u, v)$\;}
          \Else{$w(\tau)\gets w(\tau)/w_{sum}(\tau_{unselected}) \cdot (1-R(u, v))$\;}
          $\tilde{\kappa}'\gets \tilde{\kappa'} + f(\tau)\cdot w(\tau)$\;
        }
      }
      \If{Update is \Del$(u,v)$}{
        $w_{org} \gets 0$ \;
        \For{each tree $\tau \in T$ that contains edge $(u, v)$}{
          $w_{org} \gets w_{org} + w(\tau)$\;
          $\tau' \gets$ \CutLink$(\tau, (u, v))$\;
          Update $f(\tau')$\;
          Calculate weight as $w(\tau')=\frac{d(\tau')}{d(\tau)} w(\tau)$\;    
        }
        \For{each tree $\tau \in S'$} {
          \If{$\tau$ is updated}{$w(\tau)\gets w(\tau)/w_{sum}(T_{updated}) \cdot w_{org}$\;}
          $\tilde{\kappa}'\gets \tilde{\kappa}' + f(\tau)\cdot w(\tau)$\;
        }
      }
      \Return $\tilde{\kappa}', S'$;
    \end{algorithm}
\fi

\begin{theorem}[Correctness and Time Complexity of \Improved] \label{theo:correct complexity for improved update}
    \Improved returns an estimate $\tilde{\kappa}'$ that satisfies the relative-error guarantee. The overall time complexity \Improved is $$O\left(R(e)\omega\cdot (m+ n \min(\Delta, \log n))\right) ,$$ 
where $e$ denotes the inserted (or deleted) edge, and $R(e)$ represents the effective resistance of edge $e$ in the graph that includes $e$.
\end{theorem}

\iffullversion
\begin{proof}
By combining Theorem~\ref{theo:improved insertion estimation} and Theorem~\ref{theo:improved deletion estimation}, we substitute them into Theorem~\ref{theo:unbias estimation insert update} and Theorem~\ref{theo:unbias estimation delete update}, respectively. \Improved produces an estimator for \kc on the updated graph while maintaining a relative-error guarantee, under the assumption that the required sample size remains at the same level even if the maximum value of $f(\tau)$ changes.

For any $\tau \in \Gamma$, its degree $d(\tau)$ in $\mathcal{B}$ can be computed in $O(n)$ time by counting the length of the unique path between the endpoints of the inserted edge $e$. Similarly, for $\tau_e \in \Gamma_{e}'$, its degree $d(\tau_e)$ is determined by the number of trees that can be transformed into $\tau_e$ via a link-cut operation, which is equivalent to the number of trees that can be obtained from $\tau_0$ via a cut-link operation. Consequently, $d(\tau_e)$ equals the number of edges in $G$ that cross between $T_1$ and $T_2$, where $T_1 \cup T_2 = \tau_e \setminus e$. This computation has a time complexity of $O(n)$. Therefore, the overall complexity of updating for each tree is $O(n)$.

The time complexity for updating $f(\tau')$ remains $O(n \log n)$, as its worst case requires recomputing the function for the entire tree. Hence, the overall time complexity of \Improved is $O(R(e)T(m + n \log n))$. In the case of scale-free graphs, where $m = O(n \log n)$, this complexity simplifies to $O(R(e)Tn \log n)$, representing a significant improvement over its static counterpart.
\end{proof}
\fi

\stitle{Discussion.} \Improved further accelerates the UST sampling process compared to \Basic. By leveraging \textit{link-cut} and \textit{cut-link} operations, it can construct the required spanning tree more efficiently than Wilson's algorithm. Although computing the weights introduces an $O(m)$ time complexity, in practice, it is significantly faster than sampling USTs. Furthermore, since the tree structure changes only slightly, previously computed results of $f(\tau)$ can be partially reused, which reduces the overall computational cost. 

\begin{table}[t!]
	\centering
	\caption{Detailed statistics of datasets ($\Delta$ denotes the graph diameter, $\phi$ is the expected time of sampling a UST, and $\bar{R}(e)$ represents the average effective resistance of each edge)} \label{tab:datasets}
	\vspace{-0.2cm}
        \scalebox{0.84}{
	\begin{tabular}{cccccc}
		\toprule
		\bf Dataset & $n$ & $m$ & $\Delta$ & $\phi$ & $\bar{R}(e)$ \cr\midrule
	    \powergrid & 4,941 & 6,594 & 46 & $4.4\times 10^4$ & 0.749 \cr
		\hepth & 8,638 & 24,806 & 17 & $1.7\times 10^4$ & 0.348 \cr
		\astroph & 17,903 & 196,972 & 14 & $2.3\times 10^4$ & 0.091 \cr
		\emailenron & 33,696 & 180,811 & 11 & $5.0\times 10^4$ & 0.186 \cr
		\amazon & 334,863 & 925,872 & 44 & $9.0\times 10^5$ & 0.362 \cr 
		\dblp & 317,080 & 1,049,866 & 21 & $5.9\times 10^5$ & 0.302 \cr
		\youtube & 1,134,890 & 2,987,624 & 20 & $1.8\times 10^6$ & 0.380 \cr
		\roadPA & 1,087,562 & 1,541,514 & 786 & $1.7\times 10^7$ & 0.706 \cr
		\roadCA & 1,957,027 & 2,760,388 & 849 & $3.5\times 10^7$ & 0.709 \cr
		\orkut & 3,072,441 & 117,185,083 & 9 & $3.1\times 10^6$ & 0.026 \cr
		\bottomrule
	\end{tabular} }
	\vspace{-0.4cm}
\end{table}

\section{Experiments} \label{sec:experiment}
\subsection{Experimental Settings}

\stitle{Datasets.} We employ 10 real-world datasets including various type of graphs, primarily focusing on collaboration networks, social networks and road networks. All datasets can be downloaded from \cite{snapnets,nr,socialpattern}. Since Kemeny constant is defined on connected graphs, we focus on the largest connected component for each graph in this study. The detailed statistics of each dataset are summarized in Table~\ref{tab:datasets}. Following previous studies \cite{li2021efficient, xia2024efficient}, we approximate \kc using \Approx~\cite{xu2020power} with $\epsilon=0.15$ as the ground truth for small graphs ($n < 10^{6}$). For larger graphs, \Approx fails to provide results within an acceptable time, so we employ our proposed method \StaticAlg with a large sample size $T$ ($T=10^{5}$) to obtain a sufficiently accurate result as the ground truth.

\stitle{Dynamic Updates.} To simulate dynamic graph updates, we consider both edge insertions and deletions. For the insertion scenario, we first select $90\%$ of the original graph as the base graph and sample 100 edges from the remaining $10\%$ as candidate insertions. For the deletion scenario, we sample 100 edges directly from the original graph as deletions. To better reflect realistic update behavior, we further adopt a power-law update model, where most updates occur on a small fraction of edges. Specifically, an edge $(u,v)$ is chosen for insertion or deletion with probability proportional to the product of the degrees of $u$ and $v$, ensuring that edges incident to high-degree nodes are more likely to be updated.

\stitle{Different Algorithms.} For static graphs, we compare our proposed method \StaticAlg with several representative approaches, including the LERW-based method \ForestMC~\cite{xia2024efficient}, the UST-based method \SpanTree~\cite{liao2023scalable}, the truncated random walks-based methods \DynamicMC~\cite{li2021efficient} and \RefinedMC~\cite{xia2024efficient}, and the matrix-related method \Approx~\cite{xu2020power}. For dynamic graphs, we select the three most competitive static methods, \StaticAlg, \ForestMC, and \SpanTree, and apply them by re-running after each update as baseline methods, comparing with our two proposed sample-maintenance approaches, \Basic and \Improved. 

\stitle{Experimental Environment.} All algorithms used in our experiment are implemented in C++ and compiled with g++ 11.2.0 using the -O3 optimization flag, except for \Approx~\cite{xu2020power}, which is implemented in Julia. We conduct all experiments on a Linux server with a 64-core 2.9GHz AMD Threadripper 3990X CPU and 128GB memory. Each experiment is repeated five times to avoid accidental anomalies.

\begin{figure*}[t!]
  \centering
  \includegraphics[width=\textwidth]{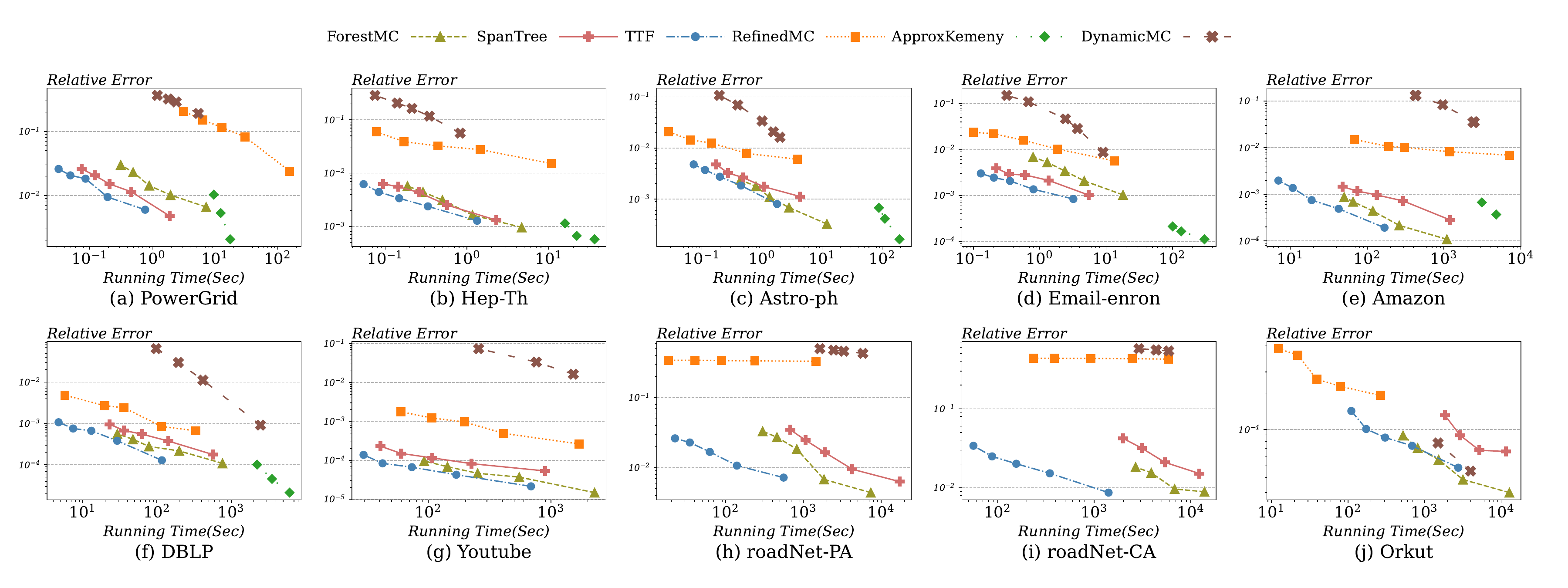}
    \vspace{-0.8cm}
  \caption{\color{RM}Relative error vs running time for different algorithms}
  \label{fig:exp_param_eps}
  \Description{Performance of different algorithms for varying $\epsilon$.}
    \vspace{-0.4cm}
\end{figure*}

\subsection{Experiment Results on Static Graphs}

\stitle{Exp-I: Error vs. Time on Static Graphs.} In this experiment, we evaluate the performance of static algorithms across all datasets by comparing running time against relative error. As shown in Fig.~\ref{fig:exp_param_eps}, \Approx can achieve very accurate results, but its runtime and memory usage are prohibitively high, making it infeasible for large graphs. Methods based on truncated random walks perform reasonably well on some datasets but struggle on road networks, where even extended runtimes fail to achieve high precision. \ForestMC exhibits good accuracy on social networks, but its performance degrades on other types of graphs. In contrast, \StaticAlg consistently achieves lower error in less time across nearly all datasets, demonstrating both high efficiency and accuracy.

\iffullversion
    \stitle{Exp-II: Effect of Path Selection.} In Algorithm~\ref{alg:static}, we use \bfs to construct a spanning tree that defines the set of paths used for path mapping. To investigate the impact of different path selection strategies, we also experiment with \dfs and Wilson’s algorithm. The results are presented in Fig.~\ref{fig:exp_param_path_time} and Fig.~\ref{fig:exp_param_path_error}. We observe that while the choice of path selection method has a negligible impact on running time, it significantly affects the accuracy of the estimation. As discussed in Theorem~\ref{theo:error_bound}, the estimation value is bounded by the maximum length of paths. When the spanning tree is generated via \bfs, the maximum path length is bounded by the diameter of the graph, while it can be substantially longer when using \dfs or randomly sampled trees. Consequently, shorter paths tend to reduce the variance of \StaticAlg, leading to more accurate estimations.

    \begin{figure}[ht!]
        \centering
        \includegraphics[width=\linewidth]{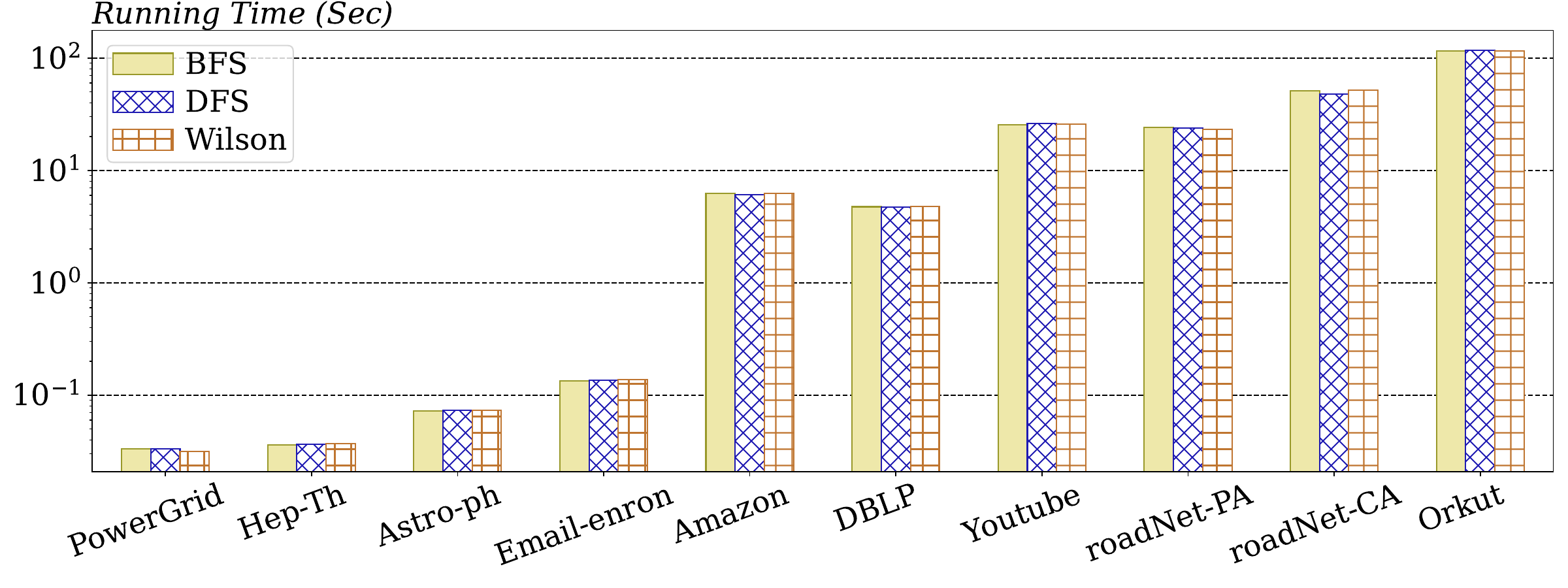}
        \vspace{-0.6cm}
        \caption{Running time of \StaticAlg on each datasets with three different path selection Methods}
        \label{fig:exp_param_path_time}
        \Description{Performance of \StaticAlg on each datasets with three different path selection Methods}
        \vspace{-0.2cm}
    \end{figure}
    
    \begin{figure}[ht!]
        \centering
        \includegraphics[width=\linewidth]{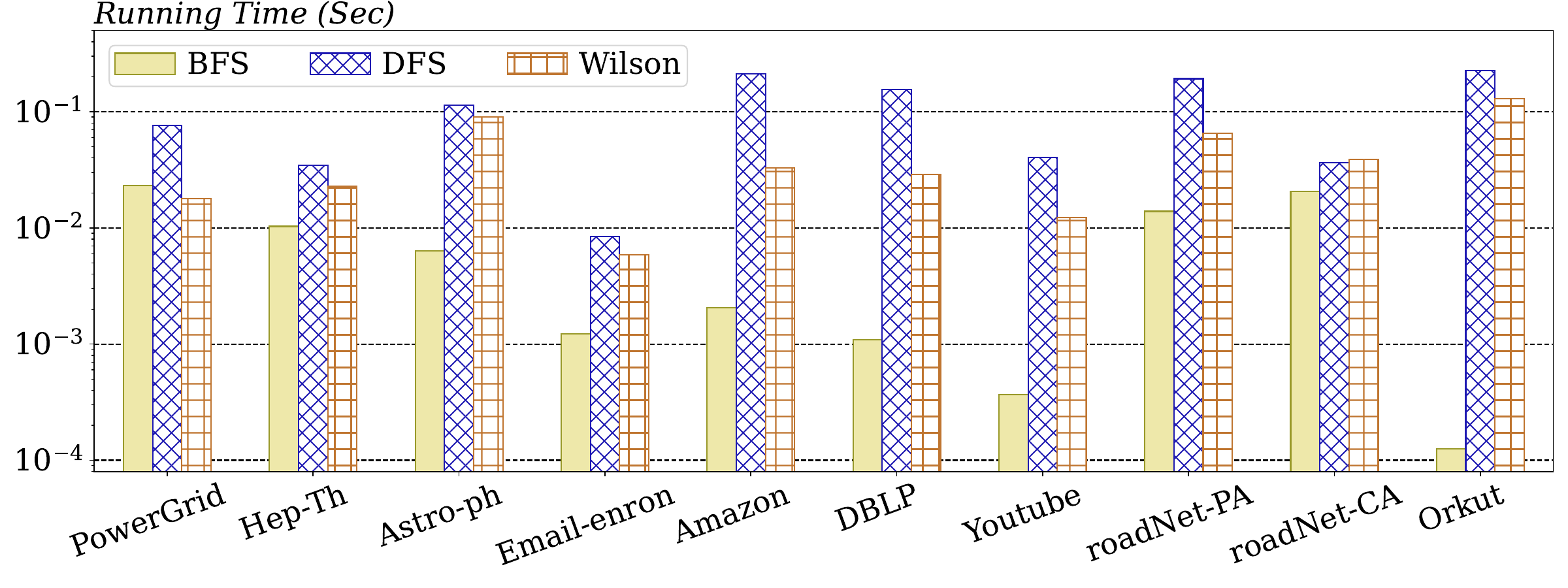}
        \vspace{-0.6cm}
        \caption{Relative error of \StaticAlg on each datasets with three different path selection Methods}
        \label{fig:exp_param_path_error}
        \Description{Performance of \StaticAlg on each datasets with three different path selection Methods}
        \vspace{-0.2cm}
    \end{figure}
\fi

\iffullversion
    \stitle{Exp-III: Effect of Root Selection.} There three static methods (\StaticAlg, \ForestMC and \SpanTree) require selecting a root node when applying Wilson’s algorithm, and this choice has a clear impact on the running time. To systematically evaluate this effect, we randomly select root nodes with degrees close to $\tfrac{d_{\max}}{5}, \tfrac{2d_{\max}}{5}, \ldots, d_{\max}$, and measure the average running time. The results, illustrated in Fig.~\ref{fig:exp_param_root_degree}, show that in most cases choosing the node with the maximum degree yields the fastest performance. Therefore, we adopt the highest-degree node as the default root in our experiments.
    \begin{figure*}[tb!]
        \centering
        \includegraphics[width=\linewidth]{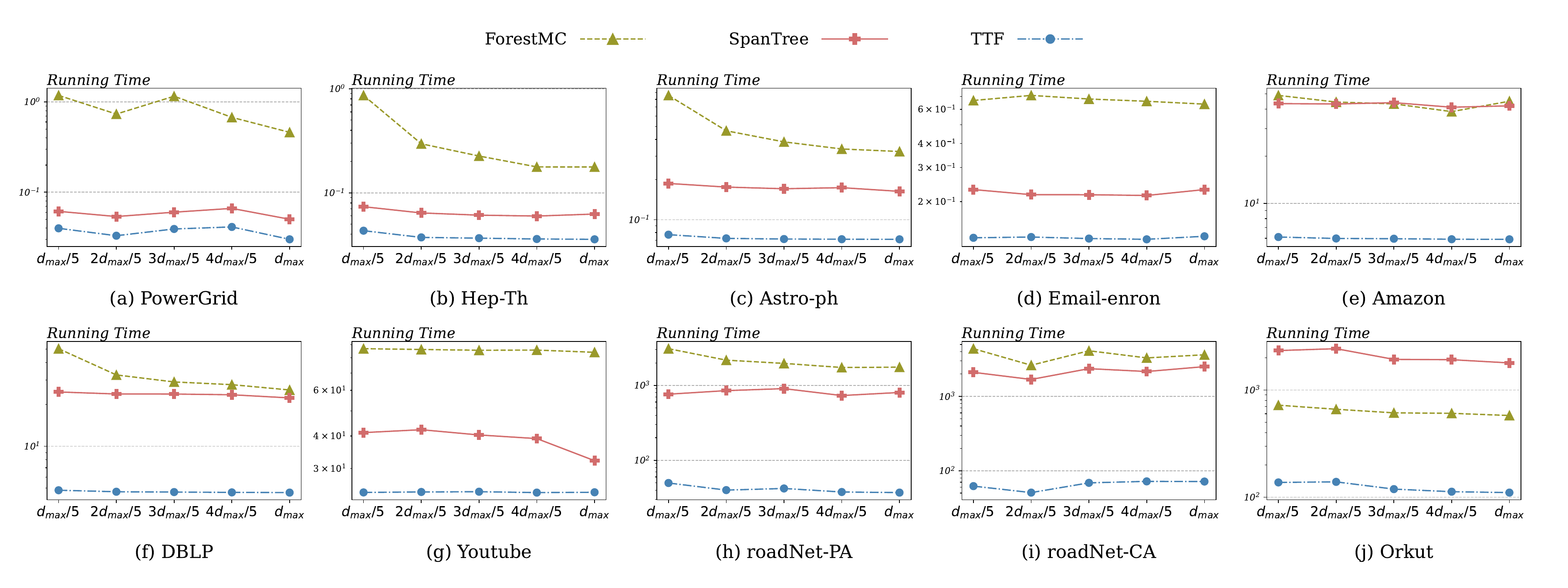}
        \vspace{-0.4cm}
        \caption{Impact of root node with different degrees on the running time of \StaticAlg}
        \label{fig:exp_param_root_degree}
        \Description{Impact of root node with different degrees on the running time of TTF}
        \vspace{-0.4cm}
    \end{figure*}
\fi

\stitle{Exp-IV: Scalability on Synthetically Generated Graphs.} In this experiment, we evaluate the scalability of different methods on randomly generated graphs. We employ two standard random graph models: (i) Erd\H{o}s--R\'{e}nyi graphs~\cite{bollobas1998random}, where each possible edge is independently included with probability $p = 5 \times 10^{-4}$, and (ii) Chung-Lu model~\cite{aiello2001random}, where the probability of creating an edge is proportional to the product of the endpoints' expected degrees, and the degree distribution follows a power law with exponent $\gamma=3.5$. The synthetic graphs are generated using the Networkit toolkit~\cite{networkit}, with their number of nodes in the largest connected component varying from $10^3$ to $10^6$. Fig.~\ref{fig:scale-er} and Fig.~\ref{fig:scale-cl} show that \StaticAlg and \ForestMC scale better, as the running time of other methods increases more rapidly.

\begin{figure}[t!]
    \centering
    \begin{subfigure}{0.49\linewidth}
        \centering
        \includegraphics[width=\textwidth]{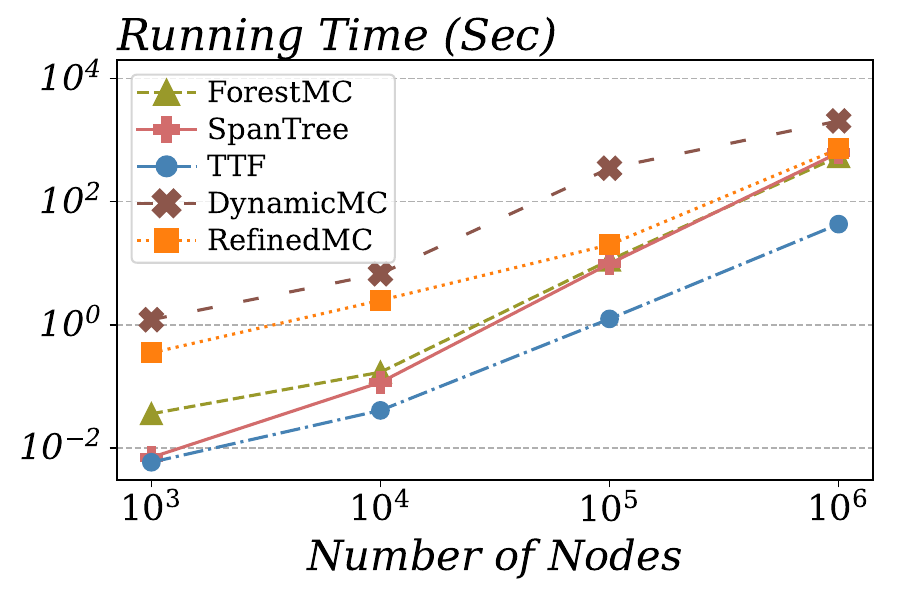}
        \caption{Erd\H{o}s--R\'{e}nyi model}
        \label{fig:scale-er}
    \end{subfigure}
    \hfill
    \begin{subfigure}{0.49\linewidth}
        \centering
        \includegraphics[width=\textwidth]{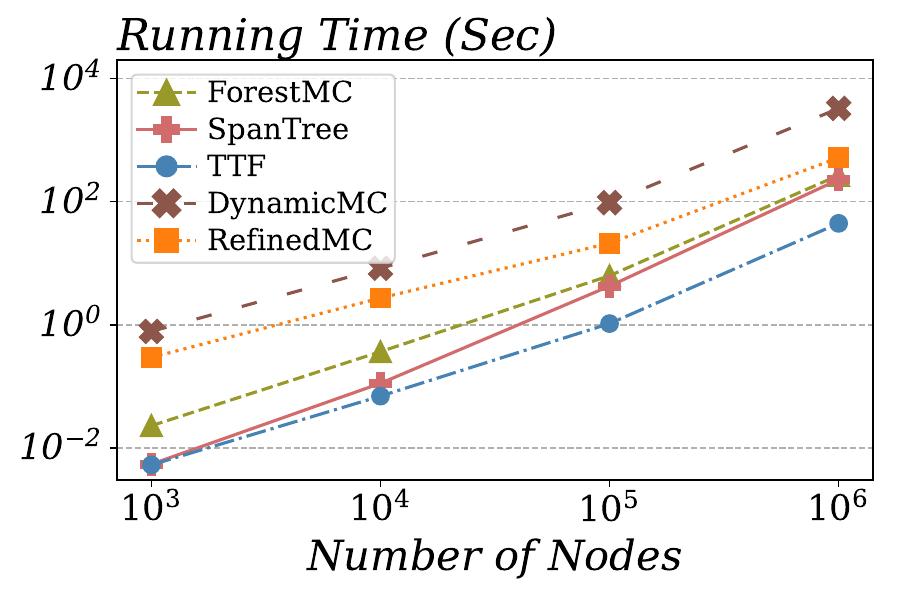}
        \caption{Chung-Lu model}
        \label{fig:scale-cl}
    \end{subfigure}
    \vspace{-0.4cm}
    \caption{\color{RM}Scalability performance of different algorithms on random generated graphs}
    \label{fig:scalability}
    \vspace{-0.6cm}
\end{figure}

\stitle{Exp-V: Scalability on Huge Real-World Graphs.} 
In this experiment, we evaluate the scalability of different methods on two huge real-world datasets: \friendster\ (65,608,366 nodes and 1,806,067,135 edges) and \twitter\ (21,297,772 nodes and 265,025,809 edges). We run \StaticAlg with sufficient sample size $T=1000$ to approximate the ground truth. For all methods, we set the same and smaller sample size, and gradually increase it to compare their relative error and running time, omitting results that exceed 24 hours. Fig.~\ref{fig:exp-scale-friendster} show that \StaticAlg and \ForestMC remain feasible even on such extremely large graphs, whereas other methods either become prohibitively slow or fail to produce accurate results. Overall, \StaticAlg achieves the best trade-off between accuracy and efficiency, demonstrating strong scalability.

\begin{figure}[t!]
    \centering
    \begin{subfigure}{0.49\linewidth}
    \centering
        \includegraphics[width=\linewidth]{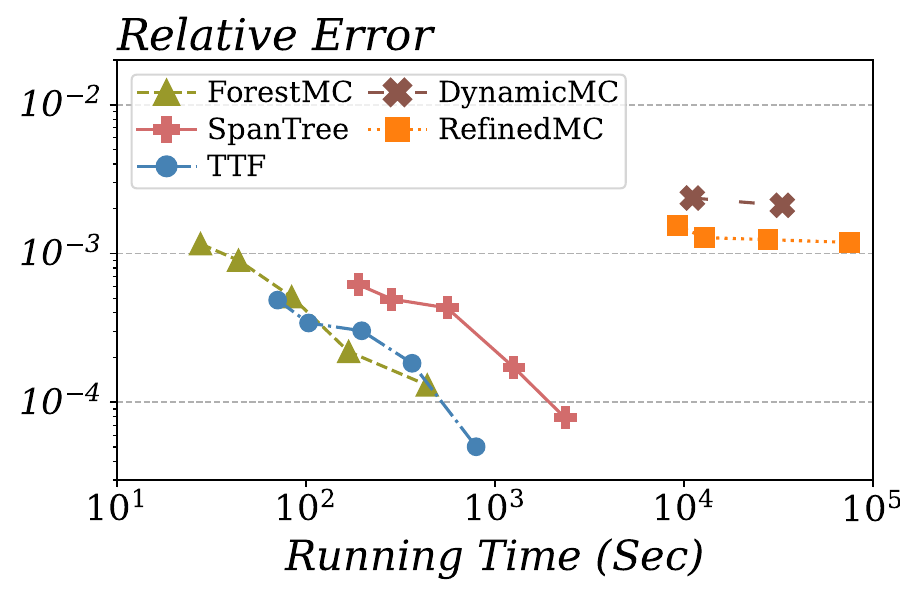}
        \vspace{-0.4cm}
        \caption{\twitter}
    \end{subfigure}
    \hfill
    \begin{subfigure}{0.49\linewidth}
        \centering
        \includegraphics[width=\linewidth]{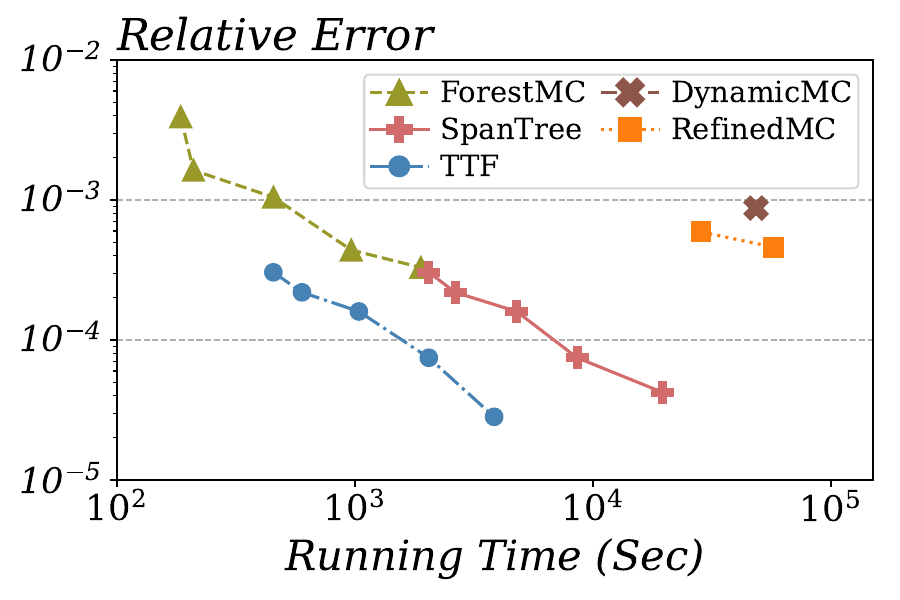}
        \vspace{-0.4cm}
        \caption{\friendster}
    \end{subfigure}
    \vspace{-0.2cm}
    \caption{\color{RM}Relative error vs running time on super-large real-world graphs}
    \Description{Relative error vs running time on \friendster}
    \label{fig:exp-scale-friendster}
    \vspace{-0.4cm}
\end{figure}

\stitle{Exp-VI: Parallelization.} We evaluate the parallel performance of different algorithms using 64 threads against their single-threaded performance. All methods except \Approx naturally parallelize across samples. To ensure fairness, we control the number of samples to make the error levels of all algorithms comparable. As shown in Fig.~\ref{fig:parallel}, all algorithms benefit significantly from parallel execution. Although \StaticAlg does not achieve the highest speedup, it remains among the most efficient methods.

\begin{figure}
    \includegraphics[width=\linewidth]{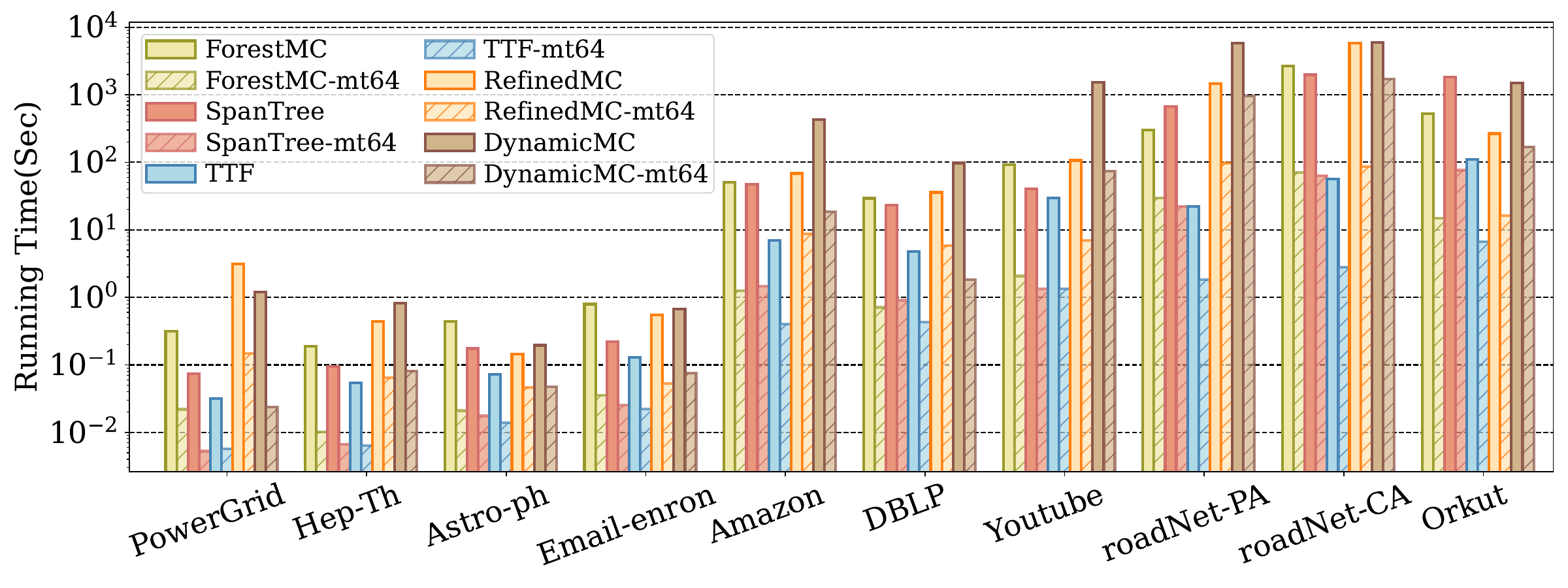}
    \vspace{-0.8cm}
    \caption{\color{RM}Parallel speedup of different algorithms}
    \Description{Parallel speedup of different algorithms on each datasets}
    \label{fig:parallel}
    \vspace{-0.4cm}
\end{figure}

\subsection{Experiment Results on Dynamic Graphs}

\stitle{Exp-VII: Index Size.} This experiment evaluates the index size of both \Basic and \Improved, compared to the size of graphs. The experimental results are shown in Fig.~\ref{fig:exp-index}. We observe that the index sizes of both methods grow proportionally with the graph size, and the growth rate remains stable across all datasets, which empirically validates the $O(n)$ space complexity discussed in Section~4. In all tested datasets, the maximum index size is 4.9GB, demonstrating that our proposed dynamic algorithms are highly space-efficient while supporting real-time sample maintenance.

\begin{figure}
    \centering
    \includegraphics[width=\linewidth]{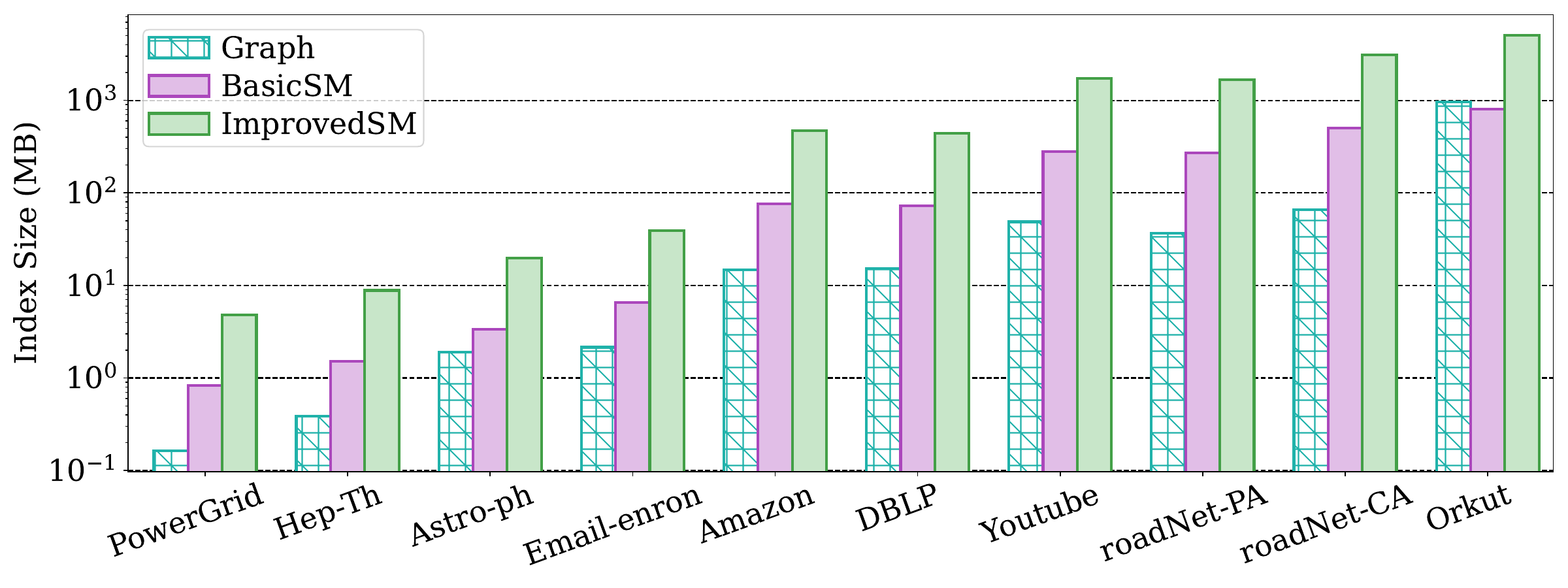}
    \vspace{-0.8cm}
    \caption{\color{RM}Index size of sample maintenance algorithms}
    \label{fig:exp-index}
    \vspace{-0.4cm}
\end{figure}

\stitle{Exp-VIII: Performance for Edge Insertion on Dynamic Graphs.}  
We evaluate the performance of the basic and improved sample maintenance methods, \Basic and \Improved, for edge insertions, comparing their accuracy and running time against three static methods that recompute \kc after each update. The results are shown in Fig.~\ref{fig:exp_ins_time} and Fig.~\ref{fig:exp_ins_error}. As expected, both \Basic and \Improved significantly outperform the static methods in terms of time efficiency. On social networks like \orkut, \Basic and \Improved are an order of magnitude faster than \StaticAlg. However, on road networks such as \roadPA and \roadCA, the speed-up of \Basic is limited, as the average effective resistance of edges is close to one, which means nearly all samples are required to be resampled. Additionally, sampling a UST that includes a specific edge might be slower regular Wilson algorithm. In contrast, \Improved still achieves notable speed-ups even on road networks, since link-cut operations are significantly faster than tree sampling, though it produces slightly higher error.

\begin{figure}[t!]
    \centering
    \includegraphics[width=\linewidth]{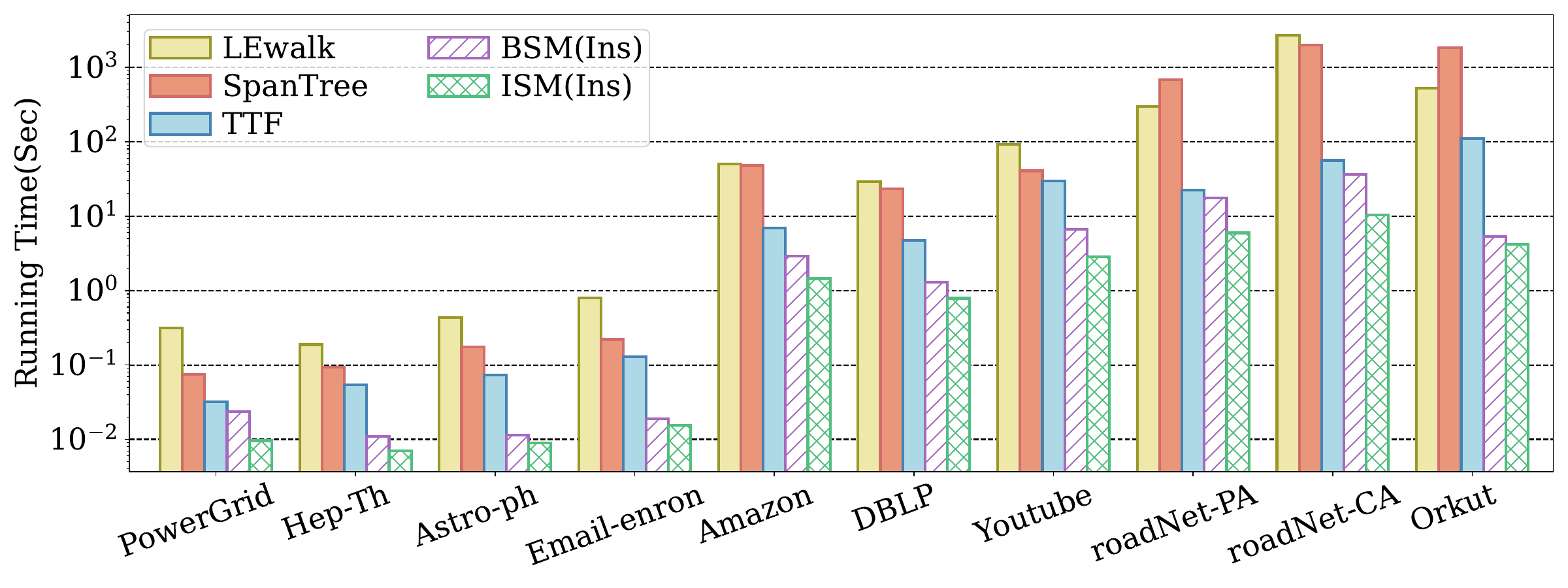}
    \vspace{-0.8cm}
    \caption{Running time of different algorithms for edges insertions on dynamic graphs}
    \label{fig:exp_ins_time}
    \vspace{-0.4cm}
\end{figure}

\begin{figure}[t!]
    \centering
    \includegraphics[width=\linewidth]{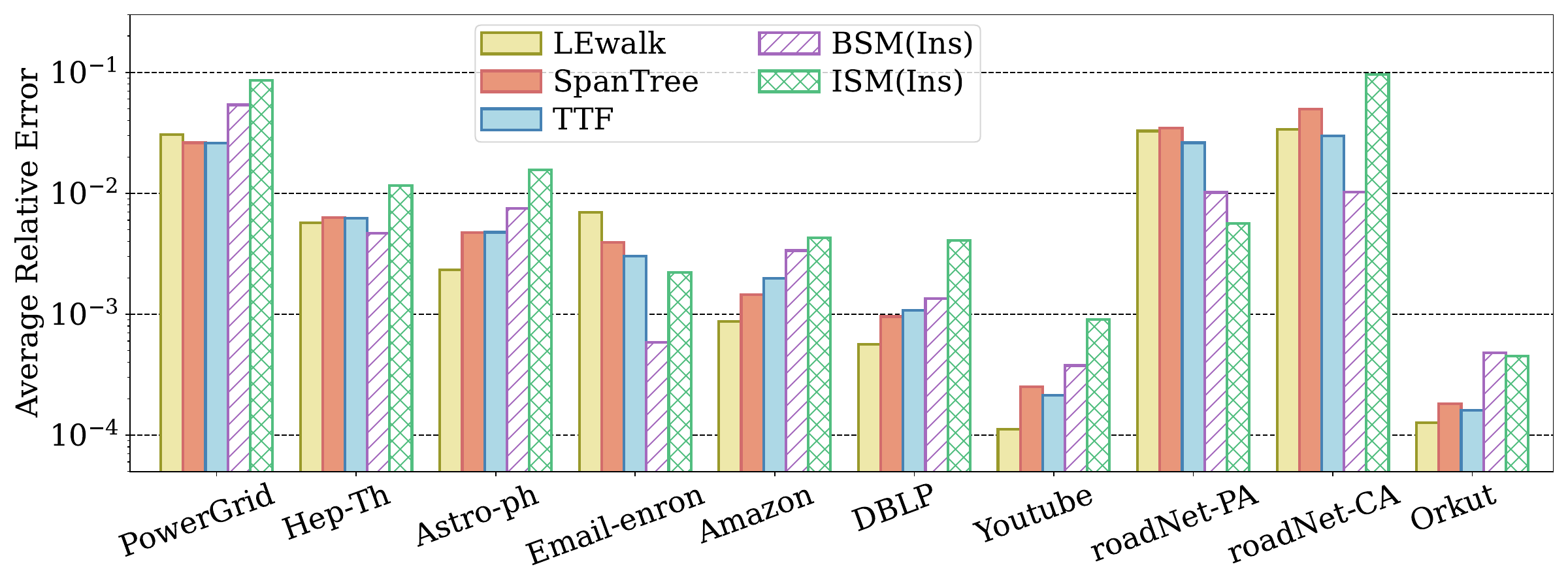}
    \vspace{-0.8cm}
    \caption{Accuracy performance of different algorithms for edges insertions on dynamic graphs}
    \label{fig:exp_ins_error}
    \vspace{-0.4cm}
\end{figure}

\iffullversion
    \stitle{Exp-VIII: Performance for Edge Deletion on Dynamic Graphs.} We also evaluate the performance of \Basic and \Improved for edge deletions. The results, shown in Fig.~\ref{fig:exp_del_time} and Fig.~\ref{fig:exp_del_error}, are generally consistent with those from the insertion scenario. Notably, \Basic shows better performance improvements on road networks during deletions than insertions, since it can employ the standard Wilson algorithm for tree sampling. Moreover, both \Basic and \Improved are slightly faster for deletions than insertions. This is likely because edge insertions require computing effective resistance, which introduces additional computational overhead, whereas deletions do not.
    
    \begin{figure}[t!]
        \centering
        \includegraphics[width=\linewidth]{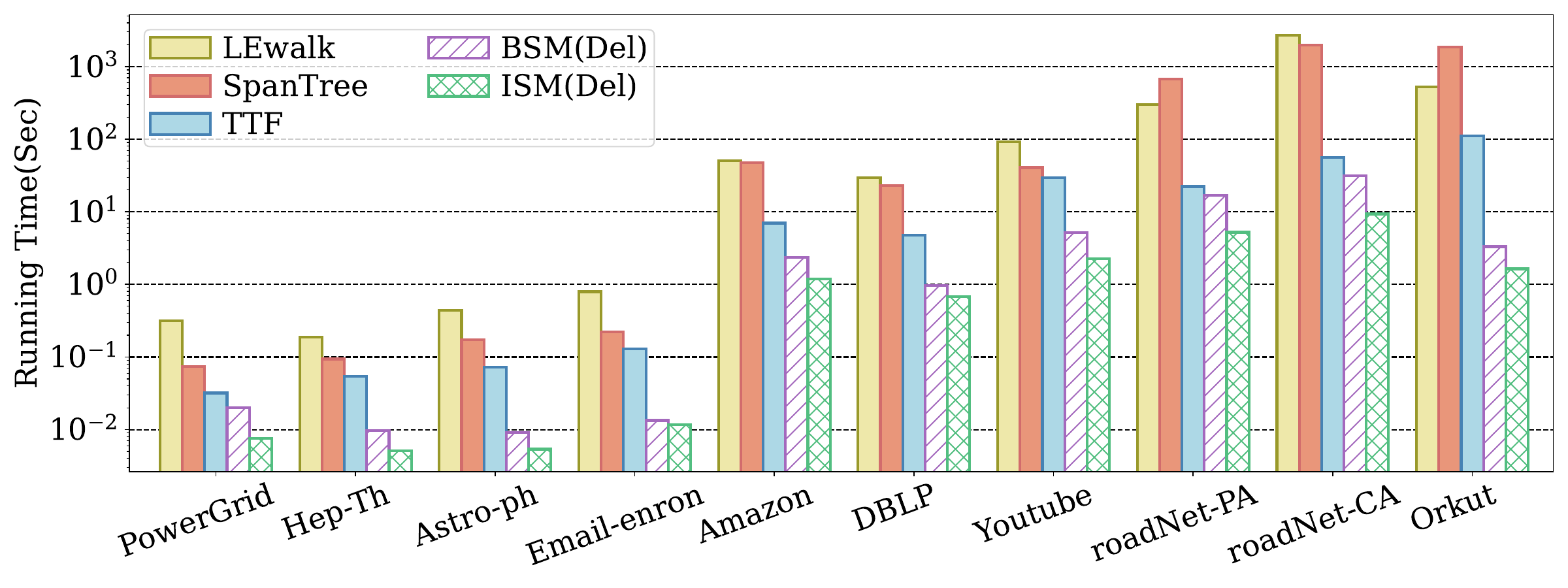}
        \vspace{-0.8cm}
        \caption{Running time of different algorithms for edges deletions on dynamic graphs}
        \vspace{-0.4cm}
        \label{fig:exp_del_time}
    \end{figure}
    
    \begin{figure}[t!]
        \centering
        \includegraphics[width=\linewidth]{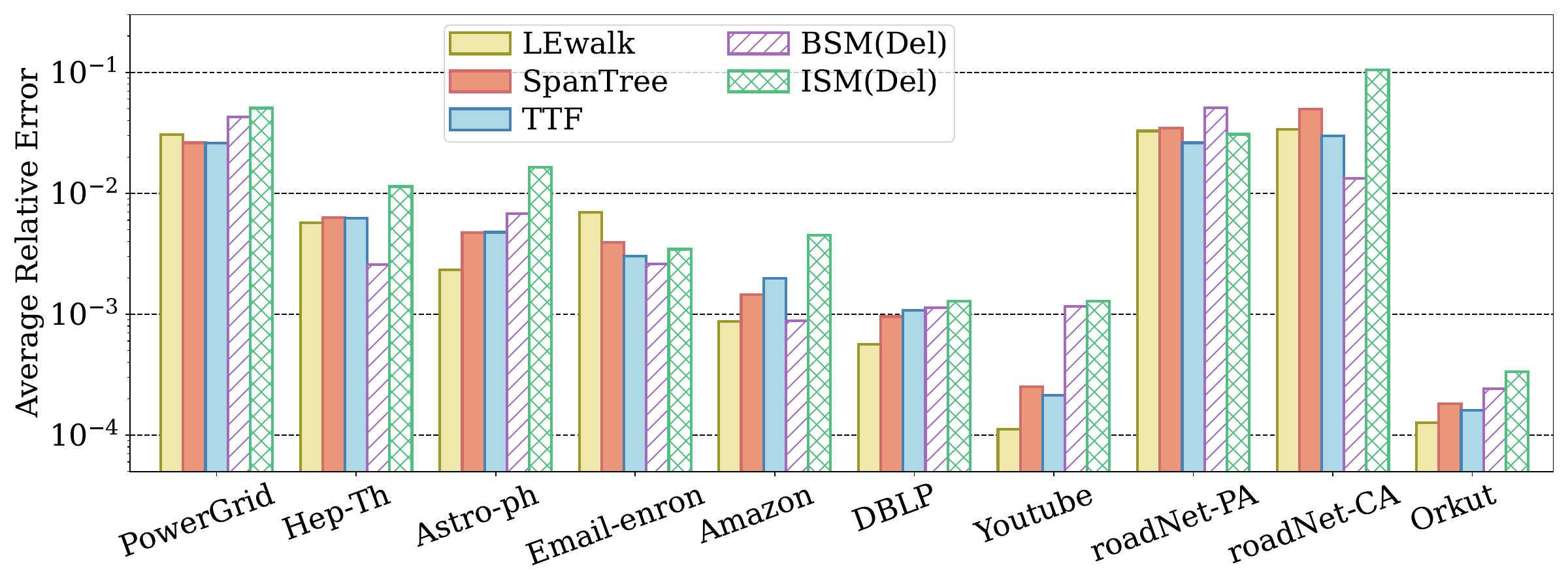}
        \vspace{-0.8cm}
        \caption{Accuracy performance of different algorithms for edges deletions on dynamic graphs}
        \label{fig:exp_del_error}
        \vspace{-0.4cm}
    \end{figure}
\fi

\subsection{Case Studies}
\stitle{Case Study: Graph Classification.} We conducted a case study to investigate the relationship between the Kemeny constant and different types of graphs. For each dataset, we normalized the Kemeny constant as $\kappa(G)/n$ and plotted the results. As shown in Fig~.\ref{fig:case_type}, different graph types exhibit distinct patterns in their normalized Kemeny constants. Social networks and collaboration networks typically have $\kappa(G)/n \approx 1$, whereas road networks generally show larger values. Intuitively, a smaller Kemeny constant indicates that nodes are more tightly connected, which facilitates reaching other nodes from any starting point. This demonstrates that the Kemeny constant is informative for understanding graph connectivity and can be leveraged as a feature for graph classification.

\begin{figure}
    \centering
    \includegraphics[width=0.99\linewidth]{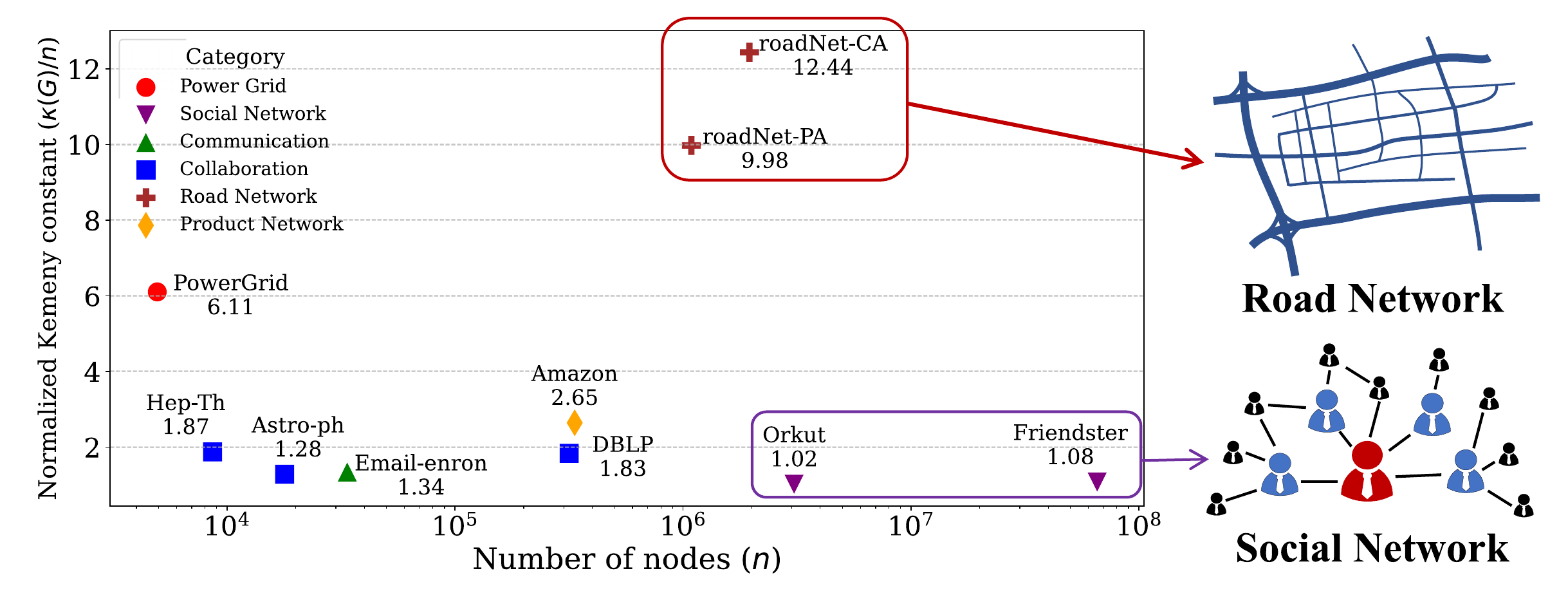}
    \vspace{-0.4cm}
    \caption{\color{RM}Normalized KC across real-world graphs}
    \Description{Normalized Kemeny constant Across Real-World Graphs}
    \label{fig:case_type}
    \vspace{-0.4cm}
\end{figure}

\stitle{Case Study: Kemeny Constant in Dynamic Graphs.}  
We further investigate the evolution of the Kemeny constant in dynamic graphs using the \highschool dataset~\cite{socialpattern}, which captures contacts among students over a single day. Each student is represented as a node, and an edge is created when two students are in close and face-to-face proximity during a 20-second time window. We divided the data into 5-minute windows and computed the Kemeny constant for each window, updating every minute. Figure~\ref{fig:case_highschool} shows the number of edges and the corresponding Kemeny constant over time. Although the number of edges fluctuates considerably, the Kemeny constant remains largely stable, with notable changes only during periods such as class breaks or lunch, when the network structure undergoes significant alterations. This demonstrates that the Kemeny constant can effectively capture structural changes in dynamic networks.

\begin{figure}[t!]
    \begin{subfigure}[c]{0.29\linewidth}
        \centering
        \includegraphics[width=\linewidth]{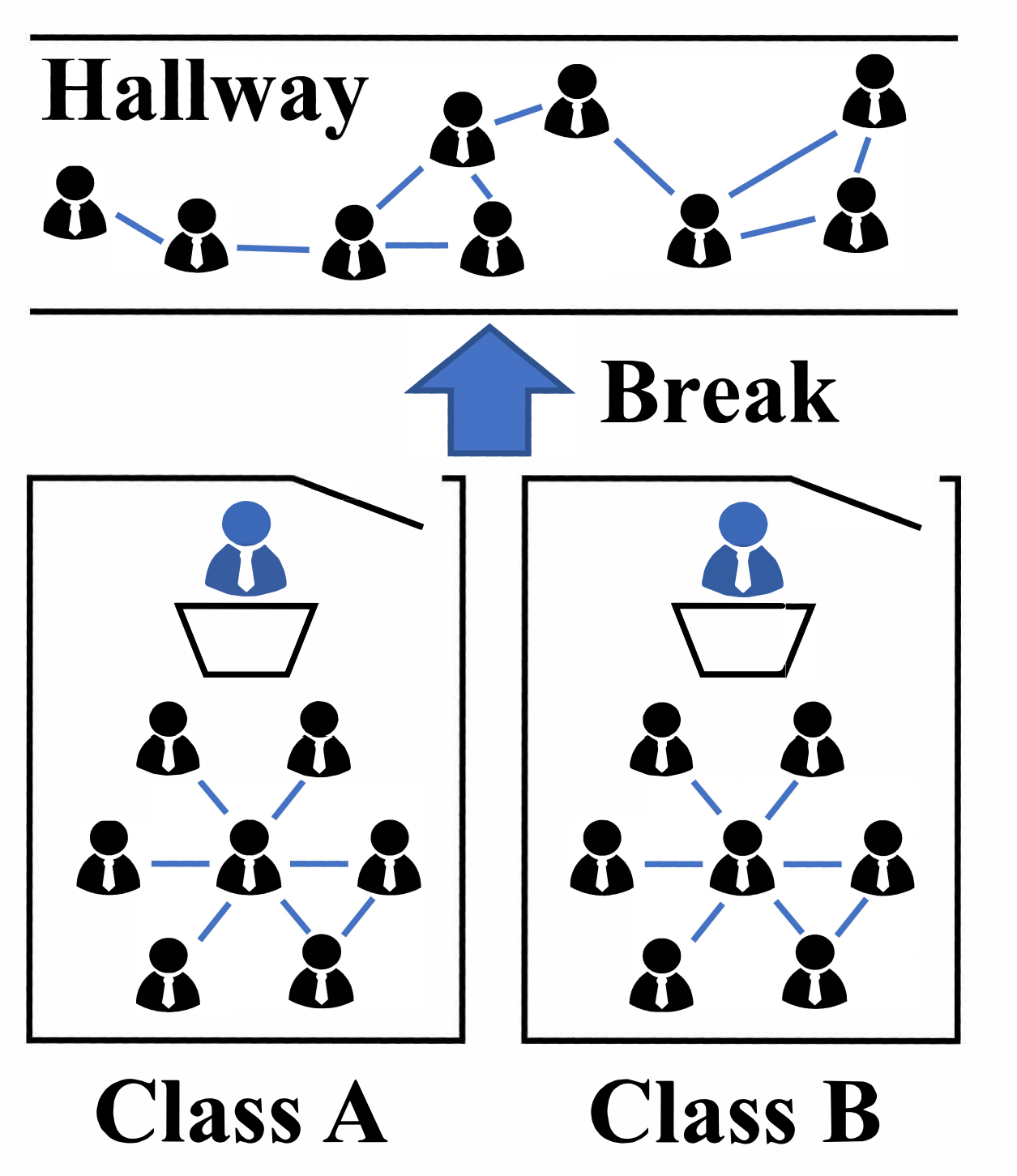}
    \end{subfigure}
    \hfill
    \begin{subfigure}[c]{0.69\linewidth}
        \centering
        \includegraphics[width=\linewidth]{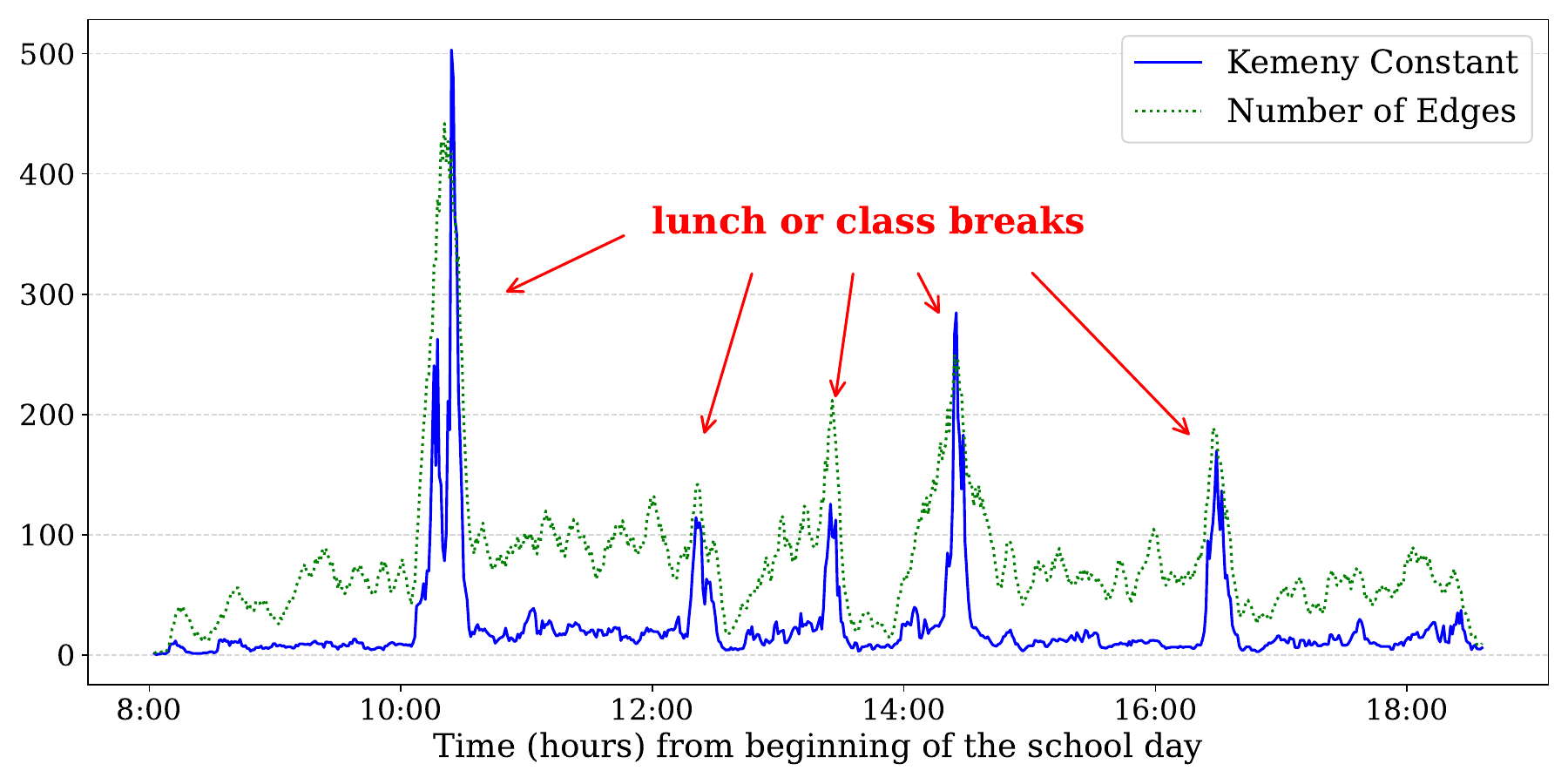}
    \end{subfigure}
    \vspace{-0.2cm}
    \caption{\color{RM}Evolution of Kemeny constant over time}
    \Description{Evolution of Kemeny Constant over Time in the \textit{socialpattern-highschool} Dataset}
    \label{fig:case_highschool}
    \vspace{-0.5cm}
\end{figure}

\section{Related Work}

\stitle{Random Walk Computation.} 
Random walks have long served as a fundamental tool in graph analysis, supporting a wide range of applications such as recommendation~\cite{rw_recommendation_paudel2021random, rw_recommendation2_xia2016scientific, vldb_app_ht_reccomandation_yin2012challenging}, network embedding~\cite{rw_embedding1_perozzi2014deepwalk, rw_embedding2_grover2016node2vec, vldb_rw_embedding_zhao2013embeddability}, and complex network analysis~\cite{rw_analysis_sarkar2011random, rw_analysis2_xia2019random}. Various random walk-based metrics, personalized PageRank (PPR)~\cite{vldb_ppr_wang2022edge, vldb_ppr_embedding_yang13homogeneous, vldb_ppr_liu2024bird}  and effective resistance (ER)~\cite{er1_bollobas2013modern, res1_tetali1991random, er2_peng2021local, er3_yang2023efficient} have received a lot of attention. For both PPR and ER, local algorithms such as \push are commonly used in recent studies~\cite{pr_push_andersen2006local, vldb_ppr_wang2022edge, res4_liao2023efficient}, as they operate efficiently on a small portion of the graph without requiring full access to its structure. Although \kc is also inherently related to random walks, it represents a fundamentally global measure, distinct from these problems, which makes existing techniques developed for PPR and ER unsuitable for direct application to \kc computation.

\stitle{Algorithms on Dynamic Graphs.} Dynamic graphs, also known as evolving graphs or graphs in an incremental setting, naturally arise in real-world applications where graph data is continuously updated. The design of algorithms for various graph problems in dynamic graphs has been an active research area for decades. In particular, dynamic algorithms for random walk-based metrics, such as Personalized PageRank, have been extensively studied~\cite{pr_push2_zhang2016approximate, pr_push1_zheng2022instant, pr_rw1_hou2023personalized, pr_pi1_li2023everything, ppr_dy_yang2024efficient}. Most dynamic algorithms build upon their corresponding static versions. For example, \LazyForward~\cite{pr_push2_zhang2016approximate} and \TrackingPPR~\cite{pr_push3_ohsaka2015efficient} were both developed as extensions of the push algorithm~\cite{pr_push_andersen2006local}. Similarly, various index update methods for random walks~\cite{pr_rw1_hou2023personalized, pr_rw3_mo2021agenda, pr_rw2_bahmani2010fast} and power iteration-based dynamic algorithms~\cite{pr_pi1_li2023everything, pr_pi2_yoon2018fast} have been explored. Recently, Liao et al. investigated the connection between USTs and both Personalized PageRank~\cite{ppr_tree1_liao2022efficient, ppr_tree2_liao2023efficient} and effective resistance~\cite{res4_liao2023efficient} successively. These methods provide a foundation for extending our index maintenance strategy, as shown in Section~\ref{sec:dynamic algorithm}, to these problems, facilitating the development of efficient dynamic algorithms.  

\section{Conclusion}
In this work, we study the problem of approximating Kemeny constant problem for both static and dynamic graphs. We propose two novel formulas of Kemeny constant and design an biased estimator based on spanning trees and 2-forests. For static graphs, we develop a sampling-based algorithm with stronger theoretical guarantees. The proposed method outperforms \sota approaches in both efficiency and accuracy on real-world datasets. For dynamic graphs, we introduce two sample maintenance strategies that efficiently preserve the correctness of samples instead of recomputing from scratch for each update. Both two methods are faster than static algorithms, with only a slight loss in precision. Extensive experiments on real-world graphs demonstrate the effectiveness and efficiency of our solutions.

\balance


\begin{thebibliography}{61}


\ifx \showCODEN    \undefined \def \showCODEN     #1{\unskip}     \fi
\ifx \showDOI      \undefined \def \showDOI       #1{#1}\fi
\ifx \showISBNx    \undefined \def \showISBNx     #1{\unskip}     \fi
\ifx \showISBNxiii \undefined \def \showISBNxiii  #1{\unskip}     \fi
\ifx \showISSN     \undefined \def \showISSN      #1{\unskip}     \fi
\ifx \showLCCN     \undefined \def \showLCCN      #1{\unskip}     \fi
\ifx \shownote     \undefined \def \shownote      #1{#1}          \fi
\ifx \showarticletitle \undefined \def \showarticletitle #1{#1}   \fi
\ifx \showURL      \undefined \def \showURL       {\relax}        \fi
\providecommand\bibfield[2]{#2}
\providecommand\bibinfo[2]{#2}
\providecommand\natexlab[1]{#1}
\providecommand\showeprint[2][]{arXiv:#2}

\bibitem[\protect\citeauthoryear{Aiello, Chung, and Lu}{Aiello et~al\mbox{.}}{2001}]%
        {aiello2001random}
\bibfield{author}{\bibinfo{person}{William Aiello}, \bibinfo{person}{Fan Chung}, {and} \bibinfo{person}{Linyuan Lu}.} \bibinfo{year}{2001}\natexlab{}.
\newblock \showarticletitle{A random graph model for power law graphs}.
\newblock \bibinfo{journal}{\emph{Experimental mathematics}} \bibinfo{volume}{10}, \bibinfo{number}{1} (\bibinfo{year}{2001}), \bibinfo{pages}{53--66}.
\newblock


\bibitem[\protect\citeauthoryear{Altafini, Bini, Cutini, Meini, and Poloni}{Altafini et~al\mbox{.}}{2023}]%
        {app_centrality}
\bibfield{author}{\bibinfo{person}{Diego Altafini}, \bibinfo{person}{Dario~A Bini}, \bibinfo{person}{Valerio Cutini}, \bibinfo{person}{Beatrice Meini}, {and} \bibinfo{person}{Federico Poloni}.} \bibinfo{year}{2023}\natexlab{}.
\newblock \showarticletitle{An edge centrality measure based on the Kemeny constant}.
\newblock \bibinfo{journal}{\emph{SIAM J. Matrix Anal. Appl.}} \bibinfo{volume}{44}, \bibinfo{number}{2} (\bibinfo{year}{2023}), \bibinfo{pages}{648--669}.
\newblock


\bibitem[\protect\citeauthoryear{Andersen, Chung, and Lang}{Andersen et~al\mbox{.}}{2006}]%
        {pr_push_andersen2006local}
\bibfield{author}{\bibinfo{person}{Reid Andersen}, \bibinfo{person}{Fan Chung}, {and} \bibinfo{person}{Kevin Lang}.} \bibinfo{year}{2006}\natexlab{}.
\newblock \showarticletitle{Local graph partitioning using pagerank vectors}. In \bibinfo{booktitle}{\emph{FOCS}}. IEEE, \bibinfo{pages}{475--486}.
\newblock


\bibitem[\protect\citeauthoryear{Angriman, Predari, van~der Grinten, and Meyerhenke}{Angriman et~al\mbox{.}}{2020}]%
        {angriman2020approximation}
\bibfield{author}{\bibinfo{person}{Eugenio Angriman}, \bibinfo{person}{Maria Predari}, \bibinfo{person}{Alexander van~der Grinten}, {and} \bibinfo{person}{Henning Meyerhenke}.} \bibinfo{year}{2020}\natexlab{}.
\newblock \showarticletitle{Approximation of the Diagonal of a Laplacian's Pseudoinverse for Complex Network Analysis}.
\newblock \bibinfo{journal}{\emph{ESA}}  \bibinfo{volume}{173} (\bibinfo{year}{2020}), \bibinfo{pages}{6:1--6:24}.
\newblock


\bibitem[\protect\citeauthoryear{Avena, Castell, Gaudilli{\`e}re, and M{\'e}lot}{Avena et~al\mbox{.}}{2018}]%
        {wilson2_avena2018random}
\bibfield{author}{\bibinfo{person}{Luca Avena}, \bibinfo{person}{Fabienne Castell}, \bibinfo{person}{Alexandre Gaudilli{\`e}re}, {and} \bibinfo{person}{Clothilde M{\'e}lot}.} \bibinfo{year}{2018}\natexlab{}.
\newblock \showarticletitle{Random forests and networks analysis}.
\newblock \bibinfo{journal}{\emph{Journal of Statistical Physics}}  \bibinfo{volume}{173} (\bibinfo{year}{2018}), \bibinfo{pages}{985--1027}.
\newblock


\bibitem[\protect\citeauthoryear{Bahmani, Chowdhury, and Goel}{Bahmani et~al\mbox{.}}{2010}]%
        {pr_rw2_bahmani2010fast}
\bibfield{author}{\bibinfo{person}{Bahman Bahmani}, \bibinfo{person}{Abdur Chowdhury}, {and} \bibinfo{person}{Ashish Goel}.} \bibinfo{year}{2010}\natexlab{}.
\newblock \showarticletitle{Fast Incremental and Personalized PageRank}.
\newblock \bibinfo{journal}{\emph{VLDB}} \bibinfo{volume}{4}, \bibinfo{number}{3} (\bibinfo{year}{2010}), \bibinfo{pages}{173--184}.
\newblock


\bibitem[\protect\citeauthoryear{Bollob{\'a}s}{Bollob{\'a}s}{2013}]%
        {er1_bollobas2013modern}
\bibfield{author}{\bibinfo{person}{B{\'e}la Bollob{\'a}s}.} \bibinfo{year}{2013}\natexlab{}.
\newblock \bibinfo{booktitle}{\emph{Modern graph theory}}. Vol.~\bibinfo{volume}{184}.
\newblock \bibinfo{publisher}{Springer Science \& Business Media}.
\newblock


\bibitem[\protect\citeauthoryear{Bollob{\'a}s and Bollob{\'a}s}{Bollob{\'a}s and Bollob{\'a}s}{1998}]%
        {bollobas1998random}
\bibfield{author}{\bibinfo{person}{B{\'e}la Bollob{\'a}s} {and} \bibinfo{person}{B{\'e}la Bollob{\'a}s}.} \bibinfo{year}{1998}\natexlab{}.
\newblock \bibinfo{booktitle}{\emph{Random graphs}}.
\newblock \bibinfo{publisher}{Springer}.
\newblock


\bibitem[\protect\citeauthoryear{Chen, Liu, and Tang}{Chen et~al\mbox{.}}{2008}]%
        {ht_clustering_chen2008clustering}
\bibfield{author}{\bibinfo{person}{Mo Chen}, \bibinfo{person}{Jianzhuang Liu}, {and} \bibinfo{person}{Xiaoou Tang}.} \bibinfo{year}{2008}\natexlab{}.
\newblock \showarticletitle{Clustering via Random Walk Hitting Time on Directed Graphs}. In \bibinfo{booktitle}{\emph{AAAI}}, Vol.~\bibinfo{volume}{8}. \bibinfo{pages}{616--621}.
\newblock


\bibitem[\protect\citeauthoryear{Chung and Zeng}{Chung and Zeng}{2023}]%
        {chung2023forest}
\bibfield{author}{\bibinfo{person}{Fan Chung} {and} \bibinfo{person}{Ji Zeng}.} \bibinfo{year}{2023}\natexlab{}.
\newblock \showarticletitle{Forest formulas of discrete Green's functions}.
\newblock \bibinfo{journal}{\emph{J. Graph Theory}} \bibinfo{volume}{102}, \bibinfo{number}{3} (\bibinfo{year}{2023}), \bibinfo{pages}{556--577}.
\newblock


\bibitem[\protect\citeauthoryear{Condamin, B{\'e}nichou, Tejedor, Voituriez, and Klafter}{Condamin et~al\mbox{.}}{2007}]%
        {condamin2007first}
\bibfield{author}{\bibinfo{person}{S Condamin}, \bibinfo{person}{O B{\'e}nichou}, \bibinfo{person}{V Tejedor}, \bibinfo{person}{R Voituriez}, {and} \bibinfo{person}{Joseph Klafter}.} \bibinfo{year}{2007}\natexlab{}.
\newblock \showarticletitle{First-passage times in complex scale-invariant media}.
\newblock \bibinfo{journal}{\emph{Nature}} \bibinfo{volume}{450}, \bibinfo{number}{7166} (\bibinfo{year}{2007}), \bibinfo{pages}{77--80}.
\newblock


\bibitem[\protect\citeauthoryear{Crisostomi, Kirkland, and Shorten}{Crisostomi et~al\mbox{.}}{2011}]%
        {crisostomi2011google}
\bibfield{author}{\bibinfo{person}{Emanuele Crisostomi}, \bibinfo{person}{Stephen Kirkland}, {and} \bibinfo{person}{Robert Shorten}.} \bibinfo{year}{2011}\natexlab{}.
\newblock \showarticletitle{A Google-like model of road network dynamics and its application to regulation and control}.
\newblock \bibinfo{journal}{\emph{Internat. J. Control}} \bibinfo{volume}{84}, \bibinfo{number}{3} (\bibinfo{year}{2011}), \bibinfo{pages}{633--651}.
\newblock


\bibitem[\protect\citeauthoryear{Doyle}{Doyle}{2009}]%
        {doyle2009kemeny}
\bibfield{author}{\bibinfo{person}{Peter~G Doyle}.} \bibinfo{year}{2009}\natexlab{}.
\newblock \showarticletitle{The Kemeny constant of a Markov chain}.
\newblock \bibinfo{journal}{\emph{arXiv preprint arXiv:0909.2636}} (\bibinfo{year}{2009}).
\newblock


\bibitem[\protect\citeauthoryear{Fenwick}{Fenwick}{1994}]%
        {Fenwick94}
\bibfield{author}{\bibinfo{person}{Peter~M. Fenwick}.} \bibinfo{year}{1994}\natexlab{}.
\newblock \showarticletitle{A New Data Structure for Cumulative Frequency Tables}.
\newblock \bibinfo{journal}{\emph{Softw. Pract. Exp.}} \bibinfo{volume}{24}, \bibinfo{number}{3} (\bibinfo{year}{1994}), \bibinfo{pages}{327--336}.
\newblock


\bibitem[\protect\citeauthoryear{Fournet and Barrat}{Fournet and Barrat}{2014}]%
        {socialpattern}
\bibfield{author}{\bibinfo{person}{Julie Fournet} {and} \bibinfo{person}{Alain Barrat}.} \bibinfo{year}{2014}\natexlab{}.
\newblock \showarticletitle{Contact Patterns among High School Students}.
\newblock \bibinfo{journal}{\emph{PLoS ONE}} \bibinfo{volume}{9}, \bibinfo{number}{9} (\bibinfo{date}{09} \bibinfo{year}{2014}), \bibinfo{pages}{e107878}.
\newblock
\urldef\tempurl%
\url{https://doi.org/10.1371/journal.pone.0107878}
\showDOI{\tempurl}


\bibitem[\protect\citeauthoryear{Grover and Leskovec}{Grover and Leskovec}{2016}]%
        {rw_embedding2_grover2016node2vec}
\bibfield{author}{\bibinfo{person}{Aditya Grover} {and} \bibinfo{person}{Jure Leskovec}.} \bibinfo{year}{2016}\natexlab{}.
\newblock \showarticletitle{node2vec: Scalable feature learning for networks}. In \bibinfo{booktitle}{\emph{KDD}}. \bibinfo{pages}{855--864}.
\newblock


\bibitem[\protect\citeauthoryear{Hayashi, Akiba, and Yoshida}{Hayashi et~al\mbox{.}}{2016}]%
        {sp_tree_hayashi2016efficient}
\bibfield{author}{\bibinfo{person}{Takanori Hayashi}, \bibinfo{person}{Takuya Akiba}, {and} \bibinfo{person}{Yuichi Yoshida}.} \bibinfo{year}{2016}\natexlab{}.
\newblock \showarticletitle{Efficient Algorithms for Spanning Tree Centrality}. In \bibinfo{booktitle}{\emph{IJCAI}}, Vol.~\bibinfo{volume}{16}. \bibinfo{pages}{3733--3739}.
\newblock


\bibitem[\protect\citeauthoryear{Hoeffding}{Hoeffding}{1994}]%
        {hoeffding1994probability}
\bibfield{author}{\bibinfo{person}{Wassily Hoeffding}.} \bibinfo{year}{1994}\natexlab{}.
\newblock \showarticletitle{Probability inequalities for sums of bounded random variables}.
\newblock \bibinfo{journal}{\emph{The collected works of Wassily Hoeffding}} (\bibinfo{year}{1994}), \bibinfo{pages}{409--426}.
\newblock


\bibitem[\protect\citeauthoryear{Hou, Guo, Zhang, Wang, and Wei}{Hou et~al\mbox{.}}{2023}]%
        {pr_rw1_hou2023personalized}
\bibfield{author}{\bibinfo{person}{Guanhao Hou}, \bibinfo{person}{Qintian Guo}, \bibinfo{person}{Fangyuan Zhang}, \bibinfo{person}{Sibo Wang}, {and} \bibinfo{person}{Zhewei Wei}.} \bibinfo{year}{2023}\natexlab{}.
\newblock \showarticletitle{Personalized PageRank on evolving graphs with an incremental index-update scheme}.
\newblock \bibinfo{journal}{\emph{SIGMOD}} \bibinfo{volume}{1}, \bibinfo{number}{1} (\bibinfo{year}{2023}), \bibinfo{pages}{1--26}.
\newblock


\bibitem[\protect\citeauthoryear{Hunter}{Hunter}{2014}]%
        {hunter2014role}
\bibfield{author}{\bibinfo{person}{Jeffrey~J Hunter}.} \bibinfo{year}{2014}\natexlab{}.
\newblock \showarticletitle{The role of Kemeny's constant in properties of Markov chains}.
\newblock \bibinfo{journal}{\emph{Communications in Statistics-Theory and Methods}} \bibinfo{volume}{43}, \bibinfo{number}{7} (\bibinfo{year}{2014}), \bibinfo{pages}{1309--1321}.
\newblock


\bibitem[\protect\citeauthoryear{Hutchinson}{Hutchinson}{1989}]%
        {hutchinson1989stochastic}
\bibfield{author}{\bibinfo{person}{Michael~F Hutchinson}.} \bibinfo{year}{1989}\natexlab{}.
\newblock \showarticletitle{A stochastic estimator of the trace of the influence matrix for Laplacian smoothing splines}.
\newblock \bibinfo{journal}{\emph{Communications in Statistics-Simulation and Computation}} \bibinfo{volume}{18}, \bibinfo{number}{3} (\bibinfo{year}{1989}), \bibinfo{pages}{1059--1076}.
\newblock


\bibitem[\protect\citeauthoryear{Kemeny, Snell, et~al\mbox{.}}{Kemeny et~al\mbox{.}}{1969}]%
        {kemeny1969}
\bibfield{author}{\bibinfo{person}{John~G Kemeny}, \bibinfo{person}{J~Laurie Snell}, {et~al\mbox{.}}} \bibinfo{year}{1969}\natexlab{}.
\newblock \bibinfo{booktitle}{\emph{Finite markov chains}}. Vol.~\bibinfo{volume}{26}.
\newblock \bibinfo{publisher}{van Nostrand Princeton, NJ}.
\newblock


\bibitem[\protect\citeauthoryear{Leskovec and Krevl}{Leskovec and Krevl}{2014}]%
        {snapnets}
\bibfield{author}{\bibinfo{person}{Jure Leskovec} {and} \bibinfo{person}{Andrej Krevl}.} \bibinfo{year}{2014}\natexlab{}.
\newblock \bibinfo{title}{{SNAP Datasets}: {Stanford} Large Network Dataset Collection}.
\newblock \bibinfo{howpublished}{\url{http://snap.stanford.edu/data}}.
\newblock


\bibitem[\protect\citeauthoryear{Levene and Loizou}{Levene and Loizou}{2002}]%
        {levene2002kemeny}
\bibfield{author}{\bibinfo{person}{Mark Levene} {and} \bibinfo{person}{George Loizou}.} \bibinfo{year}{2002}\natexlab{}.
\newblock \showarticletitle{Kemeny's constant and the random surfer}.
\newblock \bibinfo{journal}{\emph{The American mathematical monthly}} \bibinfo{volume}{109}, \bibinfo{number}{8} (\bibinfo{year}{2002}), \bibinfo{pages}{741--745}.
\newblock


\bibitem[\protect\citeauthoryear{Li, Huang, and Lee}{Li et~al\mbox{.}}{2021}]%
        {li2021efficient}
\bibfield{author}{\bibinfo{person}{Shiju Li}, \bibinfo{person}{Xin Huang}, {and} \bibinfo{person}{Chul-Ho Lee}.} \bibinfo{year}{2021}\natexlab{}.
\newblock \showarticletitle{An efficient and scalable algorithm for estimating Kemeny's constant of a Markov chain on large graphs}. In \bibinfo{booktitle}{\emph{KDD}}. \bibinfo{pages}{964--974}.
\newblock


\bibitem[\protect\citeauthoryear{Li, Fu, and He}{Li et~al\mbox{.}}{2023}]%
        {pr_pi1_li2023everything}
\bibfield{author}{\bibinfo{person}{Zihao Li}, \bibinfo{person}{Dongqi Fu}, {and} \bibinfo{person}{Jingrui He}.} \bibinfo{year}{2023}\natexlab{}.
\newblock \showarticletitle{Everything evolves in personalized pagerank}. In \bibinfo{booktitle}{\emph{WWW}}. \bibinfo{pages}{3342--3352}.
\newblock


\bibitem[\protect\citeauthoryear{Liao, Li, Dai, Chen, Qin, and Wang}{Liao et~al\mbox{.}}{2023b}]%
        {ppr_tree2_liao2023efficient}
\bibfield{author}{\bibinfo{person}{Meihao Liao}, \bibinfo{person}{Rong-Hua Li}, \bibinfo{person}{Qiangqiang Dai}, \bibinfo{person}{Hongyang Chen}, \bibinfo{person}{Hongchao Qin}, {and} \bibinfo{person}{Guoren Wang}.} \bibinfo{year}{2023}\natexlab{b}.
\newblock \showarticletitle{Efficient personalized pagerank computation: The power of variance-reduced monte carlo approaches}.
\newblock \bibinfo{journal}{\emph{SIGMOD}} \bibinfo{volume}{1}, \bibinfo{number}{2} (\bibinfo{year}{2023}), \bibinfo{pages}{1--26}.
\newblock


\bibitem[\protect\citeauthoryear{Liao, Li, Dai, Chen, Qin, and Wang}{Liao et~al\mbox{.}}{2023c}]%
        {res4_liao2023efficient}
\bibfield{author}{\bibinfo{person}{Meihao Liao}, \bibinfo{person}{Rong-Hua Li}, \bibinfo{person}{Qiangqiang Dai}, \bibinfo{person}{Hongyang Chen}, \bibinfo{person}{Hongchao Qin}, {and} \bibinfo{person}{Guoren Wang}.} \bibinfo{year}{2023}\natexlab{c}.
\newblock \showarticletitle{Efficient resistance distance computation: The power of landmark-based approaches}.
\newblock \bibinfo{journal}{\emph{SIGMOD}} \bibinfo{volume}{1}, \bibinfo{number}{1} (\bibinfo{year}{2023}), \bibinfo{pages}{1--27}.
\newblock


\bibitem[\protect\citeauthoryear{Liao, Li, Dai, Chen, and Wang}{Liao et~al\mbox{.}}{2023a}]%
        {liao2023scalable}
\bibfield{author}{\bibinfo{person}{Meihao Liao}, \bibinfo{person}{Rong-Hua Li}, \bibinfo{person}{Qiangqiang Dai}, \bibinfo{person}{Hongyang Chen}, {and} \bibinfo{person}{Guoren Wang}.} \bibinfo{year}{2023}\natexlab{a}.
\newblock \showarticletitle{Scalable Algorithms for Laplacian Pseudo-inverse Computation}.
\newblock \bibinfo{journal}{\emph{arXiv preprint arXiv:2311.10290}} (\bibinfo{year}{2023}).
\newblock


\bibitem[\protect\citeauthoryear{Liao, Li, Dai, and Wang}{Liao et~al\mbox{.}}{2022}]%
        {ppr_tree1_liao2022efficient}
\bibfield{author}{\bibinfo{person}{Meihao Liao}, \bibinfo{person}{Rong-Hua Li}, \bibinfo{person}{Qiangqiang Dai}, {and} \bibinfo{person}{Guoren Wang}.} \bibinfo{year}{2022}\natexlab{}.
\newblock \showarticletitle{Efficient personalized pagerank computation: A spanning forests sampling based approach}. In \bibinfo{booktitle}{\emph{SIGMOD}}. \bibinfo{pages}{2048--2061}.
\newblock


\bibitem[\protect\citeauthoryear{Liao, Zhou, Li, Dai, Chen, and Wang}{Liao et~al\mbox{.}}{2024}]%
        {er_liao2024efficient}
\bibfield{author}{\bibinfo{person}{Meihao Liao}, \bibinfo{person}{Junjie Zhou}, \bibinfo{person}{Rong-Hua Li}, \bibinfo{person}{Qiangqiang Dai}, \bibinfo{person}{Hongyang Chen}, {and} \bibinfo{person}{Guoren Wang}.} \bibinfo{year}{2024}\natexlab{}.
\newblock \showarticletitle{Efficient and Provable Effective Resistance Computation on Large Graphs: An Index-based Approach}.
\newblock \bibinfo{journal}{\emph{SIGMOD}} \bibinfo{volume}{2}, \bibinfo{number}{3} (\bibinfo{year}{2024}), \bibinfo{pages}{1--27}.
\newblock


\bibitem[\protect\citeauthoryear{Liu and Luo}{Liu and Luo}{2024}]%
        {vldb_ppr_liu2024bird}
\bibfield{author}{\bibinfo{person}{Haoyu Liu} {and} \bibinfo{person}{Siqiang Luo}.} \bibinfo{year}{2024}\natexlab{}.
\newblock \showarticletitle{BIRD: Efficient Approximation of Bidirectional Hidden Personalized PageRank}.
\newblock \bibinfo{journal}{\emph{VLDB}} \bibinfo{volume}{17}, \bibinfo{number}{9} (\bibinfo{year}{2024}), \bibinfo{pages}{2255--2268}.
\newblock


\bibitem[\protect\citeauthoryear{Lov{\'a}sz}{Lov{\'a}sz}{1993}]%
        {lovasz1993random}
\bibfield{author}{\bibinfo{person}{L{\'a}szl{\'o} Lov{\'a}sz}.} \bibinfo{year}{1993}\natexlab{}.
\newblock \showarticletitle{Random walks on graphs}.
\newblock \bibinfo{journal}{\emph{Combinatorics, Paul erdos is eighty}} \bibinfo{volume}{2}, \bibinfo{number}{1-46} (\bibinfo{year}{1993}), \bibinfo{pages}{4}.
\newblock


\bibitem[\protect\citeauthoryear{Martino, Morado, Li, Lu, and Rosta}{Martino et~al\mbox{.}}{2024}]%
        {app_clustering}
\bibfield{author}{\bibinfo{person}{Sam~Alexander Martino}, \bibinfo{person}{Jo{\~a}o Morado}, \bibinfo{person}{Chenghao Li}, \bibinfo{person}{Zhenghao Lu}, {and} \bibinfo{person}{Edina Rosta}.} \bibinfo{year}{2024}\natexlab{}.
\newblock \showarticletitle{Kemeny Constant-Based Optimization of Network Clustering Using Graph Neural Networks}.
\newblock \bibinfo{journal}{\emph{The Journal of Physical Chemistry B}} \bibinfo{volume}{128}, \bibinfo{number}{34} (\bibinfo{year}{2024}), \bibinfo{pages}{8103--8115}.
\newblock


\bibitem[\protect\citeauthoryear{Mo and Luo}{Mo and Luo}{2021}]%
        {pr_rw3_mo2021agenda}
\bibfield{author}{\bibinfo{person}{Dingheng Mo} {and} \bibinfo{person}{Siqiang Luo}.} \bibinfo{year}{2021}\natexlab{}.
\newblock \showarticletitle{Agenda: Robust personalized pageranks in evolving graphs}. In \bibinfo{booktitle}{\emph{CIKM}}. \bibinfo{pages}{1315--1324}.
\newblock


\bibitem[\protect\citeauthoryear{Ohsaka, Maehara, and Kawarabayashi}{Ohsaka et~al\mbox{.}}{2015}]%
        {pr_push3_ohsaka2015efficient}
\bibfield{author}{\bibinfo{person}{Naoto Ohsaka}, \bibinfo{person}{Takanori Maehara}, {and} \bibinfo{person}{Ken-ichi Kawarabayashi}.} \bibinfo{year}{2015}\natexlab{}.
\newblock \showarticletitle{Efficient pagerank tracking in evolving networks}. In \bibinfo{booktitle}{\emph{KDD}}. \bibinfo{pages}{875--884}.
\newblock


\bibitem[\protect\citeauthoryear{Patel, Agharkar, and Bullo}{Patel et~al\mbox{.}}{2015}]%
        {patel2015robotic}
\bibfield{author}{\bibinfo{person}{Rushabh Patel}, \bibinfo{person}{Pushkarini Agharkar}, {and} \bibinfo{person}{Francesco Bullo}.} \bibinfo{year}{2015}\natexlab{}.
\newblock \showarticletitle{Robotic surveillance and Markov chains with minimal weighted Kemeny constant}.
\newblock \bibinfo{journal}{\emph{IEEE Trans. Automat. Control}} \bibinfo{volume}{60}, \bibinfo{number}{12} (\bibinfo{year}{2015}), \bibinfo{pages}{3156--3167}.
\newblock


\bibitem[\protect\citeauthoryear{Paudel and Bernstein}{Paudel and Bernstein}{2021}]%
        {rw_recommendation_paudel2021random}
\bibfield{author}{\bibinfo{person}{Bibek Paudel} {and} \bibinfo{person}{Abraham Bernstein}.} \bibinfo{year}{2021}\natexlab{}.
\newblock \showarticletitle{Random walks with erasure: Diversifying personalized recommendations on social and information networks}. In \bibinfo{booktitle}{\emph{WWW}}. \bibinfo{pages}{2046--2057}.
\newblock


\bibitem[\protect\citeauthoryear{Peng, Lopatta, Yoshida, and Goranci}{Peng et~al\mbox{.}}{2021}]%
        {er2_peng2021local}
\bibfield{author}{\bibinfo{person}{Pan Peng}, \bibinfo{person}{Daniel Lopatta}, \bibinfo{person}{Yuichi Yoshida}, {and} \bibinfo{person}{Gramoz Goranci}.} \bibinfo{year}{2021}\natexlab{}.
\newblock \showarticletitle{Local algorithms for estimating effective resistance}. In \bibinfo{booktitle}{\emph{KDD}}. \bibinfo{pages}{1329--1338}.
\newblock


\bibitem[\protect\citeauthoryear{Perozzi, Al-Rfou, and Skiena}{Perozzi et~al\mbox{.}}{2014}]%
        {rw_embedding1_perozzi2014deepwalk}
\bibfield{author}{\bibinfo{person}{Bryan Perozzi}, \bibinfo{person}{Rami Al-Rfou}, {and} \bibinfo{person}{Steven Skiena}.} \bibinfo{year}{2014}\natexlab{}.
\newblock \showarticletitle{Deepwalk: Online learning of social representations}. In \bibinfo{booktitle}{\emph{KDD}}. \bibinfo{pages}{701--710}.
\newblock


\bibitem[\protect\citeauthoryear{Predari, Berner, Kooij, and Meyerhenke}{Predari et~al\mbox{.}}{2023}]%
        {wilson3_predari2023greedy}
\bibfield{author}{\bibinfo{person}{Maria Predari}, \bibinfo{person}{Lukas Berner}, \bibinfo{person}{Robert Kooij}, {and} \bibinfo{person}{Henning Meyerhenke}.} \bibinfo{year}{2023}\natexlab{}.
\newblock \showarticletitle{Greedy optimization of resistance-based graph robustness with global and local edge insertions}.
\newblock \bibinfo{journal}{\emph{Soc. Netw. Anal. Min.}} \bibinfo{volume}{13}, \bibinfo{number}{1} (\bibinfo{year}{2023}), \bibinfo{pages}{130}.
\newblock


\bibitem[\protect\citeauthoryear{Rossi and Ahmed}{Rossi and Ahmed}{2015}]%
        {nr}
\bibfield{author}{\bibinfo{person}{Ryan~A. Rossi} {and} \bibinfo{person}{Nesreen~K. Ahmed}.} \bibinfo{year}{2015}\natexlab{}.
\newblock \showarticletitle{The Network Data Repository with Interactive Graph Analytics and Visualization}. In \bibinfo{booktitle}{\emph{AAAI}}.
\newblock
\urldef\tempurl%
\url{https://networkrepository.com}
\showURL{%
\tempurl}


\bibitem[\protect\citeauthoryear{Sarkar and Moore}{Sarkar and Moore}{2011}]%
        {rw_analysis_sarkar2011random}
\bibfield{author}{\bibinfo{person}{Purnamrita Sarkar} {and} \bibinfo{person}{Andrew~W Moore}.} \bibinfo{year}{2011}\natexlab{}.
\newblock \showarticletitle{Random walks in social networks and their applications: a survey}.
\newblock \bibinfo{journal}{\emph{Social Network Data Analytics}} (\bibinfo{year}{2011}), \bibinfo{pages}{43--77}.
\newblock


\bibitem[\protect\citeauthoryear{Staudt, Sazonovs, and Meyerhenke}{Staudt et~al\mbox{.}}{2015}]%
        {networkit}
\bibfield{author}{\bibinfo{person}{Christian~L. Staudt}, \bibinfo{person}{Aleksejs Sazonovs}, {and} \bibinfo{person}{Henning Meyerhenke}.} \bibinfo{year}{2015}\natexlab{}.
\newblock \bibinfo{title}{NetworKit: A Tool Suite for Large-scale Complex Network Analysis}.
\newblock
\newblock
\showeprint[arxiv]{1403.3005}~[cs.SI]
\urldef\tempurl%
\url{https://arxiv.org/abs/1403.3005}
\showURL{%
\tempurl}


\bibitem[\protect\citeauthoryear{Tetali}{Tetali}{1991}]%
        {res1_tetali1991random}
\bibfield{author}{\bibinfo{person}{Prasad Tetali}.} \bibinfo{year}{1991}\natexlab{}.
\newblock \showarticletitle{Random walks and the effective resistance of networks}.
\newblock \bibinfo{journal}{\emph{Journal of Theoretical Probability}} \bibinfo{volume}{4}, \bibinfo{number}{1} (\bibinfo{year}{1991}), \bibinfo{pages}{101--109}.
\newblock


\bibitem[\protect\citeauthoryear{Wang, Wei, Gan, Yuan, Du, and Wen}{Wang et~al\mbox{.}}{2022}]%
        {vldb_ppr_wang2022edge}
\bibfield{author}{\bibinfo{person}{Hanzhi Wang}, \bibinfo{person}{Zhewei Wei}, \bibinfo{person}{Junhao Gan}, \bibinfo{person}{Ye Yuan}, \bibinfo{person}{Xiaoyong Du}, {and} \bibinfo{person}{Ji-Rong Wen}.} \bibinfo{year}{2022}\natexlab{}.
\newblock \showarticletitle{Edge-based local push for personalized PageRank}.
\newblock \bibinfo{journal}{\emph{VLDB}} \bibinfo{volume}{15}, \bibinfo{number}{7} (\bibinfo{year}{2022}), \bibinfo{pages}{1376--1389}.
\newblock


\bibitem[\protect\citeauthoryear{White and Smyth}{White and Smyth}{2003}]%
        {app_markovcentrality}
\bibfield{author}{\bibinfo{person}{Scott White} {and} \bibinfo{person}{Padhraic Smyth}.} \bibinfo{year}{2003}\natexlab{}.
\newblock \showarticletitle{Algorithms for estimating relative importance in networks}. In \bibinfo{booktitle}{\emph{KDD}}. \bibinfo{pages}{266--275}.
\newblock


\bibitem[\protect\citeauthoryear{Wilson}{Wilson}{1996}]%
        {wilson1996generating}
\bibfield{author}{\bibinfo{person}{David~Bruce Wilson}.} \bibinfo{year}{1996}\natexlab{}.
\newblock \showarticletitle{Generating random spanning trees more quickly than the cover time}. In \bibinfo{booktitle}{\emph{STOC}}. \bibinfo{pages}{296--303}.
\newblock


\bibitem[\protect\citeauthoryear{Xia, Liu, Lee, and Cao}{Xia et~al\mbox{.}}{2016}]%
        {rw_recommendation2_xia2016scientific}
\bibfield{author}{\bibinfo{person}{Feng Xia}, \bibinfo{person}{Haifeng Liu}, \bibinfo{person}{Ivan Lee}, {and} \bibinfo{person}{Longbing Cao}.} \bibinfo{year}{2016}\natexlab{}.
\newblock \showarticletitle{Scientific article recommendation: Exploiting common author relations and historical preferences}.
\newblock \bibinfo{journal}{\emph{IEEE Trans. Big Data}} \bibinfo{volume}{2}, \bibinfo{number}{2} (\bibinfo{year}{2016}), \bibinfo{pages}{101--112}.
\newblock


\bibitem[\protect\citeauthoryear{Xia, Liu, Nie, Fu, Wan, and Kong}{Xia et~al\mbox{.}}{2019}]%
        {rw_analysis2_xia2019random}
\bibfield{author}{\bibinfo{person}{Feng Xia}, \bibinfo{person}{Jiaying Liu}, \bibinfo{person}{Hansong Nie}, \bibinfo{person}{Yonghao Fu}, \bibinfo{person}{Liangtian Wan}, {and} \bibinfo{person}{Xiangjie Kong}.} \bibinfo{year}{2019}\natexlab{}.
\newblock \showarticletitle{Random walks: A review of algorithms and applications}.
\newblock \bibinfo{journal}{\emph{IEEE Transactions on Emerging Topics in Computational Intelligence}} \bibinfo{volume}{4}, \bibinfo{number}{2} (\bibinfo{year}{2019}), \bibinfo{pages}{95--107}.
\newblock


\bibitem[\protect\citeauthoryear{Xia and Zhang}{Xia and Zhang}{2024}]%
        {xia2024efficient}
\bibfield{author}{\bibinfo{person}{Haisong Xia} {and} \bibinfo{person}{Zhongzhi Zhang}.} \bibinfo{year}{2024}\natexlab{}.
\newblock \showarticletitle{Efficient Approximation of Kemeny's Constant for Large Graphs}.
\newblock \bibinfo{journal}{\emph{SIGMOD}} \bibinfo{volume}{2}, \bibinfo{number}{3} (\bibinfo{year}{2024}), \bibinfo{pages}{1--26}.
\newblock


\bibitem[\protect\citeauthoryear{Xu, Sheng, Zhang, Kan, and Zhang}{Xu et~al\mbox{.}}{2020}]%
        {xu2020power}
\bibfield{author}{\bibinfo{person}{Wanyue Xu}, \bibinfo{person}{Yibin Sheng}, \bibinfo{person}{Zuobai Zhang}, \bibinfo{person}{Haibin Kan}, {and} \bibinfo{person}{Zhongzhi Zhang}.} \bibinfo{year}{2020}\natexlab{}.
\newblock \showarticletitle{Power-law graphs have minimal scaling of Kemeny constant for random walks}. In \bibinfo{booktitle}{\emph{WWW}}. \bibinfo{pages}{46--56}.
\newblock


\bibitem[\protect\citeauthoryear{Yang, Wang, Wei, Wang, and Wen}{Yang et~al\mbox{.}}{2024}]%
        {ppr_dy_yang2024efficient}
\bibfield{author}{\bibinfo{person}{Mingji Yang}, \bibinfo{person}{Hanzhi Wang}, \bibinfo{person}{Zhewei Wei}, \bibinfo{person}{Sibo Wang}, {and} \bibinfo{person}{Ji-Rong Wen}.} \bibinfo{year}{2024}\natexlab{}.
\newblock \showarticletitle{Efficient algorithms for personalized pagerank computation: A survey}.
\newblock \bibinfo{journal}{\emph{IEEE TKDE}} (\bibinfo{year}{2024}).
\newblock


\bibitem[\protect\citeauthoryear{Yang, Shi, Xiao, Yang, and Bhowmick}{Yang et~al\mbox{.}}{2020}]%
        {vldb_ppr_embedding_yang13homogeneous}
\bibfield{author}{\bibinfo{person}{Renchi Yang}, \bibinfo{person}{Jieming Shi}, \bibinfo{person}{Xiaokui Xiao}, \bibinfo{person}{Yin Yang}, {and} \bibinfo{person}{Sourav~S Bhowmick}.} \bibinfo{year}{2020}\natexlab{}.
\newblock \showarticletitle{Homogeneous Network Embedding for Massive Graphs via Reweighted Personalized PageRank}.
\newblock \bibinfo{journal}{\emph{VLDB}} \bibinfo{volume}{13}, \bibinfo{number}{5} (\bibinfo{year}{2020}), \bibinfo{pages}{670--683}.
\newblock


\bibitem[\protect\citeauthoryear{Yang and Tang}{Yang and Tang}{2023}]%
        {er3_yang2023efficient}
\bibfield{author}{\bibinfo{person}{Renchi Yang} {and} \bibinfo{person}{Jing Tang}.} \bibinfo{year}{2023}\natexlab{}.
\newblock \showarticletitle{Efficient estimation of pairwise effective resistance}.
\newblock \bibinfo{journal}{\emph{Proceedings of the ACM on Management of Data}} \bibinfo{volume}{1}, \bibinfo{number}{1} (\bibinfo{year}{2023}), \bibinfo{pages}{1--27}.
\newblock


\bibitem[\protect\citeauthoryear{Yilmaz, Dudkina, Bin, Crisostomi, Ferraro, Murray-Smith, Parisini, Stone, and Shorten}{Yilmaz et~al\mbox{.}}{2020}]%
        {app_covid19}
\bibfield{author}{\bibinfo{person}{Serife Yilmaz}, \bibinfo{person}{Ekaterina Dudkina}, \bibinfo{person}{Michelangelo Bin}, \bibinfo{person}{Emanuele Crisostomi}, \bibinfo{person}{Pietro Ferraro}, \bibinfo{person}{Roderick Murray-Smith}, \bibinfo{person}{Thomas Parisini}, \bibinfo{person}{Lewi Stone}, {and} \bibinfo{person}{Robert Shorten}.} \bibinfo{year}{2020}\natexlab{}.
\newblock \showarticletitle{Kemeny-based testing for COVID-19}.
\newblock \bibinfo{journal}{\emph{PLOS ONE}} \bibinfo{volume}{15}, \bibinfo{number}{11} (\bibinfo{year}{2020}), \bibinfo{pages}{1--19}.
\newblock


\bibitem[\protect\citeauthoryear{Yin, Cui, Li, Yao, and Chen}{Yin et~al\mbox{.}}{2012}]%
        {vldb_app_ht_reccomandation_yin2012challenging}
\bibfield{author}{\bibinfo{person}{Hongzhi Yin}, \bibinfo{person}{Bin Cui}, \bibinfo{person}{Jing Li}, \bibinfo{person}{Junjie Yao}, {and} \bibinfo{person}{Chen Chen}.} \bibinfo{year}{2012}\natexlab{}.
\newblock \showarticletitle{Challenging the Long Tail Recommendation}.
\newblock \bibinfo{journal}{\emph{VLDB}} \bibinfo{volume}{5}, \bibinfo{number}{9} (\bibinfo{year}{2012}).
\newblock


\bibitem[\protect\citeauthoryear{Yoon, Jin, and Kang}{Yoon et~al\mbox{.}}{2018}]%
        {pr_pi2_yoon2018fast}
\bibfield{author}{\bibinfo{person}{Minji Yoon}, \bibinfo{person}{Woojeong Jin}, {and} \bibinfo{person}{U Kang}.} \bibinfo{year}{2018}\natexlab{}.
\newblock \showarticletitle{Fast and accurate random walk with restart on dynamic graphs with guarantees}. In \bibinfo{booktitle}{\emph{WWW}}. \bibinfo{pages}{409--418}.
\newblock


\bibitem[\protect\citeauthoryear{Zhang, Lofgren, and Goel}{Zhang et~al\mbox{.}}{2016}]%
        {pr_push2_zhang2016approximate}
\bibfield{author}{\bibinfo{person}{Hongyang Zhang}, \bibinfo{person}{Peter Lofgren}, {and} \bibinfo{person}{Ashish Goel}.} \bibinfo{year}{2016}\natexlab{}.
\newblock \showarticletitle{Approximate personalized pagerank on dynamic graphs}. In \bibinfo{booktitle}{\emph{KDD}}. \bibinfo{pages}{1315--1324}.
\newblock


\bibitem[\protect\citeauthoryear{Zhao, Chang, Sarma, Zheng, and Zhao}{Zhao et~al\mbox{.}}{2013}]%
        {vldb_rw_embedding_zhao2013embeddability}
\bibfield{author}{\bibinfo{person}{Xiaohan Zhao}, \bibinfo{person}{Adelbert Chang}, \bibinfo{person}{Atish~Das Sarma}, \bibinfo{person}{Haitao Zheng}, {and} \bibinfo{person}{Ben~Y Zhao}.} \bibinfo{year}{2013}\natexlab{}.
\newblock \showarticletitle{On the embeddability of random walk distances}.
\newblock \bibinfo{journal}{\emph{VLDB}} \bibinfo{volume}{6}, \bibinfo{number}{14} (\bibinfo{year}{2013}), \bibinfo{pages}{1690--1701}.
\newblock


\bibitem[\protect\citeauthoryear{Zheng, Wang, Wei, Liu, and Wang}{Zheng et~al\mbox{.}}{2022}]%
        {pr_push1_zheng2022instant}
\bibfield{author}{\bibinfo{person}{Yanping Zheng}, \bibinfo{person}{Hanzhi Wang}, \bibinfo{person}{Zhewei Wei}, \bibinfo{person}{Jiajun Liu}, {and} \bibinfo{person}{Sibo Wang}.} \bibinfo{year}{2022}\natexlab{}.
\newblock \showarticletitle{Instant graph neural networks for dynamic graphs}. In \bibinfo{booktitle}{\emph{KDD}}. \bibinfo{pages}{2605--2615}.
\newblock


\end{thebibliography}

\end{document}